\documentclass{amsart}

\usepackage{amsfonts, amssymb, amsmath, amsthm}
\usepackage{slashed}
\usepackage{graphicx}
\usepackage{color}
\usepackage{enumerate}                                         
\usepackage{tikz}                                              
\usetikzlibrary{decorations.pathmorphing, arrows}

\newtheorem{theorem}{Theorem}
\newtheorem{corollary}[theorem]{Corollary}
\newtheorem{lemma}[theorem]{Lemma}
\newtheorem{proposition}[theorem]{Proposition}
\newtheorem{definition}[theorem]{Definition}

\newtheorem*{remark}{Remark}

\numberwithin{equation}{section}
\numberwithin{theorem}{section}

\newcommand{\ul}[1]{\underline{\smash{#1}}}                    
\newcommand{\mf}[1]{\mathfrak{#1}}                             
\newcommand{\mc}[1]{\mathcal{#1}}                              

\newcommand{\R}{\mathbb{R}}                                    
\newcommand{\Sph}{\mathbb{S}}                                  

\newcommand{\grad}{\nabla^\sharp}                              
\newcommand{\nasla}{\slashed{\nabla}}                          

\newcommand{\ggm}{\mathring{g}}                                
\newcommand{\ggs}{\bar{g}}                                     
\newcommand{\ggr}{E}                                           
\newcommand{\gm}{\mathring{\mf{g}}}                            
\newcommand{\gs}{\bar{\mf{g}}}                                 
\newcommand{\gb}{\hat{\mf{g}}}                                 
\newcommand{\ga}{\tilde{\gamma}}                               

\newcommand{\fbd}{\eta}                                        

\newcommand{\fb}{\hat{f}}                                      
\newcommand{\fbdb}{\hat{\eta}}                                 
\newcommand{\Nb}{\hat{N}}                                      
\newcommand{\Vb}{\hat{V}}                                      
\newcommand{\Eb}{\hat{E}}                                      
\newcommand{\Sb}{\hat{S}}                                      
\newcommand{\wb}{\hat{w}}                                      
\newcommand{\hb}{\hat{h}}                                      
\newcommand{\dftb}{\hat{\pi}}                                  

\newcommand{\urho}{\ul{\rho}}                                  
\newcommand{\ut}{\ul{t}}                                       
\newcommand{\ux}{\ul{x}}                                       
\newcommand{\uggm}{\ul{\ggm}}                                  
\newcommand{\uggs}{\ul{\ggs}}                                  
\newcommand{\uO}{\ul{\mc{O}}}                                  
\newcommand{\uaff}{\ul{\sigma}}                                

\newcommand{\isch}{\mathring{P}}                               

\allowdisplaybreaks
\begin{document}

\title[Unique continuation]{Unique continuation from infinity in\\
asymptotically Anti-de Sitter spacetimes II:\\
Non-Static Boundaries}

\author{Gustav Holzegel}
\address{Department of Mathematics\\
South Kensington Campus\\
Imperial College\\
London SW7 2AZ\\ United Kingdom}
\email{gholzege@imperial.ac.uk}

\author{Arick Shao}
\address{School of Mathematical Sciences\\
Mile End Campus\\
Queen Mary University of London\\
London E1 4NS\\ United Kingdom}
\email{a.shao@qmul.ac.uk}

\begin{abstract}
We generalize our unique continuation results recently established for a class of linear and nonlinear wave equations $\Box_g \phi + \sigma \phi = \mc{G} ( \phi, \partial \phi )$ on asymptotically anti-de Sitter (aAdS) spacetimes to aAdS spacetimes admitting non-static boundary metrics.
The new Carleman estimates established in this setting constitute an essential ingredient in proving unique continuation results for the full nonlinear Einstein equations, which will be addressed in forthcoming papers.
Key to the proof is a new geometrically adapted construction of foliations of pseudoconvex hypersurfaces near the conformal boundary. 
\end{abstract}

\maketitle

\section{Introduction} \label{sec.intro}

In \cite{hol_shao:uc_ads}, we initiated the study of unique continuation properties of $(n+1)$-dimensional asymptotically anti-de Sitter (aAdS) spacetimes $\left(\mathcal{M},g\right)$ by studying a class of tensorial linear and non-linear Klein-Gordon equations,
\begin{align}
\label{wea} \Box_g \phi + \sigma \phi = \mathcal{G}\left(\phi, \partial \phi\right) \text{,}
\end{align}
in a portion of spacetime near the conformal boundary $\mc{I}$, with $\sigma \in \mathbb{R}$ and suitable assumptions on $\mathcal{G}\left(\phi, \partial\phi\right)$.

The spacetimes $\left(\mathcal{M},g\right)$ considered in equation (\ref{wea}) encompassed a large class of Lorentzian metrics, not necessarily Einstein-vacuum, including in particular non-stationary spacetimes.
The main restriction in \cite{hol_shao:uc_ads} was the assumption that the $n$-dimensional Lorentzian metric $\gm$ induced by $g$ on the boundary $\mathcal{I}$ (after conformal transformation) was \emph{static}.
In this paper, we will remove this assumption and extend the unique continuation results of \cite{hol_shao:uc_ads} to a class of metrics which are \emph{not} required to be static on the boundary.

Our main motivation originates from general relativity, where spacetimes with non-static boundary metrics appear naturally by solving an initial boundary value problem for the vacuum Einstein equations\footnote{For convenience, here and in the remainder of the paper, we normalize the Einstein-vacuum equations $\operatorname{Ric}_g - \frac{1}{2} \operatorname{Sc}_g g + \Lambda g = 0$ by setting the cosmological constant to be $\Lambda = - \frac{ n (n - 1) }{2}$.}
\begin{align}
\label{ee} \operatorname{Ric}_g = - n g \text{.}
\end{align}
Indeed, in dimension $3+1$, Friedrich \cite{Friedrich} constructed a large class of aAdS spacetimes satisfying (\ref{ee}) for which the conformal class of the $n$-dimensional metric on the boundary can be freely prescribed a priori. A particularly interesting case arises from so-called dissipative boundary conditions. Here the resulting spacetime will generally not only possess a non-static boundary metric but also exhibit a non-trivial flux of gravitational radiation through its boundary. See also \cite{adsdissipative}.

In view of the above, extending the Carleman estimates of \cite{hol_shao:uc_ads} to spacetimes with general (dynamical) boundary metrics, such as Friedrich's spacetimes, ensures that the class of metrics satisfying (\ref{ee}) for which a unique continuation
property holds agrees with the class of metrics that arises naturally from the forward initial boundary value problem for (\ref{ee}). This is a prerequisite for proving unique continuation results for the non-linear Einstein equations in a sensible class.  
We finally remark that the class of spacetimes considered here is also natural in the context of the AdS/CFT correspondence \cite{Malda}.

\subsection{The Class of aAdS Spacetimes}

We first give an informal definition of the class of aAdS spacetimes to be considered. Unlike in \cite{hol_shao:uc_ads}, we will in this paper exhibit these spacetimes in Fefferman--Graham (FG) coordinate systems near the boundary.
Such coordinates are well-adapted to the geometric problem at hand and simplify many of the computations.
We note that this does not constitute any loss of generality; in Appendix \ref{sec.fg}, we demonstrate how to transform a metric in the coordinates used in \cite{hol_shao:uc_ads} to a metric in FG form.
  
Specifically, we consider manifolds $\mc{M} = ( 0, \rho_\ast ) \times ( T_-, T_+ ) \times \mc{S}$, with $\mc{S}$ an $(n-1)$-dimensional Riemannian manifold, and with the local coordinates of $\mc{S}$ denoted collectively by $x$.
We will equip $\mc{M}$ with metrics of the form
\begin{align}
g = \frac{d\rho^2 + \mf{g} ( \rho ) }{\rho^2} \text{,}
\end{align}
where $\mf{g}$ is a family of Lorentzian metrics on the level sets of $\rho$ with the expansion
\begin{align} \label{expand}
\mf{g} ( \rho, t, x ) = \gm ( t, x ) + \rho^2 \gs (t, x ) + \rho^3 \gb ( t, x ) + { \scriptstyle \mathcal{O}} ( \rho^3 ) \text{.}
\end{align}
Here $\gm$, $\gs$, and $\gb$ are tensors defined on the level sets of $\rho$ whose components are independent of the particular level set chosen.\footnote{Theorems \ref{theo:old} and \ref{theo:mti} below only use (\ref{expand}) with the last two terms replaced by $\mathcal{O} ( \rho^3 )$.}

As is well-known \cite{MR837196, grah_witt:conf_ads}, if $g$ satisfies \eqref{ee} and $n\geq 3$, then $- \gs$ coincides with the Schouten tensor $\isch$ of $\gm$; see also Appendix \ref{sec.vacuum}.
Furthermore, in dimension $n > 3$, the tensor $\gb$ is also locally determined by $\gm$, while for $n=3$, $\gb$ is the ``stress-energy tensor" on the boundary, which is not determined by $\gm$ but by the full spacetime Weyl tensor through the formula $\gb_{ab}=\frac{1}{3}\partial_\rho \left( \rho^2 W_{a \rho b \rho}\right)|_{\rho=0}$, see \cite{Warnick, Anderson}.
For example, for Schwarzschild-AdS spacetimes in $n=3$, one computes
\[
\gm = -dt^2 + \ga \text{,} \qquad \gs = -\frac{1}{2} ( dt^2 + \ga ) \text{,} \qquad \gb = \frac{2}{3} M ( 2 dt^2 + \ga ) \text{,}
\]
where $\ga$ is the round metric on the unit sphere.
For pure AdS spacetime, the above identities hold, but with $\gb \equiv 0$. 

The metrics of interest in this paper are precisely those allowing for $\mc{L}_{\partial_t} \gm \neq 0$, while in \cite{hol_shao:uc_ads}, we assumed $\mc{L}_{\partial_t} \gm = 0$.
An important special case of this is Einstein-vacuum metrics that are small perturbations of stationary aAdS spacetimes; in this case, $\mc{L}_{\partial_t} \gm$ is expected to be small in a suitable norm; see \cite{Enciso} for examples. Another interesting (explicit) example, which has received considerable attention in the high energy physics community, is given by the Robinson-Trautman-AdS metrics \cite{Bakas, Freitas}.

\subsection{Previous results}

We turn to the unique continuation results for (\ref{wea}) on segments $\left(\mathcal{M},g\right)$ defined above.
We first recall from  \cite{hol_shao:uc_ads} the quantities
\begin{align}
\label{eq.beta} \beta_\pm = \frac{n}{2} \pm \sqrt{\frac{n^2}{4}-\sigma}
\end{align}
associated with the mass $\sigma$ in (\ref{wea}).
Precise assumptions on the right hand side $\mathcal{G}$ in (\ref{wea}) will be made below, e.g., in \eqref{Gcond} and \eqref{Gcond2}.\footnote{From the point of view of decay near the boundary, these assumptions will allow us to treat $\mathcal{G}$ as a perturbation of the linear Klein-Gordon operator on the left hand side.}

We next define the \emph{local unique continuation property} of the spacetime $( \mc{M}, g )$ that can be established for solutions to (\ref{wea}).
While the full definition is slightly technical, see Definition \ref{def.ucp}, it essentially states that $( \mc{M}, g )$ satisfies the \emph{local unique continuation property} if any classical solution $\phi$ of \eqref{wea} which satisfies\footnote{In the case that the cross-sections $\mc{S}$ are not compact, we also assume $\phi$ to have compact support on each copy of $\mc{S}$; see Definition \ref{def.ucp} for details.}
\begin{equation} \label{condition}
\left\{
 \begin{array}{rl}
      |\rho^{-\beta_+} \phi| + |\nabla_{t,\rho,x} (\rho^{-\beta_++1} \phi)| \rightarrow 0 & \text{if } \sigma \leq \frac{n^2-1}{4},\\
       |\rho^{-(n+1)/2} \phi| + |\nabla_{t,\rho,x} (\rho^{-(n-1)/2} \phi)| \rightarrow 0 & \text{if } \sigma > \frac{n^2-1}{4}
    \end{array} \right.
\end{equation}
on the conformal boundary $\mathcal{I}=\{\rho=0\} \times \left(T_-,T_+\right) \times \mathcal{S}$ of $( \mc{M},g )$, vanishes in an open neighborhood of $\mc{I}$.\footnote{The proper Definition \ref{def.ucp} stipulates slightly weaker vanishing conditions in $L^2$. Furthermore, the vanishing condition for $\nabla_x \phi$ in \eqref{condition} is in fact not necessary and is replaced in Definition \ref{def.ucp} by a weaker finite integral condition (which is a consequence of finite energy).}

An important observation from \cite{hol_shao:uc_ads}, which demonstrates that the vanishing conditions (\ref{condition}) are somewhat natural, is that if $\frac{n^2}{4}-1 < \sigma < \frac{n^2}{4}$ in \eqref{wea}, then any classical solution $\phi$ of \eqref{wea} in $( \mc{M}, g )$ which satisfies \emph{both} Dirichlet \emph{and} Neumann boundary conditions at $\mc{I}$ also satisfies \eqref{condition}.\footnote{Recall that only for the aforementioned $\sigma$-range does one have the freedom of specifying boundary conditions for the forward boundary initial value problem. See the discussion in  \cite{hol_shao:uc_ads}.}

The key uniqueness theorem of \cite{hol_shao:uc_ads} can now be rephrased in the Fefferman--Graham coordinates introduced above as:

\begin{theorem}[Theorem 1.3 and Theorem 4.2 of \cite{hol_shao:uc_ads}] \label{theo:old}
Let
\[
( \mc{M} := ( 0, \rho_\ast ) \times ( 0, T \pi ) \times \mc{S}, g )
\]
be an aAdS spacetime segment whose boundary data $\gm$ and $\gs$ satisfy that $\gm$ is static on $\mc{I}$ (i.e., $\mc{L}_{\partial_t} \gm = 0$), as well as the following pseudoconvexity condition:
\begin{equation}
\label{psc} -\bar{\mathfrak{g}} - \frac{1}{T^2} dt^2 - \zeta \mathring{\mathfrak{g}} \text{ positive-definite on $\mc{I}$ for some bounded $\zeta \in C^\infty ( \mc{I} )$.}
\end{equation}
Then, the local unique continuation property holds on $( \mc{M}, g )$ for \eqref{wea}, provided $\mathcal{G} (\phi, \partial \phi )$ satisfies the estimate
\begin{align}
\label{Gcond_old} | \mathcal{G} ( \phi, \partial \phi ) |^2 \leq C \rho^p ( \rho^4|\nabla_{t,\rho,x} \phi|^2 + \rho^{2p} |\phi|^2 )
\end{align}
in $(\mathcal{M}, g )$ for some $p > 0$ and some constant $C > 0$.
\end{theorem}

\begin{remark}
The vanishing assumptions on $\phi$ in \eqref{condition} are optimal when $\sigma \leq \frac{ n^2 - 1 }{4}$.
On the other hand, when $\sigma > \frac{ n^2 - 1 }{4}$, it is not known whether the vanishing condition \eqref{condition} is sharp.
See the introduction of \cite{hol_shao:uc_ads} for additional discussions.
\end{remark}

We remark that for $\left(\mathcal{M},g\right)$ being a segment of the exact AdS spacetime with cosmological constant $\Lambda = -\frac{n(n-1)}{2}$, any $T>1$ will guarantee that condition \eqref{psc} holds, while the condition does not hold if $T \leq 1$.
The borderline case $T = 1$ (i.e., the segment having time length $1 \cdot \pi$) corresponds precisely to the (re)focusing time of null geodesics emanating from the boundary; see Figure \ref{fig1}.
As explained in \cite{hol_shao:uc_ads}, in view of the counterexamples of \cite{alin_baou:non_unique}, this restriction $T \geq 1$ on the timespan is expected to be necessary in general; see also Theorem \ref{theo:mti2} below.

More generally, one sees that if $g$ in Theorem \ref{theo:old} is Einstein-vacuum, then a boundary metric $\gm$ with positive Schouten tensor $\isch = - \gs$ implies condition (\ref{psc}) for large enough $T$, the optimal $T$ being closely related to the refocusing time of null geodesics near the boundary.
See Appendix \ref{sec.vacuum} for quantitative statements.
\begin{figure}
\[
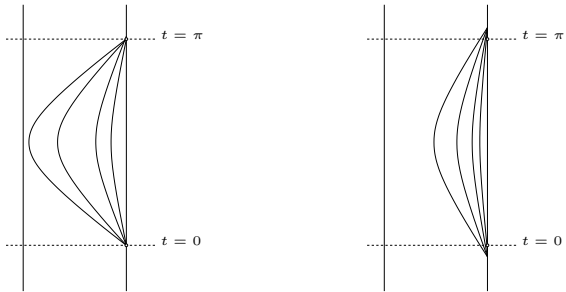
\]
\caption{Illustration of the refocusing of null geodesics (left) and the pseudoconvex foliation (right) in pure AdS.} \label{fig1}
\end{figure}

One of the key difficulties in proving Theorem \ref{theo:old} derives from the fact that the conformal boundary is \emph{zero-pseudoconvex}, that is, $\mc{I}$ is ruled by null geodesics.\footnote{Technically speaking, both pseudoconvexity and zero-pseudoconvexity are properties of hypersurfaces within an ambient spacetime. However, both notions are conformally invariant, hence it is sensible to apply them to the conformal boundary $\mc{I}$.}
As a result, standard unique continuation results fail there, hence one must consider much more carefully the spacetime geometry near $\mc{I}$.

The proof in \cite{hol_shao:uc_ads} constructed a foliation of the spacetime segment by pseudoconvex hypersurfaces near $\mc{I}$, whose existence in turn depended crucially on the time-length of the segment and eventually led to condition \eqref{psc}.
From the foliation (depicted schematically for pure AdS in Figure \ref{fig1}), we deduced after suitable renormalization a Carleman estimate which implied the unique continuation property stated in Theorem \ref{theo:old}.
We also emphasize that, with applications to general relativity in mind, we actually proved both the Carleman estimates and the uniqueness theorems for a class of \emph{tensorial} wave equations in \cite{hol_shao:uc_ads}.

\subsection{The main result}

We turn to the main result of this paper, which generalizes Theorem \ref{theo:old} by removing the staticity assumption $\mc{L}_{\partial_t} \gm = 0$ and replacing condition (\ref{psc}) appropriately.
As in \cite{hol_shao:uc_ads}, the main technical difficulty is to define a foliation of pseudoconvex timelike hypersurfaces near the boundary for this class of spacetimes.
This requires a new idea, because, as it turns out, even if the $\mc{L}_{\partial_t} \gm \neq 0$ is small, the extra terms arising from $\mc{L}_{\partial_t} \gm$ cannot be treated perturbatively: pseudoconvex hypersurfaces defined for a spacetime satisfying $\mc{L}_{\partial_t} \gm = 0$ will in general cease to be pseudoconvex if the metric is perturbed such that $\mc{L}_{\partial_t} \gm \neq 0$.

We resolve this problem by starting from a general ansatz for the level sets of the foliation, which eventually connects pseudoconvexity of the level sets to the existence of particular solutions to an ordinary differential inequality (ODI), whose coefficients depend on $\mc{L}_{\partial_t} \gm$.
Very schematically, in the boundary-static case, this ODI is simply a harmonic oscillator-type ODI.
In the dynamic case, one is instead led to a damped harmonic oscillator, with the damping term determined by $\mc{L}_{\partial_t} \gm$.\footnote{As in the boundary-static case, one can actually relate the resulting ODI to (approximate) null geodesics on these spacetimes. More precisely, the geodesic equation for the $\rho$-variable, expressed with respect to the time coordinate $t$, will be (approximately) a damped harmonic oscillator with frequency determined by $\bar{\mathfrak{g}}$ and damping determined by $\mc{L}_{\partial_t} \gm$.}

Once the foliation has been defined, the proof of the Carleman estimate and the uniqueness statements proceed as in \cite{hol_shao:uc_ads} apart from minor technical difficulties which will be discussed in the bulk of the paper.
This leads to the following rough version of our main result; see Section \ref{sec.proof} for the precise statement.

\begin{theorem} \label{theo:mti}
Let $( \mc{M} := (0, \rho_\ast) \times (0, T \pi) \times \mc{S}, g )$ be an aAdS spacetime segment whose boundary data\footnote{Note the specific form of $\gm$ here amounts to both an implicit geometric assumption and an implicit choice of gauge. For more details on this, see the remark below Definition \ref{def.aads_boundary}.} $\gm = -dt^2 + \gamma_{AB} \left(t,x\right) dx^Adx^B$ and $\bar{\mathfrak{g}}$ satisfy the following:
\begin{itemize}
\item $\xi > 0$ is such that for any vector field $Y$on $\mc{I}$ tangent to $\mc{S}$,
\begin{equation}
\label{psc2} | \mc{L}_{\partial_t} \gamma ( Y, Y ) | \leq \xi \cdot \gamma ( Y, Y ) \text{.}
\end{equation}

\item There is some bounded $\zeta \in C^\infty ( \mc{I} )$ such that
\begin{equation}
\label{psc3} -\bar{\mathfrak{g}} - \frac{1}{\tau^2} dt^2 - \zeta \mathring{\mathfrak{g}} \text{ is positive-definite on $\mc{I}$,}
\end{equation}
where $\tau > T$ is a constant that is defined in terms of $T$ and $\xi$.
\end{itemize}
Then the local unique continuation property holds on $( \mc{M}, g )$ for \eqref{wea}, provided $\mathcal{G}$ satisfies for some $p > 0$ and $C > 0$ the estimate
\begin{equation}
\label{Gcond} | \mathcal{G} ( \phi, \partial \phi ) |^2 \leq C \rho^p ( \rho^4 |\nabla_{t,\rho,x} \phi |^2 + \rho^{2p} |\phi|^2 ) \text{.}
\end{equation}
\end{theorem}

We remark that there is a simple explicit formula for $\tau$; see (\ref{eq.mu}) and (\ref{eq.Q_inf}).
One can also check that for small perturbations of the pure AdS boundary, i.e.~for $\mc{L}_{\partial_t} \gm$ small and $\bar{\mathfrak{g}}$, $\mathring{\mathfrak{g}}$ close to their pure AdS values, there is a $T>1$ close to $1$ (the closeness depending on the size of the perturbation) which ensures that conditions \eqref{psc2} and \eqref{psc3} are indeed satisfied; see Section \ref{sec:pseudo_examples}. 

\begin{remark}
We remark here that in both Theorems \ref{theo:old} and \ref{theo:mti}, we are always considering an aAdS spacetime with a fixed conformal compactification, i.e., $\gm$ is fixed.
On the other hand, there is a residual gauge freedom corresponding to a redefinition of the boundary defining function $\rho$ that keeps the metric in Fefferman--Graham form.
Under this transformation, $\gm$ changes by a conformal factor, and $\gs$ has the transformation properties of a Schouten tensor.
In particular, it is possible for \eqref{psc2}, \eqref{psc3} to hold with respect to one conformal compactification but not another.
However, once both \eqref{psc2} and \eqref{psc3} hold for a particular conformal compactification, one can find a foliation of pseudconvex hypersurfaces near the boundary, which is a gauge-independent geometric statement.
\end{remark}

\subsection{The borderline case}

The slightly more geometric approach taken in this paper also reveals an interesting new result for the special case of \emph{static} boundaries addressing the ``borderline case" $T=1$ discussed below Theorem \ref{theo:old}:

\begin{theorem} \label{theo:mti2}
Let $( \mc{M} = (0, \rho_\ast) \times (0, \pi) \times \mc{S}, g )$ be an aAdS spacetime segment, assume the boundary data satisfies
$\gm = -dt^2 + \ga$ and $\gs = -\frac{1}{2} ( dt^2 + \ga )$, where $\tilde{\gamma}$ denotes the round metric on the unit sphere, and suppose $\gb$ satisfies the following pseudoconvexity condition:
\begin{align}
\label{psc5} -\gb - \zeta \gm \text{ is positive-definite for some bounded $\zeta \in C^\infty ( \mc{I} )$.}
\end{align}
Then the unique continuation property holds on $( \mc{M}, g )$ for \eqref{wea}, provided $\mathcal{G}$ satisfies for some $p>0$ and $C > 0$ the estimate
\begin{equation}
\label{Gcond2} | \mc{G} ( \phi, \partial \phi ) |^2 \leq C \rho^p ( \rho^5 |\nabla_{t,\rho,x} \phi|^2 + \rho^{2p} |\phi|^2 ) \text{.}
\end{equation}
\end{theorem}

Note that unlike (\ref{psc}), the pseudoconvexity condition (\ref{psc5}) involves $\gb$ of the expansion (\ref{expand}). 
Observe also that the condition (\ref{Gcond2}) on $\mathcal{G}$ is more stringent than (\ref{Gcond}).
This is because the pseudoconvexity of the foliation degenerates faster toward the conformal boundary in Theorem \ref{theo:mti2} than in Theorem \ref{theo:mti}.
The proof of Theorem \ref{theo:mti2} logically proceeds along the same lines as the one of Theorem \ref{theo:mti} but requires a significant refinement of the foliation used in \cite{hol_shao:uc_ads} to exploit the higher-order pseudoconvexity.

Note that if $\gb$ is negative-definite, then the pseudoconvexity condition holds with $\zeta=0$.
In particular, a patch of the $(3+1)$-dimensional AdS-Schwarzschild spacetime with negative mass provides a simple example of a spacetime satisfying the assumptions of Theorem \ref{theo:mti2}.
This should be compared with the asymptotically flat case ($\Lambda=0$), where \emph{positive} mass ensures a foliation of pseudoconvex hypersurfaces near spacelike infinity; see  \cite{alex_schl_shao:uc_inf}.

\subsection{Overview}

We have written the paper to be essentially self-contained, although \cite{hol_shao:uc_ads} contains more extensive explanations of the basic concepts and computations. We also refer the reader to \cite{hol_shao:uc_ads}  for a more exhaustive list of references and some historical background on unique continuation problems in geometric settings.
In Section \ref{sec.aads}, we define the manifolds and their differentiable structures, as well as the class of aAdS metrics considered on these manifolds.
Appendix \ref{sec.fg} relates this class to the class considered in \cite{hol_shao:uc_ads}.
Section \ref{sec.proof} contains the precise statements and proofs of our main results: the pseudoconvexity criterion, the Carleman estimate, and the unique continuation result.
In Section \ref{sec.border}, we carry out the analogous process in the static borderline case relevant for Theorem \ref{theo:mti2} only.
Finally, in Appendix \ref{sec.vacuum}, we collect a geometric interpretation of the result and some applications in the case of Einstein-vacuum metrics with static boundary metrics.

\subsection{Acknowledgement}
The authors thank Claude Warnick for helpful discussions and sharing his notes \cite{Warnick} on the Fefferman-Graham expansions. Both authors acknowledge support through a grant from the European Research Council. Finally, we thank two anonymous referees for very useful comments leading to several improvements of the paper.

\section{Asymptotically AdS Spacetimes} \label{sec.aads}

In this section, we construct the class of asymptotically Anti-de Sitter spacetimes that we will treat in our main results.
Informally, we will consider spacetimes whose metrics near infinity are of the form
\begin{align}
\label{eq.aads_gen_pre} g &= [ r^{-2} + r^{-4} \ggs_{ \rho \rho } + \mc{O} ( r^{-5} ) ] d r^2 + \mc{O} ( r^{-3} ) \cdot d r d t + \mc{O} ( r^{-3} ) \cdot d r d x^A \\
\notag &\qquad + [ - r^2 + \ggs_{ t t } + \mc{O} ( r^{-1} ) ] dt^2 + [ \ggs_{t A} + \mc{O} ( r^{-1} ) ] d t d x^A \\
\notag &\qquad + [ r^2 \ggm_{A B} + \ggs_{A B} + \mc{O} ( r^{-1} ) ] d x^A d x^B \text{.}
\end{align}
These include AdS spacetime and, when $n \geq 3$, Schwarzschild-AdS and Kerr-AdS spacetimes \cite{Henneaux}.
A precise description of these spacetimes, in particular the specific coordinates and nature of the ``$\mc{O} ( r^k )$"-error terms, will be given below.

In contrast to \cite{hol_shao:uc_ads}, here we opt to express \eqref{eq.aads_gen_pre} in a Fefferman--Graham gauge.
Roughly, these are asymptotic expansions of the form
\begin{align}
\label{eq.aads} g &= \rho^{-2} d \rho^2 + [ - \rho^{-2} + \gs_{tt} + \mc{O} ( \rho ) ] d t^2 + [ \gs_{t A} + \mc{O} ( \rho ) ] d t d x^A \\
\notag &\qquad + [ \rho^{-2} \gm_{A B} + \gs_{A B} + \mc{O} ( \rho ) ] dx^A dx^B \text{.}
\end{align}
We stress that any expansion of the form \eqref{eq.aads_gen_pre} can be reduced to one of the form \eqref{eq.aads} via a gauge transformation (toward Fermi-type coordinates); this will be demonstrated in detail in Appendix \ref{sec.fg}.

Furthermore, we review the notions of horizontal and mixed tensor fields, introduced in \cite{hol_shao:uc_ads}, that we will use for our main results.
This will also be useful in future works, when we apply these results to tensorial quantities present within the Einstein equations.
Finally, we conclude by computing asymptotic expansions for various geometric quantities and by introducing the function $f$, whose level sets will later (under additional assumptions on $\gm$ and $\gs$) be shown to be pseudoconvex.

\subsection{Construction of the Spacetimes} \label{sec.aads_infinity}

In this subsection, we define precisely the class of aAdS spacetimes we will consider throughout this paper.
The process is a bit more elaborate than in \cite{hol_shao:uc_ads}, due to the need to include non-static boundaries.

\subsubsection{Preliminaries}

The first step is to prescribe the spacetime topology.
This is done by specifying the topology of its AdS-type boundary at infinity.

\begin{definition} \label{def.aads_manifold}
We define the following manifolds:
\begin{itemize}
\item \emph{Boundary cross-section:} Let $\mc{S}$ be an $(n - 1)$-dimensional complete manifold.

\item \emph{AdS-type boundary:} Fix $T_- < T_+$, and let $\mc{I} := (T_-, T_+) \times \mc{S}$.

\item \emph{aAdS spacetime:} Let $\rho_\ast > 0$, and let $\mc{M} := (0, \rho_\ast) \times \mc{I}$.
\end{itemize}
Also, we let $\rho$ denote the coordinate on $\mc{M}$ for the $(0, \rho_\ast)$-component, and we let $t$ denote the coordinate on both $\mc{M}$ and $\mc{I}$ for the $(T_-, T_+)$-component.
\end{definition}

Like in \cite{hol_shao:uc_ads}, we avoid employing a fully geometrically invariant approach, in that we specify most of our asymptotic assumptions in terms of coordinates.

\begin{definition} \label{def.aads_coord}
Let $\varphi = ( x^1, \dots, x^{n - 1} )$ denote a coordinate system on $\mc{S}$.
\begin{itemize}
\item Let $\varphi_t := ( t, x^1, \dots, x^{n - 1} )$ denote the coordinates on $\mc{I}$ obtained by transporting $\varphi$-coordinates along the $t$-component and appending $t$.

\item Let $\varphi_{\rho, t} := ( \rho, t, x^1, \dots, x^{n - 1} )$ be the coordinates on $\mc{M}$ obtained by transporting the $\varphi$-coordinates along the $(\rho, t)$-components and appending $( \rho, t )$.
\end{itemize}
\end{definition}

Note that the coordinate vector fields $\partial_t$ arising from the above transported coordinate systems define a \emph{global} vector field on both $\mc{I}$ and $\mc{M}$.
Similarly, the coordinate vector fields $\partial_\rho$ define a global vector field on $\mc{M}$.

\begin{definition} \label{def.aads_index}
Let $\varphi$ be a coordinate system on $\mc{S}$.
From now on, we adopt the following coordinate and indexing conventions:
\begin{itemize}
\item We use upper-case Latin indices $A, B, \dots$ to denote $\varphi$-coordinate components.
Similarly, we use $x^A, x^B, \dots$ to refer to $\varphi$-coordinate functions.

\item We use lower-case Latin indices $a, b, \dots$ to denote $\varphi_t$-coordinate components.
Similarly, we use $x^a, x^b, \dots$ to refer to $\varphi_t$-coordinate functions.

\item We use lower-case Greek indices $\alpha, \beta, \dots$ to denote $\varphi_{ \rho, t }$-coordinate components.
Similarly, we use $x^\alpha, x^\beta, \dots$ to refer to $\varphi_{ \rho, t }$-coordinate functions.
\end{itemize}
\end{definition}

Next, we define the asymptotic properties of the error terms we will encounter.
As mentioned before, the condition here is stronger than that found in \cite{hol_shao:uc_ads}.

\begin{definition} \label{def.O_scr}
Consider a smooth spacetime $( \mc{M}, g )$, where $\mc{M}$ is as in Definition \ref{def.aads_manifold}.
Let $\zeta \in C ( \mc{M} )$, and let $\varphi$ denote a coordinate system on $\mc{S}$.
We use the symbol $\mc{O}_\varphi ( \zeta )$ to denote a smooth function $u$ on an appropriate open subset of $\mc{M}$ (depending on context) such that we have the family of bounds
\begin{equation}
\label{eq.O_scr} | \partial_\rho^k \partial_{ x^{ a_1 } } \dots \partial_{ x^{ a_m } } u | \lesssim_{g, k, m, u} \rho^{-k} \zeta \text{,} \quad \text{for all $k, m \geq 0$,}
\end{equation}
where the $x^{ a_i }$'s refer to any of the coordinates in $\varphi_{\rho, t}$ except for $\rho$.

When $\varphi$ is clear from context, we omit it from notation and write $\mc{O} ( \zeta )$.
\end{definition}

\subsubsection{Admissible Spacetimes}

We can now define our class of aAdS spacetimes:

\begin{definition} \label{def.aads_gen}
Let $g$ be a smooth Lorentzian metric on $\mc{M}$.
We say that $( \mc{M}, g )$ is an \emph{admissible aAdS segment} iff the following conditions hold:
\begin{enumerate}
\item There exist symmetric covariant $2$-tensors $\ggm$, $\ggs$, $\ggr$ on $\mc{M}$, with $\ggm$ and $\ggs$ independent of $\rho$, such that $g$ can be expressed as
\begin{equation}
\label{eq.aads_gen_met} g = \rho^{-2} [ \ggm + \rho^2 \ggs + \rho^3 \ggr ] \text{.}
\end{equation}

\item $\ggm$ and $\ggs$ have the forms
\begin{equation}
\label{eq.aads_gen_g} \ggm = d \rho^2 - d t^2 + \gamma \text{,} \qquad \ggs = \varsigma \cdot d \rho^2 + \gs \text{,}
\end{equation}
where $\varsigma \in C^\infty ( \mc{I} )$, where $\gamma$ and $\gs$ are symmetric covariant $2$-tensors on $\mc{I}$, and where $\gamma$ has no components in the $t$-direction (but can depend on $t$).

\item There exists a finite family $\Xi$ of coordinate systems on $\mc{S}$ covering all of $\mc{S}$, such that for any $\varphi \in \Xi$, the components with respect to the $\varphi_{\rho, t}$-coordinates of $\ggm$, $\ggs$, $E$, and the metric dual $\ggm^{-1}$ of $\ggm$ satisfy\footnote{Since $\ggm$ and $\ggs$ are independent of $\rho$, condition \eqref{eq.aads_gen_asymp} simply implies that all $x^a$-coordinate derivatives of their components are bounded on $\mc{I}$.}
\begin{equation} \label{eq.aads_gen_asymp}
\ggm_{\alpha \beta} = \mc{O}_\varphi (1) \text{,} \qquad \ggm^{\alpha \beta} = \mc{O}_\varphi (1) \text{,} \qquad \ggs_{\alpha \beta} = \mc{O}_\varphi (1) \text{,} \qquad \ggr_{\alpha \beta} = \mc{O}_\varphi ( 1 ) \text{.}
\end{equation}
We call this $\Xi$ a \emph{bounded family of coordinates} on $( \mc{M}, g )$.
\end{enumerate}
\end{definition}

\begin{definition} \label{def.aads_boundary}
In the context of Definition \ref{def.aads_gen}, we refer to the Lorentzian manifold $( \mc{I}, \gm )$, where $\gm := - dt^2 + \gamma$ is the restriction of $\ggm$ to the $\mc{I}$-tangent directions, as the induced \emph{AdS-type boundary}.
Furthermore, we will occasionally slightly abuse notation and use $\mc{I}$ to refer to this conformal boundary $\{ \rho = 0 \}$ of $( \mc{M}, g )$.
\end{definition}

Suppose $( \mc{M}, g )$ is such an admissible aAdS segment, on which $\Xi$ is a bounded family of coordinates.
Then, the conditions \eqref{eq.aads_gen_met}--\eqref{eq.aads_gen_asymp} can be stated less formally in the following coordinate representation: with respect to any $\varphi \in \Xi$,
\begin{align} \label{eq.aads_gen}
\rho^2 g &= ( 1 + \rho^2 \ggs_{ \rho \rho } ) d \rho^2 + ( - dt^2 + \ggm_{A B} d x^A d x^B + \rho^2 \ggs_{a b} dx^a dx^b ) \\
\notag &\qquad + \mc{O}_\varphi ( \rho^3 ) \cdot d x^\alpha d x^\beta \text{.}
\end{align}
Letting $r := \rho^{-1}$, then \eqref{eq.aads_gen} takes the more familiar form \eqref{eq.aads_gen_pre}.

The spacetimes described in Definition \ref{def.aads_gen} resemble those in \cite[Definition 2.6]{hol_shao:uc_ads}, but with the following specific differences:
\begin{enumerate}
\item We apply a more restrictive prescription of ``$\mc{O}_\varphi$" error terms: unlike in \cite{hol_shao:uc_ads}, $x^a$-derivatives do not lose powers of $\rho$.
Although this is strictly more stringent than before, the condition is more natural and is necessary in order to convert \eqref{eq.aads_gen} to Fefferman--Graham form; see below.

\item On the other hand, because of (1), the metric expansion here (see \eqref{eq.aads_gen_g}, \eqref{eq.aads_gen}) is allowed to be more general than was prescribed in \cite[Definition 2.6]{hol_shao:uc_ads}, in which $\ggs_{\rho \rho} \equiv -1$.

\item We now allow for the boundary metric $\gm$ to be time-dependent.
\end{enumerate}

\begin{remark}
We note the conditions $\ggm_{tt} \equiv -1$ and $\ggm_{t A} \equiv 0$ in \eqref{eq.aads_gen_g} contain an implicit semi-global geometric assumption.
In particular, \eqref{eq.aads_gen_g} can be locally forced by defining $t$ as the affine parameters of a family of normal timelike geodesics emanating from a cross-section of $\mc{I}$.
What is nontrivially assumed, though, is that the foliation defined by this special $t$ remains globally regular on all of $\mc{I}$.

The assumptions \eqref{eq.aads_gen_g} serve to simplify computations and the construction of the pseudoconvex hypersurfaces in Section \ref{sec.aads_f}.
It is not clear to us at the moment whether the main result holds without this assumption.
\end{remark}

\begin{remark}
We also remark that this semi-global regularity condition mentioned above is gauge-dependent.
More specifically, this condition is not necessarily preserved by a conformal transformation of $\gm$ (arising from a change of the spacetime coordinate $\rho$), as described in the remark below Theorem \ref{theo:mti}.
\end{remark}

\begin{remark}
In fact, the finiteness of $\Xi$ in Definition \ref{def.aads_gen} is not strictly necessary.
However, if $\Xi$ is to be infinite, then we must also assume that the constants associated with all the ``$\mc{O}_\varphi (1)$'s" in \eqref{eq.aads_gen_asymp} are independent of $\varphi$.
For clarity and simplicity, we assume finite $\Xi$, since this is satisfied by all aAdS spacetimes of interest.
\end{remark}

Finally, we fix the following notations:

\begin{definition} \label{def.aads_nabla}
Let $\nabla$ denote the Levi-Civita connection associated with $g$, and let $\nasla$ denote the induced connections on the level sets of $(\rho, t)$, i.e., the copies of $\mc{S}$.
\end{definition}

\subsubsection{Fefferman--Graham Spacetimes}

By applying an appropriate change of coordinates, we can convert an admissible aAdS segment into ``Fefferman-Graham form", for which all the information in $g$ resides in the $\mc{I}$-tangent directions; see \cite{MR837196}.
This is described in detail in Appendix \ref{sec.fg}.
Our analysis therefore reduces to admissible Fefferman--Graham-aAdS segments, which we define as follows:

\begin{definition} \label{def.aads}
We say that $( \mc{M}, g )$ is an \emph{admissible Fefferman--Graham-aAdS} (\emph{FG-aAdS}) \emph{segment} iff the following conditions hold:
\begin{enumerate}
\item $( \mc{M}, g )$ is an \emph{admissible aAdS segment}, as in Definition \ref{def.aads_gen}.

\item Both $\ggs$ and $\ggr$ only contain components tangent to the level sets of $\rho$.
\end{enumerate}
Moreover, for such an admissible FG-aAdS segment $( \mc{M}, g )$, we let $\gm$ and $\gs$ denote the restrictions of $\ggm$ and $\ggs$, respectively, to the $\mc{I}$-tangent directions.
\end{definition}

The main content of Definition \ref{def.aads} is the less formal but more intuitive coordinate expansion \eqref{eq.aads}.
Indeed, from \eqref{eq.aads_gen} and Definition \ref{def.aads}, we observe for an admissible FG-aAdS segment $( \mc{M}, g )$: for any $\varphi \in \Xi$, where $\Xi$ is a bounded family of coordinates on $( \mc{M}, g )$, one has the expansion \eqref{eq.aads} for $g$.
Through most of this paper, we will work directly with the representation \eqref{eq.aads}, with the implicit understanding that the precise descriptions are as in Definitions \ref{def.aads_gen} and \ref{def.aads}. 

For example, by a change of the $\rho$-variable, one can express AdS spacetime as an FG-aAdS segment.
In terms of our current notations, we then have
\begin{equation}
\label{eq.aads_ads} \gm = - dt^2 + \ga \text{,} \qquad \gs = - \frac{1}{2} ( dt^2 + \ga ) \text{,}
\end{equation}
where $\ga$ denotes the canonical metric on the unit sphere $\Sph^{n - 1}$.
More generally, by employing an appropriate change of variables \cite{Henneaux}, one can show that any Kerr-AdS spacetime with $n\geq 3$ also has the expansion \eqref{eq.aads_ads}.

\subsection{Horizontal Tensor Fields} \label{sec.aads_hor}

Assume, as detailed in Definition \ref{def.aads}, an admissible FG-aAdS spacetime segment $( \mc{M}, g )$.
The other half of the formalism we will require in this article is the notion of horizontal and mixed tensor fields on $\mc{M}$.
By ``horizontal fields", we refer to fields on $\mc{M}$ which are tensor fields on each level set of $(\rho, t)$, i.e., each copy of $\mc{S}$.
By ``mixed fields", we refer to fields which are combinations of standard and horizontal tensor fields.

We will adopt the same definitions and notations as \cite{hol_shao:uc_ads}.
We briefly review these below; for more details, see \cite[Sect.~2.4]{hol_shao:uc_ads}.

\subsubsection{Horizontal and Mixed Fields}

The main objects of interest are as follows:
\begin{itemize}
\item We denote by $T^\mu_\lambda \mc{M}$ the usual $(\mu, \lambda)$-tensor bundle over $\mc{M}$, consisting of all tensors at all points of $\mc{M}$ of rank $(\mu, \lambda)$.
The space of smooth sections of $T^\mu_\lambda \mc{M}$---the \emph{tensor fields} of rank $(\mu, \lambda)$---are denoted $\Gamma T^\mu_\lambda \mc{M}$.

\item We denote by $\ul{T}^m_l \mc{M}$ the ($\mc{S}$-)\emph{horizontal bundle} over $\mc{M}$, containing all tensors of rank $(m, l)$ on each level set of $(\rho, t)$ in $\mc{M}$ (i.e., all \emph{horizontal tensors}).
We let $\Gamma \ul{T}^m_l \mc{M}$ denote the space of smooth sections of $\ul{T}^\mu_\lambda \mc{M}$, i.e., the \emph{horizontal tensor fields} of rank $(m, l)$.

\item We generalize and unify the above by defining the \emph{mixed bundles} as
\begin{equation}
\label{eq.mixed_bundle} T^\mu_\lambda \ul{T}^m_l \mc{M} := T^\mu_\lambda \mc{M} \otimes \ul{T}^m_l \mc{M} \text{.}
\end{equation}
Similarly, we let $\Gamma T^\mu_\lambda \ul{T}^m_l \mc{M}$ denote the corresponding space of smooth sections of $T^\mu_\lambda \ul{T}^m_l \mc{M}$, i.e., the \emph{mixed tensor fields}.
\end{itemize}
Recall that by duality, we can consider any $A \in \Gamma T^\mu_\lambda \ul{T}^m_l \mc{M}$ as a $C^\infty ( \mc{M} )$-multilinear map on the appropriate number of standard and horizontal vector fields and $1$-forms.
Moreover, note in particular that $\Gamma T^0_0 \mc{M} = \Gamma \ul{T}^0_0 \mc{M} = C^\infty ( \mc{M} )$.

\subsubsection{Connection and Curvature}

Next, recall that the Levi-Civita connection $\nabla$ on $( \mc{M}, g )$ induces a bundle connection $\nabla$ on any $T^\mu_\lambda \mc{M}$.
These connections in turn induce \emph{horizontal connections} $\nasla$ on the $\ul{T}^m_l \mc{M}$ by projecting to the level sets of $(\rho, t)$.
The connections $\nabla$ and $\nasla$ can then be canonically combined to obtain \emph{mixed connections}---also denoted $\nabla$ here---on the $T^\mu_\lambda \ul{T}^m_l \mc{M}$'s.

For practical purposes, the main properties of mixed connections are as follows:
\begin{itemize}
\item $\nabla$ annihilates both $g$ and the restrictions $\gamma$ of $g$ to the level sets of $(\rho, t)$ (both of which can be considered as mixed tensor fields).

\item Given a fully covariant mixed tensor field $A \in \Gamma T^0_\lambda \ul{T}^0_l \mc{M}$ and a vector field $X \in \Gamma T^1_0 \mc{M}$, then $\nabla_X A$ can be characterized by its actions on vector fields: if $Z_1, \dots, Z_\lambda \in \Gamma T^1_0 \mc{M}$, and if $Y_1, \ldots, Y_l \in \Gamma \ul{T}^1 \mc{M}$, then
\begin{equation} \label{eq.mixed_deriv} \begin{split}
\nabla_X A ( Z_1, \dots, Z_\lambda; Y_1, \dots, Y_l ) &= X [ A ( Z_1, \dots, Z_\lambda; Y_1, \dots, Y_l ) \\
&\qquad - A ( \nabla_X Z_1, \dots, Z_\lambda; Y_1, \dots, Y_l ) - \dots \\
&\qquad - A ( Z_1, \dots, \nabla_X Z_\lambda; Y_1, \dots, Y_l ) \\
&\qquad - A ( Z_1, \dots, Z_\lambda; \nasla_X Y_1, \dots, Y_l ) - \ldots \\
&\qquad - A ( Z_1, \dots, Z_\lambda; Y_1, \dots, \nasla_X Y_l ) \text{.}
\end{split} \end{equation}
\end{itemize}
For more details on the basic definitions, see \cite[Sect.~2.4.1]{hol_shao:uc_ads}.

Finally, given any $A \in \Gamma T^\mu_\lambda \ul{T}^m_l \mc{M}$:
\begin{itemize}
\item We define its \emph{mixed covariant differential} $\nabla A \in \Gamma T^\mu_{ \lambda + 1 } \ul{T}^m_l \mc{M}$ to be the mixed tensor field mapping a vector field $X$ to $\nabla_X A$.

\item In particular, we can make sense of $\Box A \in \Gamma T^\mu_\lambda \ul{T}^m_l \mc{M}$ as the $g$-trace of $\nabla^2 A$, with the trace being applied to the two $\nabla^2$-components.

\item The \emph{mixed curvature} is defined as follows: given $X, Y \in \Gamma T^1_0 \mc{M}$, we set
\begin{equation}
\label{eq.mixed_curv} \mc{R} A \in \Gamma T^\mu_{ \lambda + 2 } \ul{T}^m_l \mc{M} \text{,} \qquad \mc{R}_{XY} [A] := \nabla^2_{XY} A - \nabla^2_{YX} A \text{.}
\end{equation}
\end{itemize}
From \eqref{eq.mixed_deriv} and direct computations, we obtain the following identity:

\begin{proposition} \label{thm.mixed_curv_hor}
Let $\phi \in \Gamma \ul{T}^0_l \mc{M}$.
Then, given any spacetime vector fields $X, Y$ and horizontal vector fields $Z_1, \dots, Z_l$, we have:
\begin{equation} \label{eq.mixed_curv_hor} \begin{split}
\mc{R}_{XY} \phi (Z_1, \dots, Z_l) &= - \phi ( \nasla_X ( \nasla_Y Z_1 ) - \nasla_Y ( \nasla_X Z_1 ) - \nasla_{ [X, Y] } Z_1, \dots, Z_l ) \\
&\qquad - \dots \\
&\qquad - \phi ( Z_1, \dots, \nasla_X ( \nasla_Y Z_l ) - \nasla_Y ( \nasla_X Z_l ) - \nasla_{ [X, Y] } Z_l ) \text{.}
\end{split} \end{equation}
In particular, if both $X$ and $Y$ are also horizontal, then \eqref{eq.mixed_curv_hor} reduces to the usual Riemann curvature operator on the level sets of $(\rho, t)$.
\end{proposition}

\subsubsection{Index Conventions}

We will use capital Latin letters to denote \emph{horizontal multi-indices}, i.e., zero or more horizontal indices.
Repeated indices represent summations over all individual indices.

Furthermore, for horizontal tensors, we let $| \cdot |$ denote the pointwise tensor norm:
\begin{equation}
\label{eq.tensor_norm} | \phi |^2 := \phi_I \phi^I \text{.}
\end{equation}
Note that the above notational conventions also cover the purely scalar case, in which all multi-indices can essentially be ignored.

\subsection{Asymptotic Expansions} \label{sec.aads_comp}

Again, we assume an admissible FG-aAdS spacetime segment $( \mc{M}, g )$.
In this subsection, we compute asymptotic expansions associated with various geometric quantities on $\mc{M}$.

Here, and in the remainder of this paper, we assume a bounded family $\Xi$ of coordinates on $( \mc{M}, g )$, as in Definition \ref{def.aads_gen}.
Throughout, when we write $\mc{O} ( \zeta )$ (see Definition \ref{def.O_scr}), we will implicitly assume this to be with respect to some $\varphi \in \Xi$.

\subsubsection{Metric and Christoffel Symbol Expansions}

First, we list the asymptotics of the metric and its corresponding Christoffel symbols.

\begin{proposition} \label{thm.g}
With respect to any $\varphi \in \Xi$, the following hold:
\begin{itemize}
\item The components of $g$ satisfy
\begin{equation}
\label{eq.g} g_{\rho \rho} = \rho^{-2} \text{,} \qquad g_{\rho a} = 0 \text{,} \qquad g_{a b} = \rho^{-2} \gm_{ab} + \gs_{ab} + \mc{O} ( \rho ) \text{.}
\end{equation}
In particular,
\begin{equation}
\label{eq.g_t} g_{tt} = - \rho^{-2} + \gs_{tt} + \mc{O} ( \rho ) \text{,} \qquad g_{t A} = \gs_{t A} + \mc{O} ( \rho ) \text{.}
\end{equation}

\item The dual of $g$ satisfies
\begin{equation}
\label{eq.g_inv} g^{\rho \rho} = \rho^2 \text{,} \qquad g^{\rho a} = 0 \text{,} \qquad g^{a b} = \rho^2 \gm^{ab} - \rho^4 \gm^{ac} \gm^{bd} \gs_{cd} + \mc{O} ( \rho^5 ) \text{.}
\end{equation}
In particular,
\begin{equation}
\label{eq.g_inv_t} g^{tt} = - \rho^2 - \rho^4 \gs_{tt} + \mc{O} ( \rho^5 ) \text{,} \qquad g^{t A} = \rho^4 \gm^{A B} \gs_{t B} + \mc{O} ( \rho^5 ) \text{.}
\end{equation}

\item The Christoffel symbols with respect to these coordinates satisfy
\begin{align}
\label{eq.Gamma} \Gamma^\rho_{\rho \rho} = - \rho^{-1} \text{,} &\qquad \Gamma^\rho_{\rho a} = 0 \text{,} \\
\notag \Gamma^\rho_{a b} = \rho^{-1} \gm_{ab} + \mc{O} ( \rho^2 ) \text{,} &\qquad \Gamma^a_{\rho \rho} = 0 \text{,} \\
\notag \Gamma^a_{\rho b} = - \rho^{-1} \delta^a_b + \rho \gm^{ac} \gs_{cb} + \mc{O} ( \rho^2 ) \text{,} &\qquad \Gamma^a_{b c} = \mathring{\Gamma}^a_{b c} + \mc{O} ( \rho^2 ) \text{,}
\end{align}
where $\mathring{\Gamma}^a_{b c}$ denotes the corresponding Christoffel symbol associated with $\gm$.
In addition, when $\Gamma^a_{b c}$ contains a $t$-component, we have:\footnote{The presence of the (leading-order) quantity $\partial_t \gm$ in \eqref{eq.Gamma_t} is a fundamental difference between the current setting and that of \cite{hol_shao:uc_ads}.}
\begin{equation}
\label{eq.Gamma_t} \Gamma^t_{a b} = \frac{1}{2} \partial_t \gm_{a b} + \mc{O} ( \rho^2 ) \text{,} \qquad \Gamma^a_{t b} = \frac{1}{2} \gm^{a c} \partial_t \gm_{c b} + \mc{O} ( \rho^2 ) \text{.}
\end{equation}
\end{itemize}
\end{proposition}

\subsubsection{Curvature Coefficients}

We will also need to compute the asymptotics for the mixed curvature operator $\mc{R}$ defined in \eqref{eq.mixed_curv}.

\begin{proposition} \label{thm.curv}
Let $\phi \in \Gamma \ul{T}^0_l \mc{M}$.
Then, with respect to any $\varphi \in \Xi$, we have
\begin{equation}
\label{eq.curv} | \mc{R}_{\rho a} \phi | \lesssim_{g, l} \rho | \phi | \text{,} \qquad | \mc{R}_{a b} \phi | \lesssim_{g, l} | \phi | \text{.}
\end{equation}
\end{proposition}

\begin{proof}
The computations are analogous to those found in \cite{hol_shao:uc_ads}; however, since \eqref{eq.curv} contains some nonstandard definitions involving mixed tensor fields, we give some details for the reader's convenience.

First, using that $\nasla_\alpha \partial_A$ is the orthogonal projection of $\nabla_\alpha \partial_A$ to the $(\rho, t)$-level sets, along with the asymptotic identities in Proposition \ref{thm.g}, we obtain
\begin{equation}
\label{eq.curv_0} \nasla_\alpha \partial_A = \Gamma^B_{ \alpha A } \partial_B + \sum_{ B = 1 }^{ n - 1 } \mc{O} ( \rho^2 ) \cdot \partial_B \text{.}
\end{equation}
Differentiating \eqref{eq.curv_0} and then applying \eqref{eq.curv_0} yields the identity
\begin{align}
\label{eq.curv_1} \nasla_\alpha ( \nasla_\beta \partial_A ) - \nasla_\beta ( \nasla_\alpha \partial_A ) &= \partial_\alpha \Gamma^B_{\beta A} \partial_B - \partial_\beta \Gamma^B_{\alpha A} \partial_B + \Gamma^B_{\beta A} \Gamma^C_{\alpha B} \partial_C \\
\notag &\qquad - \Gamma^B_{\alpha A} \Gamma^C_{\beta B} \partial_C + \sum_{ C = 1 }^{n - 1} \mc{O} ( \rho ) \cdot \partial_C \text{.}
\end{align}

Since \eqref{eq.Gamma} implies
\begin{equation}
\label{eq.curv_2} \partial_\rho \Gamma^B_{a A} - \partial_a \Gamma^B_{\rho A} = \mc{O} ( \rho ) \text{,} \qquad \Gamma^B_{a A} \Gamma^C_{ \rho B } - \Gamma^B_{\rho A} \Gamma^C_{a B} = \mc{O} ( \rho ) \text{,}
\end{equation}
then combining \eqref{eq.curv_0} and \eqref{eq.curv_2} yields the first part of \eqref{eq.curv}.
Similarly, since
\begin{equation}
\label{eq.curv_3} \partial_a \Gamma^B_{b A} - \partial_b \Gamma^B_{a A} = \mc{O} ( 1 ) \text{,} \qquad \Gamma^B_{b A} \Gamma^C_{a B} - \Gamma^B_{a A} \Gamma^C_{b B} = \mc{O} ( 1 ) \text{,}
\end{equation}
by \eqref{eq.Gamma}, then \eqref{eq.curv_0} and \eqref{eq.curv_3} implies the second part of \eqref{eq.curv}.
\end{proof}

\subsection{The $f$-Foliation} \label{sec.aads_f}

As before, let $( \mc{M}, g )$ be an admissible FG-aAdS segment, see Definitions \ref{def.aads_manifold} and \ref{def.aads}.
Moreover, we now normalize the time interval as
\begin{equation}
\label{eq.t_normalize} 0 = T_- < t < T_+ = \pi T \text{.}
\end{equation}

Analogous to \cite{hol_shao:uc_ads}, we construct a function on $\mc{M}$ whose level sets will, \emph{under additional assumptions} (see Definition \ref{def.pcp}), be shown to be pseudoconvex near the conformal boundary $\mc{I}$.
For this purpose, we define the following:

\begin{definition} \label{def.f}
Fix a constant $\xi \geq 0$.
We then define $f := f_{T, \xi} : \mc{M} \rightarrow \R$ by
\begin{align} \label{eq.f}
f ( \rho, t, x^A ) = \frac{ \rho }{ \fbd (t) } \text{,}
\end{align}
where $\fbd: [0, \pi T] \rightarrow \R$ satisfies:
\begin{equation}
\label{eq.psi} \fbd (t) :=
\begin{cases}
  \exp \left( \frac{\xi}{4} t \right) \sin ( \mu t ) & t \in \left[ 0, \frac{ \pi T }{2} \right) \text{,} \\
  \exp \left( \frac{\xi}{4} ( \pi T - t ) \right) \sin ( \mu ( \pi T - t ) ) & t \in \left[ \frac{ \pi T }{2}, \pi T \right] \text{.}
\end{cases}
\end{equation}
Here, $\mu$ is the unique constant satisfying
\begin{equation}
\label{eq.mu} T = \frac{2}{\mu} - \frac{2}{\mu \pi} \arctan \left( \frac{ 4 \mu }{ \xi } \right) \text{,} \qquad \frac{1}{T} \leq \mu < \frac{2}{T} \text{.}
\end{equation}
\end{definition}

Note that the level sets of $f$ foliate a neighborhood of $\mc{I}$ in $\mc{M}$.
This specific choice of $\fbd$ will be justified in the proof of Theorem \ref{thm.pseudoconvex}.
For now, observe:
\begin{itemize}
\item $\fbd \in C^2 [0, \pi T]$, and $\fbd$ is smooth on $[ 0, \frac{ \pi T }{ 2 } )$ and $( \frac{ \pi T }{ 2 }, \pi T ]$.

\item $\fbd$ is strictly positive on $(0, \pi T)$, and $\fbd (0) = \fbd ( \pi T ) = 0$.

\item $\fbd = \mc{O} (1)$ on each of the intervals $( 0, \frac{ \pi T }{ 2 } )$ and $( \frac{ \pi T }{ 2 }, \pi T )$.
\end{itemize}

\begin{remark}
Compared to the $f$ employed in \cite{hol_shao:uc_ads}, the new element here is the parameter $\xi$, which will be used to compensate for the non-static boundary.
By choosing $\xi = 0$, we recover precisely the corresponding function $f$ used in \cite{hol_shao:uc_ads}.\footnote{In this case, $\eta (t) = \sin ( T^{-1} t )$, and hence $f$ is in fact everywhere smooth.}
\end{remark}

A technical issue here that was not encountered in \cite{hol_shao:uc_ads} is that $f$ fails to be smooth.
(In particular, $f$ fails to be $C^3$ at $\mc{M} \cap \{ t = \frac{ \pi T }{2} \}$.)
Thus, we often restrict attention to regions in which all objects are smooth:

\begin{definition} \label{def.M_mp}
We define the following regions,
\begin{equation}
\label{eq.M_mp} \mc{M}_- := \mc{M} \cap \left\{ t < \frac{ \pi T }{2} \right\} \text{,} \qquad \mc{M}_+ := \mc{M} \cap \left\{ t > \frac{ \pi T }{2} \right\} \text{,}
\end{equation}
\end{definition}

Furthermore, similar to \cite{hol_shao:uc_ads}, in our main Carleman estimate, it will often be convenient to work not with the gradient of $f$, but with the following:

\begin{definition} \label{def.S}
Let $\grad f$ denote the $g$-gradient of $f$,
\begin{equation}
\label{eq.grad} \grad f := g^{\alpha \beta} \nabla_\alpha f \cdot \partial_\beta \text{,}
\end{equation}
and let $S$ denote the following rescaling of $\grad f$:
\begin{equation}
\label{eq.S} S := f^{n - 3} \grad f \text{.}
\end{equation}
\end{definition}

\subsubsection{Asymptotic Expansions}

The next step is to compute asymptotic properties for $f$.
For this, it will be convenient to introduce a weaker notion (than the $\mathcal{O}$ of Definition \ref{def.O_scr}) of asymptotic error terms, i.e., the notion of asymptotics errors used throughout our previous paper \cite{hol_shao:uc_ads}:

\begin{definition} \label{def.O0_scr}
Let $\zeta \in C ( \mc{M} )$.
We use $\mc{O}_0 ( \zeta )$ to denote any function $u$ on an appropriate open subset of $\mc{M}_+ \cup \mc{M}_-$ such that we have the family of bounds
\begin{equation} \label{eq.O0_scr}
| \partial_{ x^{ \alpha_1 } } \dots \partial_{ x^{ \alpha_m } } u | \lesssim_{g, m, u} \rho^{-m} \zeta \text{,} \qquad m \geq 0 \text{,}
\end{equation}
where the $x^{ \alpha_i }$'s refer to any of the (spacetime) coordinates used in Definition \ref{def.aads}.
\end{definition}

\begin{remark}
The main reason for introducing the above is that while $f \neq \mc{O} ( f )$, we have $f = \mc{O}_0 ( f )$, which allows for easier bookkeeping of error terms.
\end{remark}

The subsequent proposition lists some basic asymptotic properties of $f$:

\begin{proposition} \label{thm.f_lo}
Let $f$, $\fbd$ be as in Definition \ref{def.f}.
Then, the gradient of $f$ satisfies
\begin{align}
\label{eq.f_grad} \grad f &= f \rho \partial_\rho + \fbd' f^2 [ \rho + \mc{O} ( \rho^3 ) ] \partial_t + \fbd' f^2 \cdot \sum_{ A = 1 }^{ n - 1 } \mc{O} ( \rho^3 ) \cdot \partial_{ x^A } \text{,} \\
\notag \nabla^\alpha f \nabla_\alpha f &= f^2 [ 1 - ( \fbd' )^2 f^2 + ( \fbd' )^2 f^2 \cdot \mc{O} ( \rho^2 ) ] \\
\notag &= f^2 + \mc{O}_0 ( f^4 ) \text{.}
\end{align}
In addition, $\Box f$ satisfies 
\begin{equation}
\label{eq.f_box} \Box f = - (n - 1) f + \mc{O}_0 ( f^3 ) \text{.}
\end{equation}
\end{proposition}

\begin{proof}
The first step is to compute derivatives of $f$:
\begin{align}
\label{eq.f_deriv} \partial_\rho f &= f \rho^{-1} \text{,} \qquad \partial_t f = - \fbd' f^2 \rho^{-1} \text{,} \\
\label{eq.f_deriv_2} \partial^2_{\rho \rho} f = 0 \text{,} \qquad \partial^2_{\rho t} f &= - \fbd' f^2 \rho^{-2} \text{,} \qquad \partial^2_{t t} f = - f^2 \rho^{-2} [ \fbd'' \rho - 2 ( \fbd' )^2 f ] \text{.}
\end{align}
Both equations in \eqref{eq.f_grad} follow from \eqref{eq.g_inv}, \eqref{eq.g_inv_t}, and \eqref{eq.f_deriv}.
Next, from \eqref{eq.Gamma}, \eqref{eq.Gamma_t}, and \eqref{eq.f_deriv_2}, we obtain expansions for components of $\nabla^2 f$:
\begin{align}
\label{eq.f_hess} \nabla_{\rho \rho} f &= f \rho^{-2} \text{,} \\
\notag \nabla_{\rho t} f &= - \fbd' f^2 \rho^{-2} [ 2 + \rho^2 \gs_{tt} + \mc{O} ( \rho^3 ) ] \text{,} \\
\notag \nabla_{t t} f &= f \rho^{-2} [ 1 - \fbd'' f \rho + 2 ( \fbd' )^2 f^2 + \mc{O} ( \rho^3 ) + \fbd' f \cdot \mc{O} ( \rho^3 ) ] \text{,} \\
\notag \nabla_{A B} f &= - f \rho^{-2} [ \gm_{A B} + \mc{O} ( \rho^3 ) ] + \frac{1}{2} \fbd' f^2 \rho^{-2} [ \rho \partial_t \gm_{A B} + \mc{O} ( \rho^3 ) ] \text{,} \\
\notag \nabla_{\rho A} f &= - \fbd' f^2 \rho^{-2} [ \rho^2 \gs_{t A} + \mc{O} ( \rho^3 ) ] \text{,} \\
\notag \nabla_{t A} f &= f \rho^{-2} \cdot \mc{O} ( \rho^3 ) + \fbd' f^2 \rho^{-2} \cdot \mc{O} ( \rho^3 ) \text{.}
\end{align}
The final identity \eqref{eq.f_box} now follows from \eqref{eq.g_inv}, \eqref{eq.g_inv_t}, and \eqref{eq.f_hess}.
\end{proof}

\begin{corollary} \label{thm.S_lo}
$S$ satisfies the following asymptotic properties:
\begin{equation}
\label{eq.S_lo} \nabla^\alpha S_\alpha = - 2 f^{n - 2} + \mc{O}_0 ( f^n ) \text{,} \qquad S^\alpha S_\alpha = f^{2n - 4} + \mc{O}_0 ( f^{2n - 2} ) \text{.}
\end{equation}
\end{corollary}

As in \cite{hol_shao:uc_ads}, the $\mc{O}_0$-classes satisfy systematic derivative properties:

\begin{proposition} \label{thm.error_deriv}
Let $\zeta \in C ( \mc{M} )$, and suppose $u = \mc{O}_0 ( \zeta )$ is smooth.
Then,
\begin{equation}
\label{eq.error_deriv} \Box u = \mc{O}_0 ( \zeta ) \text{,} \qquad \nabla^\alpha f \nabla_\alpha u = \mc{O}_0 ( f \zeta ) \text{.}
\end{equation}
\end{proposition}

\begin{proof}
These are consequences of Proposition \ref{thm.g} and \eqref{eq.f_grad}.
\end{proof}

\subsubsection{Adapted frames}

Similar to \cite{hol_shao:uc_ads}, we define a collection of orthonormal frames adapted to the foliation by the level sets of $f$:

\begin{definition} \label{def.frame}
We define local frames $( N, V, E_1, \dots, E_{n-1} )$ as follows:\footnote{$V$ here corresponds to the vector field ``$T$" in \cite{hol_shao:uc_ads}.}
\begin{itemize}
\item Let $( E_1, \dots, E_{n-1} )$ denote local orthonormal frames on the level sets of $(\rho, t)$.
Note that by \eqref{eq.g}, these frames can be chosen such that
\begin{equation}
\label{eq.frame_E} E_X := \rho E_X^A \partial_A \text{,} \qquad E_X^A = \mc{O} ( 1 ) \text{.}
\end{equation}

\item Let $N$ denote the inward-pointing unit normal to level sets of $f$:
\begin{equation}
\label{eq.frame_N} N := | \nabla^\alpha f \nabla_\alpha f |^{ - \frac{1}{2} } \grad f \text{.}
\end{equation}

\item The final (future, timelike) frame component is then given by:
\begin{align}
\label{eq.frame_V} V &:= | g ( \tilde{V}, \tilde{V} ) |^{- \frac{1}{2}} \tilde{V} \text{,} \\
\notag \tilde{V} &:= \partial_t + \fbd' f \partial_\rho - \sum_{ X = 1 }^{ n - 1 } g ( \partial_t + \fbd' f \partial_\rho, E_X ) \cdot E_X \text{.}
\end{align}
\end{itemize}
\end{definition}

Direct computations using \eqref{eq.g} and \eqref{eq.f_grad} then yield the following:

\begin{proposition} \label{thm.frame}
The frames $( N, V, E_X )$ in \eqref{eq.frame_E}-\eqref{eq.frame_V} are orthonormal.
Also:
\begin{itemize}
\item $N$ and $V$ have asymptotic expansions
\begin{align}
\label{eq.frame_NV} N &= [ 1 - ( \fbd' )^2 f^2 + \mc{O}_0 ( f^2 \rho^2 ) ]^{ - \frac{1}{2} } \left[ \rho \partial_\rho + \fbd' f \rho \partial_t + \sum_a \mc{O}_0 ( f \rho^3 ) \cdot \partial_a \right] \text{,} \\
\notag V &= [ 1 - ( \fbd' )^2 f^2 - \gs_{tt} \rho^2 + \mc{O}_0 ( \rho^3 ) ]^{ - \frac{1}{2} } \\
\notag &\qquad \cdot \left\{ \rho \partial_t + \fbd' f \rho \partial_\rho - \rho^3 E_X^A E_X^B \gs_{t A} \partial_B + \sum_B \mc{O}_0 ( \rho^4 ) \cdot \partial_B \right\} \text{.}
\end{align}

\item Furthermore, for $f \ll_g 1$, the following inversion formulas hold:
\begin{align}
\label{eq.frame_inv} [ 1 - ( \fbd' )^2 f^2 ]^\frac{1}{2} \rho \partial_\rho &= [ 1 + \mc{O}_0 ( \rho^2 ) ] N - [ \fbd' f + \mc{O}_0 ( \rho^2 ) ] V + \sum_{ X = 1 }^{ n - 1 } \mc{O}_0 ( \rho^2 ) \cdot E_X \text{,} \\
\notag [ 1 - ( \fbd' )^2 f^2 ]^\frac{1}{2} \rho \partial_t &= [ 1 + \mc{O}_0 ( \rho^2 ) ] V - [ \fbd' f + \mc{O}_0 ( \rho^2 ) ] N + \sum_{ X = 1 }^{ n - 1 } \mc{O}_0 ( \rho^2 ) \cdot E_X \text{.}
\end{align}
\end{itemize}
\end{proposition}

We will also require the following curvature bounds involving the above frames:

\begin{proposition} \label{thm.curv_frame}
Let $\phi \in \Gamma \ul{T}^0_l \mc{M}$.
Then, with $( N, V, E_X )$ as in \eqref{eq.frame_E}-\eqref{eq.frame_V},
\begin{equation}
\label{eq.curv_frame} | \mc{R}_{N V} \phi | \lesssim_{g, l} \rho^3 | \phi | \text{,} \qquad | \mc{R}_{N E_X} \phi | \lesssim_{g, l} f \rho^2 | \phi | \text{.}
\end{equation}
\end{proposition}

\begin{proof}
This follows from using \eqref{eq.curv} along with \eqref{eq.frame_E} and \eqref{eq.frame_NV}.
\end{proof}

\section{The Main Results} \label{sec.proof}

We are now prepared to state and prove the main results of this paper:
\begin{itemize}
\item In Section \ref{sec.proof_pseudo}, we define the \emph{pseudoconvexity criterion}, Definition \ref{def.pcp}, which are assumptions on $\gm$ and $\gs$ (i.e., quantities on $\mc{I}$).
We also show in Theorem \ref{thm.pseudoconvex} that the pseudoconvexity criterion implies the level sets of $f$ (see Definition \ref{def.f}) are pseudoconvex near $\mc{I}$.\footnote{More accurately, we obtain directly the positivity implied by this pseudoconvexity.}

\item In Section \ref{sec.proof_carleman}, we prove that the pseudoconvexity criterion implies our main \emph{Carleman estimate}, Theorem \ref{thm.carleman} near the conformal boundary $\mc{I}$.

\item Finally, in Section \ref{sec.proof_uc}, we apply the Carleman estimate to establish our main \emph{unique continuation} result, Theorem \ref{thm.uc_ads}.
\end{itemize}

\begin{remark}
We stress that in future applications, we will be applying the pseudoconvexity criterion and the Carleman estimate rather than the unique continuation result itself.
As such, we wish to highlight all three points equally in the formal presentation of the main results in this section.
\end{remark}

\subsection{The Pseudoconvexity Criterion} \label{sec.proof_pseudo}

In order to determine pseudoconvexity properties of the level sets of $f$, we need to compute the components of $\nabla^2 f$ in the frame \eqref{eq.frame_E}-\eqref{eq.frame_V}.
More precisely, we must compute the frame components of
\begin{equation}
\label{eq.Q} Q_{\xi, \zeta} = -\nabla^2 f - w_{\xi, \zeta} \cdot g \text{,}
\end{equation}
for a suitable function $w_{\xi, \zeta}$, to be specified below in \eqref{eq.w}.
That $Q_{\xi, \zeta}$ is positive-definite on the tangent spaces of the level sets of $f$ implies that these hypersurfaces are pseudoconvex; see Definition 2.13 and Proposition 2.14 in \cite{hol_shao:uc_ads}.

In the context of our Carleman estimates, it will be more convenient to express this positivity in terms of $\nabla S$ rather than $\nabla^2 f$, where $S$ was defined in \eqref{eq.S}.

\begin{definition} \label{def.dft}
Given a constant $\xi \geq 0$ and $\zeta \in C^\infty ( \mc{M} )$, we let
\begin{equation}
\label{eq.w} w_{\xi, \zeta} := f - \frac{1}{2} f \rho^2 \xi \frac{ | \fbd' | }{ \fbd } + f \rho^2 \zeta \text{,}
\end{equation}
as well as the following modified deformation tensor,
\begin{equation}
\label{eq.dft} \pi_{\xi, \zeta} := - ( \nabla S + f^{n-3} w_{\xi, \zeta} \cdot g ) \text{.}
\end{equation}
\end{definition}

Note that $Q_{\xi, \zeta}$ being positive-definite on the tangent spaces of the level sets of $f$ is equivalent to $\pi_{\xi, \zeta}$ being positive-definite on the same spaces.

\begin{remark}
The term $w_{\xi, \zeta} \cdot g$ in \eqref{eq.Q} and \eqref{eq.dft} reflects the conformal invariance of pseudoconvexity and provides an additional degree of freedom to establish positivity.
The factor $w_{\xi, \zeta}$ is carefully chosen based on the algebraic properties of $\nabla^2 f$ (see \eqref{eq.f_hess} and \eqref{eq.f_hess_tan}) so that $\pi_{\xi, \zeta}$ satisfies the lower bound \eqref{eq.pseudoconvex}.
\end{remark}

\subsubsection{Pseudoconvexity and Positivity}

We now define our main \emph{pseudoconvexity criterion}, which is \emph{stated only in terms of the metric data at infinity}:

\begin{definition} \label{def.pcp}
We say that the \emph{pseudoconvexity property} holds at $\mc{I}$ iff there are constants $K > 0$, $\xi \geq 0$ and a function $\zeta \in C^\infty ( \mc{M} )$ such that:
\begin{enumerate}
\item $\zeta = \mc{O} ( 1 )$.

\item For any vector field $Y := Y^A \partial_A$ on $\mc{I}$ that is tangent to $\mc{S}$, we have
\begin{equation}
\label{eq.pcp_nonstatic} | \mc{L}_{ \partial_t } \gm ( Y, Y ) | \leq \xi \cdot \gm ( Y, Y ) \text{.}
\end{equation}

\item For any vector field $X := X^t \partial_t + X^A \partial_A$ on $\mc{I}$, the tensor field
\begin{equation}
\label{eq.Q_inf} \mathbf{Q}_{\xi, \zeta} := - \gs - \left( \mu^2 + \frac{ \xi^2 }{16} \right) dt^2 - \zeta \gm \text{,}
\end{equation}
where $\mu$ is defined implicitly by \eqref{eq.mu}, satisfies the positivity property
\begin{equation}
\label{eq.pcp_positive} \mathbf{Q}_{\xi, \zeta} ( X, X ) \geq K [ (X^t)^2 + \gm_{AB} X^A X^B ] \text{.}
\end{equation}
\end{enumerate}
\end{definition}

The main point here is that Definition \ref{def.pcp}, which is a condition purely on the metric asymptotics at infinity $\mc{I}$, implies that the level sets of $f$ in the spacetime are indeed pseudoconvex, at least for $f \ll_g 1$.
This is captured in the form that we will use later through the following theorem:

\begin{theorem} \label{thm.pseudoconvex}
Suppose the pseudoconvexity property holds at $\mc{I}$, and let $K$, $\xi$, $\zeta$ be the parameters from Definition \ref{def.pcp}.
Then, for any $1$-form $\theta$ on $\mc{M}$,
\begin{align}
\label{eq.pseudoconvex} \pi_{\xi, \zeta}^{\alpha \beta} \theta_\alpha \theta_\beta &\geq [ K f^{n-2} \rho^2 + \mc{O}_0 ( f^{n - 1} \rho^2 ) ] \left( | \theta (V) |^2 + \sum_{ X = 1 }^{ n-1 } | \theta ( E_X ) |^2 \right) \\
\notag &\qquad - [ (n - 1) f^{n-2} + \mc{O}_0 ( f^n ) ] | \theta (N) |^2 \text{.}
\end{align}
\end{theorem}

\begin{remark}
We note in particular that the pseudoconvexity criterion of Definition \ref{def.pcp} is gauge-dependent, that is, it is not necessarily preserved by a conformal transformation of $\gm$; see also the remark below Theorem \ref{theo:mti}.
\end{remark}

\subsubsection{Proof of Theorem \ref{thm.pseudoconvex}}

The main step of the proof is the computation for $\pi_{\xi, \zeta}$, which proceeds analogously to that in \cite{hol_shao:uc_ads}.

\begin{lemma} \label{thm.Q_frame}
$\pi_{\xi, \zeta}$ is symmetric, and its $f$-tangent components satisfy
\begin{align}
\label{eq.Q_tan} \pi_{\xi, \zeta} (V, V) &= \frac{1}{ \fbd } \left( \fbd'' - \frac{\xi}{2} | \fbd' | \right) f^{n - 2} \rho^2 - ( \gs_{t t} + \zeta \gm_{t t} ) f^{n - 2} \rho^2 + \mc{O}_0 ( f^{n - 1} \rho^2 ) \text{,} \\
\notag \pi_{\xi, \zeta} (V, E_X) &= - E^A_X ( \gs_{t A} + \zeta \gm_{t A} ) \cdot f^{n - 2} \rho^2 + \mc{O}_0 ( f^{n - 1} \rho^2 ) \text{,} \\
\notag \pi_{\xi, \zeta} (E_X, E_Y) &= \frac{1}{2 \fbd} \cdot E_X^A E_X^B ( | \fbd' | \xi \gm_{A B} - \fbd' \partial_t \gm_{A B} ) f^{n - 2} \rho^2 \\
\notag &\qquad - E_X^A E_Y^B ( \gs_{A B} + \zeta \gm_{A B} ) f^{n - 2} \rho^2 + \mc{O}_0 ( f^{n - 1} \rho^3 ) \text{.}
\end{align}
Moreover, the remaining components of $\pi_{\xi, \zeta}$ satisfy
\begin{align}
\label{eq.Q_nor} \pi_{\xi, \zeta} (N, N) &= - (n - 1) f^{n - 2} + \mc{O}_0 ( f^n ) \text{,} \\
\notag \pi_{\xi, \zeta} (N, V) &= \mc{O}_0 ( f^n \rho ) \text{,} \\
\notag \pi_{\xi, \zeta} (N, E_X) &= \mc{O}_0 ( f^{n - 1} \rho^2 ) \text{.}
\end{align}
\end{lemma}

\begin{proof}
The main computations behind \eqref{eq.Q_tan} and \eqref{eq.Q_nor} are the formulas for $\nabla^2 f$, with respect to the aforementioned orthonormal frames.
Similar to \cite{hol_shao:uc_ads}, we will need more precise expansions for components tangent to the level sets of $f$:
\begin{align}
\label{eq.f_hess_tan} \nabla_{V V} f &= f + ( \gs_{tt} - \fbd'' \fbd^{-1} ) f \rho^2 + \mc{O}_0 ( f^2 \rho^2 ) \text{,} \\
\notag \nabla_{V E_X} f &= E^A_X \gs_{t A} \cdot f \rho^2 + \mc{O}_0 ( f^2 \rho^2 ) \text{,} \\
\notag \nabla_{E_X E_Y} f &= - \delta_{X Y} f + E_X^A E_Y^B \left( \frac{1}{2} \fbd' \fbd^{-1} \partial_t \gm_{A B} + \gs_{A B} \right) f \rho^2 + \mc{O}_0 ( f \rho^3 ) \text{.}
\end{align}
For the remaining components of $\nabla^2 f$, we have
\begin{align}
\label{eq.f_hess_nor} \nabla_{N N} f &= f - 2 ( \fbd' )2 f^3 + \mc{O}_0 ( f^5 ) = f + \mc{O}_0 ( f^3 ) \text{,} \\
\notag \nabla_{N V} f &= [ 1 + \mc{O}_0 ( f^2 ) ] [ - \fbd' \fbd'' f^3 \rho  + \mc{O}_0 ( f^2 \rho^2 ) ] = \mc{O}_0 ( f^3 \rho ) \text{,} \\
\notag \nabla_{N E_X} f &= \mc{O}_0 ( f^2 \rho^2 ) \text{.}
\end{align}
Combining \eqref{eq.dft}, \eqref{eq.f_hess_tan}, and \eqref{eq.f_hess_nor} results in \eqref{eq.Q_tan} and \eqref{eq.Q_nor}.
\end{proof}

\begin{remark}
We note that $\pi_{\xi, \zeta} (N, V)$ behaves worse compared to the static case considered in \cite{hol_shao:uc_ads}, while the other components behave similarly compared to \cite{hol_shao:uc_ads}.
\end{remark}

The following properties of $\fbd$ can be verified through direct computations:

\begin{lemma} \label{thm.psi}
Given $\xi > 0$ and $T > 0$, the function $\fbd$ in \eqref{eq.psi} satisfies
\begin{equation}
\label{eq.psi_ode} \fbd^{\prime \prime} (t) - \frac{\xi}{2} | \fbd^\prime (t) | + \left( \frac{ \xi^2 }{16} + \mu^2 \right) \fbd (t) = 0 \text{,} \qquad t \in \left( 0, \frac{ \pi T }{2} \right) \cup \left( \frac{ \pi T }{2}, T \right) \text{.}
\end{equation}
Furthermore, $\fbd'''$ has jump discontinuity at $\frac{\pi T}{2}$, since
\begin{align}
\label{eq.psi_jump} \lim_{ t \rightarrow \frac{ \pi T }{2} - } \fbd''' &= - \frac{\xi}{2} \left( \frac{\xi^2}{16} + \mu^2 \right) \exp \left( \frac{ \pi T \xi }{8} \right) \sin \left( \frac{ \pi T \mu }{2} \right) < 0 \text{,} \\
\notag \lim_{ t \rightarrow \frac{ \pi T }{2} + } \fbd''' &= + \frac{\xi}{2}  \left( \frac{\xi^2}{16} + \mu^2 \right) \exp \left( \frac{ \pi T \xi }{8} \right) \sin \left( \frac{ \pi T \mu }{2} \right) > 0 \text{.}
\end{align}
\end{lemma}

By combining Definition \ref{def.pcp} with \eqref{eq.Q_tan} and \eqref{eq.psi_ode}, we connect $\pi_{\xi, \zeta}$ and ${\bf Q}_{\xi, \zeta}$:

\begin{lemma} \label{eq.Q_ex}
The following identities hold:
\begin{align}
\label{eq.Q_tan_ex} \pi_{\xi, \zeta} (V, V) &= f^{n - 2} \rho^2 \cdot {\bf Q}_{\xi, \zeta} ( \partial_t, \partial_t ) + \mc{O}_0 ( f^{n - 1} \rho^2 ) \text{,} \\
\notag \pi_{\xi, \zeta} (V, E_X) &= f^{n - 2} \rho^2 \cdot {\bf Q}_{\xi, \zeta} ( \partial_t, E_X^A \partial_A ) + \mc{O}_0 ( f^{n - 1} \rho^2 ) \text{,} \\
\notag \pi_{\xi, \zeta} (E_X, E_Y) &= \frac{1}{ 2 \eta } f^{n - 2} \rho^2 \cdot E_X^A E_Y^B ( | \eta' | \xi \gm_{A B} - \fbd' \partial_t \gm_{A B} ) \\
\notag &\qquad + f^{n - 2} \rho^2 \cdot {\bf Q}_{\xi, \zeta} ( E^A_X \partial_A, E^B_Y \partial_B ) + \mc{O}_0 ( f^{n - 1} \rho^2 ) \text{.}
\end{align}
\end{lemma}

Finally, let $\theta$ be as in the hypotheses of Theorem \ref{thm.pseudoconvex}, and define in addition
\begin{equation}
\label{eq.theta_var} \slashed{\theta} := \sum_{ X = 1 }^{n - 1} \theta ( E_X ) \cdot E_X^A \partial_A \text{,} \qquad \check{\theta} := - \theta ( V ) \cdot \partial_t + \slashed{\theta} \text{,}
\end{equation}
which can be viewed as vector fields on $\mc{M}$ or as $\rho$-parametrized families of vector fields on $\mc{I}$.
Using \eqref{eq.Q_nor}, \eqref{eq.Q_tan_ex}, and that $( N, V, E_X )$ is orthonormal, we have
\begin{align}
\label{eq.pseudoconvex_0} \pi_{\xi, \zeta}^{\alpha \beta} \theta_\alpha \theta_\beta &= f^{n - 2} \rho^2 \cdot {\bf Q}_{\xi, \zeta}^{a b} \check{\theta}_a \check{\theta}_b + \frac{1}{ 2 \eta } f^{n - 2} \rho^2 \cdot ( | \eta' | \xi \gm_{A B} - \fbd' \partial_t \gm_{A B} ) \slashed{\theta}^A \slashed{\theta}^B \\
\notag &\qquad + \sum_{ \mf{A}, \mf{B} \in \{ V, E_1, \dots, E_{n-1} \} } \mc{O}_0 ( f^{n - 1} \rho^2 ) \cdot \theta ( \mf{A} ) \theta ( \mf{B} ) \\
\notag &\qquad - [ (n - 1) f^{n - 2} + \mc{O}_0 ( f^n ) ] \cdot | \theta ( N ) |^2 \\
\notag &\qquad + \sum_{ \mf{A} \in \{ V, E_1, \dots, E_{n-1} \} } \mc{O}_0 ( f^n \rho ) \cdot \theta ( N ) \theta ( \mf{A} ) \\
\notag &= I_1 + I_2 + I_3 + I_4 + I_5 \text{.}
\end{align}

By \eqref{eq.pcp_nonstatic}, we have
\begin{equation}
\label{eq.pseudoconvex_1} I_2 \geq 0 \text{,}
\end{equation}
while \eqref{eq.pcp_positive}, along with \eqref{eq.g} and the identity $g ( E_X, E_Y ) = \delta_{X Y}$, implies
\begin{align}
\label{eq.pseudoconvex_2} I_1 &\geq K f^{n - 2} \rho^2 \left[ | \theta ( V ) |^2 + \sum_{ X, Y = 1 }^{n - 1} \theta ( E_X ) \theta ( E_Y ) \cdot E_X^A E_Y^B \gm_{AB} \right] \\
\notag &\geq K f^{n - 2} \rho^2 \left\{ | \theta ( V ) |^2 + [ 1 + \mc{O}_0 ( \rho^2 ) ] \sum_{ X = 1 }^{n - 1} | \theta ( E_X ) |^2 \right\} \text{.}
\end{align}
Observe that the remaining terms, which are errors, satisfy
\begin{align}
\label{eq.pseudoconvex_3} I_3 &\geq \mc{O}_0 ( f^{n - 1} \rho^2 ) \cdot \left[ | \theta ( V ) |^2 + \sum_{ X = 1 }^{n - 1} | \theta ( E_X ) |^2 \right] \text{,} \\
\notag I_5 &\geq \mc{O}_0 ( f^n ) \cdot | \theta ( N ) |^2 + \mc{O}_0 ( f^n \rho^2 ) \cdot \left[ | \theta ( V ) |^2 + \sum_{ X = 1 }^{n - 1} | \theta ( E_X ) |^2 \right] \text{.}
\end{align}
Combining \eqref{eq.pseudoconvex_0}-\eqref{eq.pseudoconvex_3} yields our desired inequality \eqref{eq.pseudoconvex}.

\subsubsection{Some Examples} \label{sec:pseudo_examples}

Recall that as an FG-aAdS segment, AdS spacetime (and more generally, the Kerr-AdS family for $n \geq 3$, after a change of coordinates from the usual Boyer-Lindquist coordinates, cf.~\cite{Henneaux}) have the expansion \eqref{eq.aads_ads}.
We now check when the pseudoconvexity property of Definition \ref{def.pcp} is satisfied.

Since $\gm$ is static, we can take $\xi = 0$.
Note that:
\begin{itemize}
\item \eqref{eq.pcp_nonstatic} is trivially satisfied.

\item $\mu = T^{-1}$ from \eqref{eq.mu}, hence ${\bf Q}_{0, \zeta}$ from \eqref{eq.Q_inf} is given by
\begin{equation}
\label{eq.Q_inf_ads} {\bf Q}_{0, \zeta} = \left( \frac{1}{2} - \frac{1}{T} + \zeta \right) dt^2 + \left( \frac{1}{2} - \zeta \right) \ga \text{.}
\end{equation}
\end{itemize}
Observe that one can find $\zeta$ such that \eqref{eq.Q_inf_ads} is positive-definite if and only if $T > 1$.
In other words, the pseudoconvexity property is satisfied for AdS and Kerr-AdS spacetimes if and only if we consider a segment with time length strictly greater than $\pi$.
This confirms the equivalent results on AdS spacetime established in \cite{hol_shao:uc_ads}.

Moreover, one can now easily construct a large class of examples satisfying the pseudoconvexity criterion by taking (static or nonstatic) perturbations of $\gm$ and $\gs$ from \eqref{eq.aads_ads}.
In particular, for a small enough perturbation, there exists an $\varepsilon > 0$ such that Definition \ref{def.pcp} is satisfied for $\xi = \varepsilon$ and for $T > 1 + \varepsilon$---that is, a time length slightly greater than $( 1 + \varepsilon ) \pi$.

Finally, in any setting for which $\gm$ is static and $( \mc{M}, g )$ is Einstein-vacuum, we can directly relate the pseudoconvexity condition with positive curvature of the level sets of $t$ on the conformal boundary $\mc{I}$.
See Appendix \ref{sec.vacuum} for details.
\subsection{The Carleman Estimate} \label{sec.proof_carleman}

In this section, we prove the following Carleman estimate on admissible FG-aAdS segments:

\begin{theorem} \label{thm.carleman}
Consider an $(n+1)$-dimensional admissible FG-aAdS segment
\[
( \mc{M}, g ) \text{,} \qquad \mc{M} = ( 0, \rho_\ast ) \times ( 0, T \pi ) \times \mc{S} \text{,}
\]
and suppose the pseudoconvexity property holds on $\mc{I}$, with associated parameters $K$, $\xi$, $\zeta$.
Fix also an integer $l \geq 0$, along with constants $p, \kappa \in \R$ satisfying
\begin{equation}
\label{eq.p_kappa} 0 < p < 1 \text{,} \qquad \kappa \geq \frac{n - 1}{2} \text{.}
\end{equation}
In addition, fix constants $0 < \rho_0 \ll f_0 \ll_{g, l, p, K} 1$, and define the region
\begin{equation}
\label{eq.Omega} \Omega_{ f_0, \rho_0 } := \{ f < f_0 \text{, } \rho > \rho_0 \} \text{.}
\end{equation}

Then, there exist constants $C, \mc{C} > 0$, depending on $g$, $p$, and $K$, such that for any $\sigma \in \R$ and $\lambda \in [1 + \kappa, \infty)$, and for any $\phi \in \Gamma \ul{T}^0_l \mc{M}$ such that
\begin{itemize}
\item $\phi$ has compact support on every level set of $( \rho, t )$, and

\item both $\phi$ and $\nabla \phi$ vanish on $\{ f = f_0 \}$,
\end{itemize}
the following inequality holds:
\begin{align}
\label{eq.carleman} &\int_{ \Omega_{ f_0, \rho_0 } } f^{n - 2 - 2 \kappa} e^\frac{ - 2 \lambda f^p }{p} f^{-p} | ( \Box + \sigma ) \phi |^2 \\
\notag &\qquad + \mc{C} \lambda ( \lambda^2 + | \sigma | ) \int_{ \{ \rho = \rho_0 \} } [ | \nabla_t ( \rho^{ - \kappa } \phi ) |^2 + | \nabla_\rho ( \rho^{ - \kappa } \phi ) |^2 + | \rho^{- \kappa - 1} \phi |^2 ] d \gm \\
\notag &\quad \geq C \lambda \int_{ \Omega_{ f_0, \rho_0 } } f^{n - 2 - 2 \kappa} e^\frac{ - 2 \lambda f^p }{p} ( \rho^4 | \nabla_t \phi |^2 + \rho^4 | \nabla_\rho \phi |^2 + \rho^2 | \nasla \phi |^2 ) \\
\notag &\quad \qquad + \lambda [ \kappa^2 - ( n - 2 ) \kappa + \sigma - (n - 1) ] \int_{ \Omega_{ f_0, \rho_0 } } f^{n - 2 - 2 \kappa} e^\frac{ - 2 \lambda f^p }{p} | \phi |^2 \\
\notag &\quad \qquad + C \lambda^3 \int_{ \Omega_{ f_0, \rho_0 } } f^{n - 2 - 2 \kappa} e^\frac{ - 2 \lambda f^p }{p} f^{2p} | \phi |^2 \text{.}
\end{align}
\end{theorem}
\begin{figure}
\[
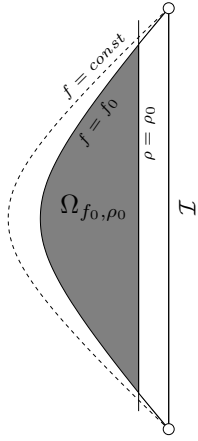
\]
\caption{The region $\Omega_{f_0,\rho_0}$ where the estimate (\ref{eq.carleman}) holds. In applications, the limit $\rho_0 \rightarrow 0$ is taken.}
\end{figure}
\begin{remark}
Note the constant $\kappa^2 - ( n - 2 ) \kappa + \sigma - (n - 1)$ in \eqref{eq.carleman} is non-negative if $\kappa$ is chosen as in \eqref{eq.ucp_vanish}, which will be the case in applications of \eqref{eq.carleman}.
\end{remark}

\begin{remark}
The compact support assumption for $\phi$ in Theorem \ref{thm.carleman} is required for integrations by parts within its proof.
We note that by standard methods, this compact support can be replaced by weaker vanishing or integrability conditions.
Also, in the main cases of interest, $\mc{S}$ will be compact, and this assumption can be ignored.
\end{remark}

In the remainder of this subsection, we prove Theorem \ref{thm.carleman}.
The proof is mostly analogous to the corresponding proof in \cite{hol_shao:uc_ads}, with the main difference being that we must also account for $f$ and $\eta$ not being everywhere smooth.

As in \cite{hol_shao:uc_ads}, the derivation of \eqref{eq.carleman} revolves around multiplier, or vector field, estimates for the \emph{conjugated} wave operator \eqref{eq.L}, with the goal being to obtain positivity in the bulk spacetime terms.
For this purpose, we choose our multiplier as in \eqref{eq.Sw} (i.e., the gradient of $f$, reweighted and modified by a zero-order term).
The desired positivity of the bulk terms then follows from two specific considerations:
\begin{itemize}
\item This particular choice of multiplier (including the zero-order modification $h_{\xi, \zeta}$) implies that derivatives of $\phi$ along level sets of $f$ arise only from quadratic terms involving $\pi_{\xi, \zeta}$.
These are positive due to the pseudoconvexity property and Theorem \ref{thm.pseudoconvex}; see \eqref{eq.set_pi}.

\item The positivity of the remaining terms involving $\phi$ and the normal derivative of $\phi$ follow from an additional freedom: the choice of reparametrization of the $f$-foliation.
For this, we show that $F$, defined below in \eqref{eq.F}, suffices.\footnote{The leading-order term $\kappa \cdot \log f$ of $F$ is particularly essential, as it is ultimately manifested in the specific rate of vanishing \eqref{eq.ucp_vanish} in the ensuing local unique continuation property.}
\end{itemize}

The only significant departure from \cite{hol_shao:uc_ads} is the failure of $f$ and $\eta$ to be smooth here.
As a result, extra care must be taken in the integral estimates in Section \ref{sec.proof_carleman_integral} in order to handle extra terms arising from this lack of smoothness.

\subsubsection{Pointwise Estimates}

Analogous to \cite{hol_shao:uc_ads}, we define the following:
\begin{itemize}
\item We first construct the Carleman weight for our estimate:
\begin{equation}
\label{eq.F} F := \kappa \cdot \log f + \lambda p^{-1} f^p \text{,} \qquad \psi := e^{-F} \phi \text{.}
\end{equation}
Let $'$ denote differentiation with respect to $f$, e.g.,
\begin{equation}
\label{eq.F_deriv} F' = \kappa f^{-1} + \lambda f^{-1 + p} \text{,} \qquad F'' = - \kappa f^{-2} - \lambda (1 - p) f^{-2 + p} \text{.}
\end{equation}

\item Recalling $S$ and $w_{\xi, \zeta}$ (see \eqref{eq.S} and \eqref{eq.w}, respectively), we define
\begin{equation}
\label{eq.Sw} S_{\xi, \zeta} \psi := \nabla_S \psi + h_{\xi, \zeta} \psi \text{,} \qquad h_{\xi, \zeta} := f^{n - 3} w_{\xi, \zeta} + \frac{1}{2} \nabla^\alpha S_\alpha \text{.}
\end{equation}

\item We define the conjugated wave operator $\mc{L}$ by
\begin{equation}
\label{eq.L} \mc{L} := e^{-F} ( \Box + \sigma ) e^F \text{.}
\end{equation}

\item Note that the inward unit normal to the level sets of $\rho$ is given by
\begin{equation}
\label{eq.rho_normal} \mc{N} := | \nabla^\alpha \rho \nabla_\alpha \rho |^{- \frac{1}{2}} \grad \rho = \rho \partial_\rho \text{.}
\end{equation}
\end{itemize}
The key step in proving Theorem \ref{thm.carleman} is the following pointwise estimate for $\psi$:

\begin{lemma} \label{thm.psi_est}
There exists $C > 0$, depending on $g$, $p$, $K$, such that on the regions $\Omega_{ f_0, \rho_0 } \cap \mc{M}_\pm$, where $\mc{M}_\pm$ are as defined in \eqref{eq.M_mp}, we have\footnote{Note in particular that all quantities are smooth in $\Omega_{ f_0, \rho_0 } \cap \mc{M}_\pm$.}
\begin{align}
\label{eq.psi_est} \lambda^{-1} f^{n - 2 - p} | \mc{L} \psi |^2 &\geq C \lambda f^{n - 2 + p} | \nabla_N \psi |^2 + C f^{n - 2} \rho^2 ( | \nabla_V \psi |^2 + | \nasla \psi |^2 ) \\
\notag &\qquad + [ \kappa^2 - ( n - 2 ) \kappa + \sigma - (n - 1) ] f^{n - 2} | \psi |^2 \\
\notag &\qquad + C ( \lambda f^{n - 2 + p} + \lambda^2 f^{n - 2 + 2p} ) | \psi |^2 + \nabla^\beta P_\beta \text{,}
\end{align}
where the $1$-form $P$ satisfies, for some $\mc{C} > 0$ depending on $g$ and $p$,
\begin{equation}
\label{eq.psi_est_boundary} P ( \mc{N} ) \leq \mc{C} f^{n - 2} \rho^2 ( | \nabla_t \psi |^2 + | \nabla_\rho \psi |^2 ) + \mc{C} ( \lambda^2 + | \sigma | ) f^{n - 2} | \psi |^2 \text{.}
\end{equation}
\end{lemma}

\begin{proof}
Let $Q$ be the stress-energy tensor for the wave equation with respect to $\psi$:
\begin{equation}
\label{eq.emt} Q_{\alpha\beta} := \nabla_\alpha \psi^I \nabla_\beta \psi_I - \frac{1}{2} g_{\alpha\beta} \nabla^\mu \psi^I \nabla_\mu \psi_I \text{.}
\end{equation}
A direct computation yields that the current
\begin{equation}
\label{eq.P_sharp} P^Q_\beta := Q_{\alpha\beta} S^\alpha + \frac{1}{2} h_{\xi, \zeta} \cdot \nabla_\beta | \psi |^2 - \frac{1}{2} \nabla_\beta h_{\xi, \zeta} \cdot | \psi |^2
\end{equation}
satisfies the identity
\begin{align}
\label{eq.set_2} \nabla^\beta P^Q_\beta &= \Box \psi^I S_{\xi, \zeta} \psi_I - \pi_{\xi, \zeta}^{\alpha \beta} \nabla_\alpha \psi^I \nabla_\beta \psi_I \\
\notag &\qquad - S^\alpha \nabla^\beta \psi_I \mc{R}_{\alpha \beta} \psi^I - \frac{1}{2} \Box h_{\xi, \zeta} \cdot | \psi |^2 \text{.}
\end{align}

The pseudoconvexity property and Theorem \ref{thm.pseudoconvex} imply that
\begin{align}
\label{eq.set_pi} \pi_{\xi, \zeta}^{\alpha \beta} \nabla_\alpha \psi^I \nabla_\beta \psi_I &\geq [ K f^{n - 2} \rho^2 + \mc{O}_0 ( f^{n - 1} \rho^2 ) ] \cdot ( | \nabla_V \psi |^2 + | \nasla \psi |^2 ) \\
\notag &\qquad - [ (n - 1) f^{n - 2} + \mc{O}_0 ( f^n ) ] \cdot | \nabla_N \psi |^2 \text{.}
\end{align}
Next, using \eqref{eq.S_lo}, \eqref{eq.error_deriv}, \eqref{eq.w}, \eqref{eq.Sw}, and the assumption $\zeta = \mc{O} (1)$, we obtain
\begin{equation}
\label{eq.h_est} h_{\xi, \zeta} = \mc{O}_0 ( f^n ) \text{,} \qquad \Box h_{\xi, \zeta} = \mc{O}_0 ( f^n ) \text{.}
\end{equation}
For the curvature term in \eqref{eq.set_2}, we recall \eqref{eq.S_lo} and \eqref{eq.curv_frame} and expand
\begin{align}
\label{eq.set_R} - S^\alpha \nabla^\beta \psi_I \mc{R}_{\alpha \beta} \psi^I &= \mc{O}_0 ( f^{n - 2} ) \left[ \mc{R}_{N V} \psi^I \nabla_V \psi_I - \sum_{ X = 1 }^{n - 1} \mc{R}_{N E_X} \psi^I \nabla_{ E_X } \psi_I \right] \\
\notag &\geq \mc{O}_0 ( f^{n - 1} \rho^2 ) \cdot ( | \nabla_V \psi | + | \nasla \psi | ) | \psi | \\
\notag &\geq \mc{O}_0 ( f^{n - 2} \rho^4 ) \cdot ( | \nabla_V \psi |^2 + | \nasla \psi |^2 ) + \mc{O}_0 ( f^n ) \cdot | \psi |^2 \text{.}
\end{align}
Thus, applying \eqref{eq.set_pi}-\eqref{eq.set_R} to \eqref{eq.set_2} yields
\begin{align}
\label{eq.set} \Box \psi^I S_{\xi, \zeta} \psi_I &\geq \nabla^\beta P^Q_\beta + \mc{O}_0 ( f^n ) \cdot | \psi |^2 - [ (n - 1) f^{n - 2} + \mc{O}_0 ( f^n ) ] | \nabla_N \psi |^2 \\
\notag &\qquad + [ K f^{n - 2} \rho^2 + \mc{O}_0 ( f^{n - 1} \rho^2 ) ] ( | \nabla_V \psi |^2 + | \nasla \psi |^2 ) \text{.}
\end{align}

Next, we expand $\mc{L}$ to obtain
\begin{align}
\label{eq.conj_1}
\mc{L} \psi &= \Box \psi + 2 F' f^{-n + 3} \cdot \nabla_S \psi + \mc{A}_0 \cdot \psi \text{,} \\
\notag \mc{A}_0 &= [ ( F^\prime )^2 + F^{\prime\prime} ] \nabla^\alpha f \nabla_\alpha f + F^\prime \Box f + \sigma \text{.}
\end{align}
Defining the quantities
\begin{equation}
\label{eq.P_flat} P^S_\beta := \frac{1}{2} \mc{A} S_\beta \cdot | \psi |^2 \text{,} \qquad \mc{A} = \mc{A}_0 + 2 F' f^{-n + 3} h_{\xi, \zeta} \text{,}
\end{equation}
contracting \eqref{eq.conj_1} with $S_{\xi, \zeta} \psi$, and applying the product rule, we see that
\begin{align}
\label{eq.conj_2} \mc{L} \psi^I S_{\xi, \zeta} \psi_I &= \nabla^\alpha P^S_\beta + \Box \psi^I S_{\xi, \zeta} \psi_I + 2 F' f^{-n + 3} \cdot | \nabla_S \psi |^2 \\
\notag &\qquad + \left[ h_{\xi, \zeta} \mc{A}_0 - \frac{1}{2} \nabla^\alpha ( S_\alpha \mc{A} ) \right] \cdot | \psi |^2 \text{.}
\end{align}

By \eqref{eq.S_lo} and \eqref{eq.F_deriv}, we have that
\begin{equation}
\label{eq.conj_zero_1} 2 F' f^{-n + 3} \cdot | \nabla_S \psi |^2 = [ 2 \kappa f^{n - 2} + 2 \lambda f^{n - 2 + p} + \lambda \cdot \mc{O}_0 ( f^n ) ] | \nabla_N \psi |^2 \text{.}
\end{equation}
Moreover, using Proposition \ref{thm.f_lo}, \eqref{eq.F_deriv}, and \eqref{eq.h_est}, we see that
\begin{equation}
\label{eq.conj_zero_2} \mc{A} = ( \kappa^2 - n \kappa + \sigma ) + \lambda ( 2 \kappa - n + p ) f^p + \lambda^2 f^{2 p} + \lambda^2 \cdot \mc{O}_0 ( f^2 ) \text{.}
\end{equation}
Applying \eqref{eq.S_lo}, \eqref{eq.f_grad}, \eqref{eq.error_deriv}, and \eqref{eq.h_est}, we also see that
\begin{align}
\label{eq.conj_zero_3} h_{\xi, \zeta} \mc{A}_0 - \frac{1}{2} \nabla^\alpha ( S_\alpha \mc{A} ) &= ( \kappa^2 - n \kappa + \sigma ) f^{n - 2} + \frac{ 2 - 2 p }{ 2 } \lambda ( 2 \kappa - n + p ) f^{n - 2 + p} \\
\notag &\qquad + ( 1 - p ) \lambda^2 f^{n - 2 + 2 p} + \lambda^2 \cdot \mc{O}_0 ( f^n ) \text{.}
\end{align}
Thus, applying \eqref{eq.set} and \eqref{eq.conj_zero_1}-\eqref{eq.conj_zero_3} to \eqref{eq.conj_2} yields
\begin{align}
\label{eq.conj} \mc{L} \psi^I S_{\xi, \zeta} \psi_I &\geq [ ( 2 \kappa - n + 1 ) f^{n - 2} + 2 \lambda f^{n - 2 + p} + \lambda \cdot \mc{O}_0 (f^n) ] | \nabla_N \psi |^2 \\
\notag &\qquad + [ K f^{n - 2} \rho^2 + \mc{O}_0 ( f^{n - 1} \rho^2 ) ] ( | \nabla_V \psi |^2 + | \nasla \psi |^2 ) \\
\notag &\qquad + \left[ ( \kappa^2 - n \kappa + \sigma ) f^{n - 2} + \frac{2 - p}{2} \lambda ( 2 \kappa - n + p ) f^{n - 2 + p} \right] | \psi |^2 \\
\notag &\qquad + [ (1 - p) \lambda^2 f^{n - 2 + 2p} + \lambda^2 \cdot \mc{O}_0 ( f^n ) ] | \psi |^2 + \nabla^\beta ( P^Q_\beta + P^S_\beta ) \text{.}
\end{align}

Next, fix $b, q \in \R$, and observe that
\begin{align}
\label{eq.hardy_pre} 0 &\leq f^{q - 2} | \nabla^\beta f \nabla_\beta \psi + b f \cdot \psi |^2 \\
\notag &= f^{q - 2} | \nabla^\beta f \nabla_\beta \psi |^2 + [ b^2 f^q - b (q - 1) f^{q - 2} \nabla^\beta f \nabla_\beta f - b f^{q - 1} \Box f ] | \psi |^2 \\
\notag &\qquad + \nabla^\beta ( b f^{q - 1} \nabla_\beta f \cdot | \psi |^2 ) \text{.}
\end{align}
Recalling Proposition \ref{thm.f_lo}, setting $b = \frac{1}{2} (q - n)$, and rearranging terms, \eqref{eq.hardy_pre} yields, for any $q$, a pointwise weighted Hardy-type inequality:
\begin{align}
\label{eq.hardy_ptwise} f^q | \nabla_N \psi |^2 &\geq \frac{1}{4} (q - n)^2 f^q \cdot | \psi |^2 + \mc{O}_0 ( f^{q + 2} ) \cdot ( | \nabla_N \psi |^2 + | \psi |^2 ) \\
\notag &\qquad - \frac{1}{2} (q - n) \nabla^\beta ( f^{q - 1} \nabla_\beta f \cdot | \psi |^2 ) \text{.}
\end{align}

We now apply \eqref{eq.hardy_ptwise} to the terms $f^{n - 2} | \nabla_N \psi |^2$ and $f^{n - 2 + p} | \nabla_N \psi |^2$ in the right-hand side of \eqref{eq.conj}.
Defining in addition the $1$-form
\begin{equation}
\label{eq.P_natural} P^H_\beta := ( 2 \kappa - n + 1 ) f^{n - 3} \nabla_\beta f \cdot | \psi |^2 + \frac{2 - p}{2} \lambda f^{n - 3 + p} \nabla_\beta f \cdot | \psi |^2 \text{,}
\end{equation}
and noting that $2 \kappa - n + 1 \geq 0$ by \eqref{eq.p_kappa}, we see that the above process yields
\begin{align}
\label{eq.hardy_2} \mc{L} \psi^I S_{\xi, \zeta} \psi_I &\geq \nabla^\beta ( P^Q_\beta + P^S_\beta + P^H_\beta ) + [ \lambda f^{n - 2 + p} + \lambda \cdot \mc{O}_0 (f^n) ] | \nabla_N \psi |^2 \\
\notag &\qquad + [ K f^{n - 2} \rho^2 + \mc{O}_0 ( f^{n - 1} \rho^2 ) ] ( | \nabla_V \psi |^2 + | \nasla \psi |^2 ) \\
\notag &\qquad + [ \kappa^2 - ( n - 2 ) \kappa + \sigma - (n - 1) ] f^{n - 2} | \psi |^2 \\
\notag &\qquad + \frac{2 - p}{2} \, \lambda \left( 2 \kappa - n + 1 + \frac{p}{2} \right) f^{n - 2 + p} | \psi |^2 \\
\notag &\qquad + [ (1 - p) \lambda^2 f^{n - 2 + 2p} + \lambda^2 \cdot \mc{O}_0 ( f^n ) ] | \psi |^2 \text{.}
\end{align}

Next, applying the Cauchy-Schwarz inequality, \eqref{eq.S_lo}, and \eqref{eq.h_est}, we have
\begin{align}
\label{eq.hardy_21} \mc{L} \psi^I S_{\xi, \zeta} \psi_I &\leq \lambda^{-1} f^{n - 2 - p} | \mc{L} \psi |^2 + \frac{1}{2} \lambda [ f^{n - 2 + p} + \mc{O}_0 ( f^n ) ] | \nabla_N \psi |^2 \\
\notag &\qquad + \lambda \cdot \mc{O}_0 ( f^n ) \cdot | \psi |^2 \text{,}
\end{align}
which combined with \eqref{eq.hardy_2} yields
\begin{align}
\label{eq.hardy_3} \lambda^{-1} f^{n - 2 - p} | \mc{L} \psi |^2 &\geq \nabla^\beta ( P^Q_\beta + P^S_\beta + P^H_\beta ) \\
\notag &\qquad + \left[ \frac{1}{2} \lambda f^{n - 2 + p} + \lambda \cdot \mc{O}_0 (f^n) \right] | \nabla_N \psi |^2 \\
\notag &\qquad + [ K f^{n - 2} \rho^2 + \mc{O}_0 ( f^{n - 1} \rho^2 ) ] ( | \nabla_V \psi |^2 + | \nasla \psi |^2 ) \\
\notag &\qquad + [ \kappa^2 - ( n - 2 ) \kappa + \sigma - (n - 1) ] f^{n - 2} | \psi |^2 \\
\notag &\qquad + \frac{2 - p}{2} \, \lambda \left( 2 \kappa - n + 1 + \frac{p}{2} \right) f^{n - 2 + p} | \psi |^2 \\
\notag &\qquad + [ (1 - p) \lambda^2 f^{n - 2 + 2p} + \lambda^2 \cdot \mc{O}_0 ( f^n ) ] | \psi |^2 \text{.}
\end{align}
Recalling that $f \ll_{g, l, p, K} 1$ in \eqref{eq.hardy_3} and setting
\begin{equation}
\label{eq.P} P := P^Q + P^S + P^H
\end{equation}
results in the first identity \eqref{eq.psi_est}.
To complete the proof of Lemma \ref{thm.psi_est}, it remains to show that $P$, as defined in \eqref{eq.P}, satisfies \eqref{eq.psi_est_boundary}.

Applying \eqref{eq.rho_normal} to \eqref{eq.P_flat} and \eqref{eq.P_natural}, we see that
\begin{align}
\label{eq.boundary_1} P^S ( \mc{N} ) + P^H ( \mc{N} ) &= \frac{1}{2} \mc{A} f^{n - 3} \rho \partial_\rho f \cdot | \psi |^2 + ( 2 \kappa - n + 1 ) f^{n - 2} \cdot | \psi |^2 \\
\notag &\qquad + \frac{2 - p}{2} \lambda f^{n - 2 + p} \cdot | \psi |^2 \\
\notag &\leq \mc{C} ( \lambda^2 + | \sigma | ) f^{n - 2} \cdot | \psi |^2 \text{,}
\end{align}
for some $\mc{C} > 0$.
Next, for $P^Q$, we expand using \eqref{eq.P_sharp}:
\begin{align}
\label{eq.boundary_20} P^Q ( \mc{N} ) &= \rho \cdot \nabla_S \psi^I \nabla_\rho \psi_I - \frac{1}{2} \rho \cdot g ( S, \partial_\rho ) \cdot \nabla^\mu \psi^I \nabla_\mu \psi_I \\
\notag &\qquad + h_{\xi, \zeta} \cdot \psi^I \nabla_{ \mc{N} } \psi_I - \frac{1}{2} \rho \partial_\rho h_{\xi, \zeta} \cdot | \psi |^2 \text{.}
\end{align}
Using Proposition \ref{thm.g}, \eqref{eq.f_grad}, \eqref{eq.error_deriv}, \eqref{eq.rho_normal}, and \eqref{eq.h_est}, we see that, for some $\mc{C} > 0$,
\begin{align}
\label{eq.boundary_2} P^Q ( \mc{N} ) &\leq \mc{C} f^{n - 2} \rho^2 ( | \nabla_\rho \psi |^2 + | \nabla_t \psi |^2 ) + \mc{C} f^n \cdot | \psi |^2 \text{.}
\end{align}
In both \eqref{eq.boundary_1} and \eqref{eq.boundary_2}, we used that $f \ll_{g, l, p, K} 1$.
Furthermore, similar to \cite{hol_shao:uc_ads}, we used that the leading-order $| \nasla \psi |^2$-term in the expansion of $P^Q ( \mc{N} )$ is negative and hence can be omitted.
Finally, summing \eqref{eq.boundary_1} and \eqref{eq.boundary_2} results in the bound \eqref{eq.psi_est_boundary} and completes the proof of the lemma.
\end{proof}

We now convert Lemma \ref{thm.psi_est} into estimates for $\phi$:

\begin{lemma} \label{thm.phi_est}
There exists $C > 0$, depending on $g$, $p$, $K$, such that
\begin{align}
\label{eq.phi_est} \lambda^{-1} f^{-p} E^p_{\kappa, \lambda} | ( \Box + \sigma ) \phi |^2 &\geq C E^p_{\kappa, \lambda} ( \rho^4 | \nabla_t \phi |^2 + \rho^4 | \nabla_\rho \phi |^2 + \rho^2 | \nasla \phi |^2 ) \\
\notag &\qquad + E^p_{\kappa, \lambda} [ \kappa ^2 - ( n - 2 ) \kappa + \sigma - (n - 1) ] | \phi |^2 \\
\notag &\qquad + C \lambda^2 E^p_{\kappa, \lambda} f^{2p} | \phi |^2 + \nabla^\beta P_\beta \text{,}
\end{align}
on the regions $\Omega_{ f_0, \rho_0 } \cap \mc{M}_\pm$, where
\begin{equation}
\label{eq.carleman_exp} E^p_{\kappa, \lambda} := e^{-2 F} f^{n - 2} = f^{n - 2 - 2 \kappa} e^\frac{ - 2 \lambda f^p }{p} \text{,}
\end{equation}
and where $P$, from \eqref{eq.psi_est_boundary}, satisfies, for some $\mc{C} > 0$ depending on $g$ and $p$,
\begin{equation}
\label{eq.phi_est_boundary} \rho^{-n} \cdot P ( \mc{N} ) \leq \mc{C} [ | \nabla_t ( \rho^{-\kappa} \phi ) |^2 + | \nabla_\rho ( \rho^{-\kappa} \phi ) |^2 ] + \mc{C} ( \lambda^2 + | \sigma | ) | \rho^{-\kappa - 1} \phi |^2 \text{.}
\end{equation}
\end{lemma}

\begin{proof}
From \eqref{eq.psi_est}, we use the largeness of $\lambda$ and the smallness of $f$ to obtain
\begin{align}
\label{eq.psi_est_0} \lambda^{-1} f^{n - 2 - p} | \mc{L} \psi |^2 &\geq C f^{n - 2 + 2p} | \nabla_N \psi |^2 + C f^{n - 2} \rho^2 ( | \nabla_V \psi |^2 + | \nasla \psi |^2 ) \\
\notag &\qquad + [ \kappa^2 - ( n - 2 ) \kappa + \sigma - (n - 1) ] f^{n - 2} | \psi |^2 \\
\notag &\qquad + C \lambda^2 f^{n - 2 + 2p} | \psi |^2 + \nabla^\beta P_\beta \text{.}
\end{align}
By \eqref{eq.f_grad}, \eqref{eq.F_deriv}, and the assumption $\lambda \geq 1 + \kappa$,
\begin{equation} \label{eq.psi_phi_deriv} e^{-2 F} | \nabla_N \phi |^2 = | \nabla_N \psi + F' \nabla_N f \cdot \psi |^2 \lesssim | \nabla_N \psi |^2 + \lambda^2 | \psi |^2 \text{.} \end{equation}
As a result, applying \eqref{eq.carleman_exp} and \eqref{eq.psi_phi_deriv} to \eqref{eq.psi_est_0} yields
\begin{align}
\label{eq.lower_1} \lambda^{-1} f^{-p} E^p_{\kappa, \lambda} | ( \Box + \sigma ) \phi |^2 &\geq C E^p_{\kappa, \lambda} ( \rho^2 | \nabla_V \phi |^2 + \rho^2 | \nasla \phi |^2 + f^{2p} | \nabla_N \phi |^2 ) \\
\notag &\qquad + E^p_{\kappa, \lambda} [ \kappa^2 - ( n - 2 ) \kappa + \sigma - (n - 1) ] | \phi |^2 \\
\notag &\qquad + C \lambda^2 E^p_{\kappa, \lambda} f^{2p} | \phi |^2 + \nabla^\beta P_\beta \text{,}
\end{align}
for some (possibly different) $C > 0$.
Moreover, by \eqref{eq.frame_inv},
\begin{align}
\label{eq.TN_trho} \rho^4 | \nabla_\rho \phi |^2 &\lesssim \rho^2 | \nabla_N \phi |^2 + f^2 \rho^2 | \nabla_V \phi |^2 + \rho^6 | \nasla \phi |^2 \text{,} \\
\notag \rho^4 | \nabla_t \phi |^2 &\lesssim \rho^2 | \nabla_V \phi |^2 + f^2 \rho^2 | \nabla_N \phi |^2 + \rho^6 | \nasla \phi |^2 \text{,}
\end{align}
which when combined with \eqref{eq.lower_1} results in \eqref{eq.phi_est}.

For \eqref{eq.phi_est_boundary}, we first define the shorthand
\begin{equation}
\label{eq.carly_exp} \mc{E}_{p, \lambda} := e^{- \frac{ \lambda f^p }{p} } \leq 1 \text{,}
\end{equation}
and we note that
\begin{equation}
\label{eq.carly_exp_deriv} | \partial_\rho \mc{E}_{p, \lambda} | \lesssim \lambda f^p \rho^{-1} \mc{E}_{p, \lambda} \text{,} \qquad | \partial_t \mc{E}_{p, \lambda} | \lesssim \lambda f^{p + 1} \rho^{-1} \mc{E}_{p, \lambda} \text{.}
\end{equation}
Noting that $2 \kappa \geq n - 1$ by \eqref{eq.p_kappa}, we see from \eqref{eq.carly_exp} and \eqref{eq.carly_exp_deriv} that
\begin{align}
\label{eq.upper_30} f^{n - 2} \rho^{-n} \cdot | \psi |^2 &\leq | \rho^{ - \kappa - 1 } \phi |^2 \text{,} \\
\notag f^{n - 2} \rho^{-n + 2} | \nabla_\rho \psi |^2 &\lesssim | \nabla_\rho ( \rho^{-\kappa} \phi ) |^2 + \lambda^2 | \rho^{-\kappa - 1} \phi |^2 \text{,} \\
\notag f^{n - 2} \rho^{-n + 2} | \nabla_t \psi |^2 &\lesssim | \nabla_t ( \rho^{-\kappa} \phi ) |^2 + \lambda^2 | \rho^{-\kappa - 1} \phi |^2 \text{.}
\end{align}
Combining \eqref{eq.psi_est_boundary} with \eqref{eq.upper_30} yields \eqref{eq.phi_est_boundary}.
\end{proof}

\subsubsection{Integral Estimates} \label{sec.proof_carleman_integral}

It remains to integrate \eqref{eq.phi_est} over $\Omega_{ f_0, \rho_0 }$ and apply the divergence theorem.
Compared to \cite{hol_shao:uc_ads}, the process here is a bit more complex, since we must account for the lack of smoothness at $t = \frac{ \pi T }{2}$:
\begin{itemize}
\item First, we integrate \eqref{eq.phi_est} over $\Omega_{ f_0, \rho_0 } \cap \mc{M}_+$ and apply the divergence theorem.
The term $\nabla^\alpha P_\alpha$ yields boundary terms on $\{ \rho = \rho_0 \}$ and $\{ t = \frac{ \pi T }{2} \}$.

\item We also integrate \eqref{eq.phi_est} over $\Omega_{ f_0, \rho_0 } \cap \mc{M}_-$, which yields corresponding boundary terms on $\{ \rho = \rho_0 \}$ and $\{ t = \frac{ \pi T }{2} \}$.
\end{itemize}
(On the other hand, we do not obtain boundary terms on $\{ f = f_0 \}$, since we assumed both $\phi$ and $\nabla \phi$ vanished on $\{ f = f_0 \}$.)

Summing the two inequalities obtained above, we obtain that
\begin{align}
\label{eq.integral_1} &\lambda^{-1} \int_{ \Omega_{ f_0, \rho_0 } } f^{-p} E^p_{\kappa, \lambda} | ( \Box + \sigma ) \phi |^2 + \int_{ \{ \rho = \rho_0 \} } P ( \mc{N} ) \\
\notag &\quad \geq C \int_{ \Omega_{ f_0, \rho_0 } } E^p_{\kappa, \lambda} ( \rho^4 | \nabla_t \phi |^2 + \rho^4 | \nabla_\rho \phi |^2 + \rho^2 | \nasla \phi |^2 ) \\
\notag &\quad \qquad + \int_{ \Omega_{ f_0, \rho_0 } } E^p_{\kappa, \lambda} [ \kappa^2 - ( n - 2 ) \kappa + \sigma - (n - 1) ] | \phi |^2 \\
\notag &\quad \qquad + C \lambda^2 \int_{ \Omega_{ f_0, \rho_0 } } E^p_{\kappa, \lambda} f^{2p} | \phi |^2 + \lim_{ \tau \rightarrow \frac{ \pi T }{2} + } \int_{ \Omega_{ f_0, \rho_0 } \cap \{ t = \tau \} } P ( \mc{T} ) \\
\notag &\quad \qquad - \lim_{ \tau \rightarrow \frac{ \pi T }{2} - } \int_{ \Omega_{ f_0, \rho_0 } \cap \{ t = \tau \} } P ( \mc{T} ) \text{,}
\end{align}
where $\mc{T}$ denotes the future-pointing ($g$-)unit normal on the level sets of $t$.
Rewriting the integral over $\{ \rho = \rho_0 \}$ in \eqref{eq.integral_1} (with respect to the induced metric) in terms of the volume form from $\gm$ and then applying \eqref{eq.phi_est_boundary}, we have
\begin{align}
\label{eq.integral_2} \int_{ \{ \rho = \rho_0 \} } P ( \mc{N} ) &\leq \mc{C}' \int_{ \{ \rho = \rho_0 \} } \rho^{-n} P ( \mc{N} ) \cdot d \gm \\
\notag &\leq \mc{C} \int_{ \{ \rho = \rho_0 \} } [ | \nabla_t ( \rho^{-\kappa} \phi ) |^2 + | \nabla_\rho ( \rho^{-\kappa} \phi ) |^2 ] d \gm \\
\notag &\qquad + \mc{C} ( \lambda^2 + | \sigma | ) \int_{ \{ \rho = \rho_0 \} } | \rho^{-\kappa - 1} \phi |^2 d \gm \text{,}
\end{align}
for appropriate constants $\mc{C}'$ and $\mc{C}$.

It remains to control the spacelike boundary terms
\begin{equation}
\label{eq.match_Y} \mc{Y} := \lim_{ \tau \rightarrow \frac{ \pi T }{2} + } \int_{ \Omega_{ f_0, \rho_0 } \cap \{ t = \tau \} } P ( \mc{T} ) - \lim_{ \tau \rightarrow \frac{ \pi T }{2} - } \int_{ \Omega_{ f_0, \rho_0 } \cap \{ t = \tau \} } P ( \mc{T} ) \text{,}
\end{equation}
which we show is an error term that can be absorbed by the remaining terms.
To see this, we examine the various terms within $P ( \mc{T} )$; see \eqref{eq.P_sharp}, \eqref{eq.P_flat}, \eqref{eq.P_natural}, and \eqref{eq.P}.
Since $f$ and $\eta$ are $C^2$ (in particular, $w_{\xi, \zeta}$ and $h_{\xi, \zeta}$ are both continuous at $\{ t = \frac{ \pi T }{2} \}$), it follows that all the terms in $P ( \mc{T} )$ in the limits $t \rightarrow \frac{ \pi T }{2} +$ and $t \rightarrow \frac{ \pi T }{2} -$ cancel out, except for the term $- \frac{1}{2} \mc{T} h_{\xi, \zeta} \cdot | \psi |^2$, i.e., the last term in \eqref{eq.P_sharp}.

Since both $h_{\xi, \zeta} |_{ \Omega_{ f_0, \rho_0 } \cap \mc{M}_\pm }$ extend smoothly to $\{ t = \frac{ \pi T }{2} \}$, we can then bound
\begin{align}
\label{eq.integral_30} \mc{Y} &\lesssim \lim_{ \tau \rightarrow \frac{ \pi T }{2} + } \int_{ \Omega_{ f_0, \rho_0 } \cap \{ t = \tau \} } | \mc{T} h_{ \xi, \zeta } | | \psi |^2 + \lim_{ \tau \rightarrow \frac{ \pi T }{2} - } \int_{ \Omega_{ f_0, \rho_0 } \cap \{ t = \tau \} } | \mc{T} h_{ \xi, \zeta } | | \psi |^2 \\
\notag &\lesssim \int_{ \Omega_{ f_0, \rho_0 } \cap \{ t = \frac{ \pi T }{2} \} } f^n | \psi |^2 \text{,}
\end{align}
where in the last step, we applied \eqref{eq.error_deriv} and \eqref{eq.h_est}.
Using that $\fbd \simeq 1$, and hence $f \simeq \rho$, near $\{ t = \frac{ \pi T }{2} \}$, and applying Proposition \ref{thm.g} to expand the volume form on $\{ t = \frac{ \pi T }{2} \}$ in the usual coordinates, we obtain that
\begin{align}
\label{eq.integral_31} \mc{Y} &\lesssim \sum_\varphi \int_{ \Omega_{ f_0, \rho_0 } \cap \{ t = \frac{ \pi T }{2} \} } f^{n - 2} \cdot \rho^2 | \psi |^2 \cdot \rho^{-n} d \rho d x^1 \dots d x^{n - 1} \\
\notag &\lesssim \sum_\varphi \int_{ \Omega_{ f_0, \rho_0 } \cap \{ t = \frac{ \pi T }{2} \} } \rho^2 E^p_{\kappa, \lambda} | \phi |^2 \cdot \rho^{-n} d \rho d x^1 \dots d x^{n - 1} \text{,}
\end{align}
where the summation is over coordinate systems $\varphi = ( x^1, \dots, x^{n - 1} )$ on $\mc{S}$ comprising the definition of admissible AdS segments, and where $E^p_{\kappa, \lambda}$ is as in \eqref{eq.carleman_exp}.

Next, we take small $0 < A \ll T$, apply the fundamental theorem of calculus in the $t$-direction, and convert $d \rho d x^1 \dots d x^{n - 1}$ to a spacetime volume form:
\begin{align}
\label{eq.integral_32} \mc{Y} &\lesssim_A \sum_\varphi \int_{ \Omega_{ f_0, \rho_0 } \cap \{ | t - \frac{ \pi T }{2} | < A \} } \rho^2 \partial_t ( E^p_{\kappa, \lambda} | \phi |^2 ) \cdot \rho^{-n} d \rho d t d x^1 \dots d x^{n - 1} \\
\notag &\qquad + \sum_\varphi \int_{ \Omega_{ f_0, \rho_0 } \cap \{ | t - \frac{ \pi T }{2} | < A \} } \rho^2 E^p_{\kappa, \lambda} | \phi |^2 \cdot \rho^{-n} d \rho d t d x^1 \dots d x^{n - 1} \\
\notag &\lesssim \int_{ \Omega_{ f_0, \rho_0 } \cap \{ | t - \frac{ \pi T }{2} | < A \} } \rho^3 [ \partial_t ( E^p_{\kappa, \lambda} | \phi |^2 ) + E^p_{\kappa, \lambda} | \phi |^2 ] \text{.}
\end{align}
Recalling again that $f \simeq \rho$ here in our region of integration, we conclude
\begin{align}
\label{eq.integral_3} \mc{Y} &\lesssim_A \int_{ \Omega_{ f_0, \rho_0 } \cap \{ | t - \frac{ \pi T }{2} | < A \} } E^p_{\kappa, \lambda} \cdot \rho^3 ( | \phi | | \nabla_t \phi | + \lambda | \phi |^2 ) \\
\notag &\lesssim \int_{ \Omega_{ f_0, \rho_0 } \cap \{ | t - \frac{ \pi T }{2} | < A \} } E^p_{\kappa, \lambda} \cdot ( \rho^5 | \nabla_t \phi |^2 + \lambda \rho | \phi |^2 ) \text{.}
\end{align}

Note the right-hand side of \eqref{eq.integral_3}, and hence $\mc{Y}$, can be absorbed into the first and third terms on the right-hand side of \eqref{eq.integral_1}, as long as $\rho \lesssim f_0$ is small.
Finally, combining this with \eqref{eq.integral_1} and \eqref{eq.integral_2} results in \eqref{eq.carleman} and proves Theorem \ref{thm.carleman}.

\subsection{The Unique Continuation Result} \label{sec.proof_uc} 

Finally, we conclude the section by proving the main unique continuation result of this article.
We begin by first defining the precise local unique continuation property we wish to establish:

\begin{definition} \label{def.ucp}
Let $( \mc{M}, g )$ be an $(n+1)$-dimensional admissible FG-aAdS segment, described in Definitions \ref{def.aads_manifold} and \ref{def.aads}, and consider the wave equation
\begin{equation}
\label{eq.wave} \Box_g \phi + \sigma \phi = \mc{G} ( \phi, \nabla \phi ) \text{,} \qquad \phi \in \Gamma \ul{T}^0_l \mc{M} \text{,}
\end{equation}
where $l \geq 0$, $\sigma \in \R$, and $\mc{G}: \Gamma \ul{T}^0_l \mc{M} \times \Gamma T^0_1 \ul{T}^0_l \mc{M} \rightarrow \Gamma \ul{T}^0_l \mc{M}$.

We say the \emph{local unique continuation property} holds on $( \mc{M}, g )$ for \eqref{eq.wave} iff given any smooth solution $\phi$ of \eqref{eq.wave} which satisfies
\begin{enumerate}
\item $\phi$ has compact support on every level set of $( \rho, t )$.

\item The following vanishing condition holds:\footnote{Note that $\kappa = \beta_+ - 1$ when $\sigma \leq ( n^2 - 1 ) / 4$, where $\beta_+$ is as in \eqref{eq.beta}.}
\begin{align}
\label{eq.ucp_vanish} 0 &= \lim_{ \rho' \searrow 0 } \int_{ \{ \rho = \rho' \} } [ | \nabla_t ( \rho^{ - \kappa } \phi ) |^2 + | \nabla_\rho ( \rho^{ - \kappa } \phi ) |^2  + | \rho^{- \kappa - 1 } \phi |^2 ] d \gm \text{,} \\
\notag \kappa &:= \begin{cases} \frac{n - 2}{2} + \sqrt{ \frac{n^2}{4} - \sigma } & \text{if } \sigma \leq \frac{n^2 - 1}{4} \text{,} \\ \frac{n-1}{2} & \text{if } \sigma > \frac{n^2 - 1}{4} \text{,} \end{cases}
\end{align}

\item The following finiteness condition holds:\footnote{This arises from the fact that no vanishing assumption was imposed for $\nasla \phi$.}
\begin{equation}
\label{eq.ucp_finite} \int_{ \mc{M} } \rho^{2+p} | \nasla \phi |^2 < \infty \text{,} \quad \text{for some fixed $0 < p < 1$,}
\end{equation}
\end{enumerate}
then $\phi$ must vanish in an open neighborhood of the conformal boundary $\mc{I}$. 
\end{definition}

Our main unique continuation result can now be stated as follows:

\begin{theorem} \label{thm.uc_ads}
Consider an $(n+1)$-dimensional admissible FG-aAdS segment
\[
( \mc{M}, g ) \text{,} \qquad \mc{M} = ( 0, \rho_\ast ) \times ( 0, T \pi ) \times \mc{S} \text{,}
\]
and consider on $( \mc{M}, g )$ the wave equation \eqref{eq.wave}, for some $l \geq 0$ and $\sigma \in \R$.  
Furthermore, assume that the following properties hold:
\begin{enumerate}
\item The pseudoconvexity property (see Definition \ref{def.pcp}) holds on $\mc{I}$.

\item There exist $p > 0$ and $C > 0$ such that for any $\phi \in \Gamma \ul{T}^0_l \mc{M}$,
\begin{equation}
\label{eq.uc_ads_rhs} | \mc{G} ( \phi, \nabla \phi) |^2 \leq C \rho^p [ \rho^4 | \nabla_\rho \phi |^2 + \rho^4 | \nabla_t \phi |^2 + \rho^2 | \nasla \phi |^2 + \rho^{2p} | \phi |^2 ] \text{.}
\end{equation}
\end{enumerate}
Then, the local unique continuation property holds on $( \mc{M}, g )$ for \eqref{eq.wave}.
\end{theorem}

\subsubsection{Proof of Theorem \ref{thm.uc_ads}}

This is analogous to the corresponding proof in \cite{hol_shao:uc_ads}, hence we give only an abridged summary.
Assume the hypotheses of Theorem \ref{thm.uc_ads}, and let $\bar{\chi} : [ 0, f_0 ] \rightarrow [0, 1]$ denote a smooth cut-off function satisfying
\begin{equation}
\label{eq.chitoff} \bar{\chi} ( s ) = \begin{cases} 1 & 0 \leq s \leq \frac{f_0}{2} \text{,} \\ 0 & s > \frac{3 f_0}{4} \text{.} \end{cases}
\end{equation}
Letting $\chi := \bar{\chi} \circ f$ and letting $\prime$ denote differentiation with respect to $f$, we have
\begin{align}
\label{eq.wave_cutoff} ( \Box_g + \sigma ) ( \chi \cdot \phi ) &= \chi^\prime ( 2 \nabla^\alpha f \nabla_\alpha \phi + \Box_g f \cdot \phi ) + \chi^{\prime \prime} \nabla^\alpha f \nabla_\alpha f \cdot \phi \\
\notag &\qquad + \chi ( \Box_g \phi + \sigma \phi ) \text{.}
\end{align}
Note that $\chi^\prime$ and $\chi^{\prime\prime}$ are supported in $[ \frac{1}{2} f_0, \frac{3}{4} f_0 ]$.
Letting $\mc{F}$ denote the right-hand side of \eqref{eq.wave_cutoff}, then applying \eqref{eq.f_grad}, \eqref{eq.f_box}, and \eqref{eq.wave}, we compute\footnote{In the case $\frac{1}{2} f_0 \leq f \leq \frac{3}{4} f_0$, we also used that $f \simeq_{ f_0 } 1$.}
\begin{equation}
\label{eq.F_bound} | \mc{F} |^2 \begin{cases} \lesssim_{g, f_0} ( \rho^2 | \nabla_\rho \phi |^2 + \rho^2 | \nabla_t \phi |^2 + \rho^{2+p} | \nasla \phi |^2 + |\phi|^2 ) & \frac{f_0}{2} \leq f \leq \frac{3f_0}{4} \text{,} \\ \lesssim \rho^p ( \rho^4 | \nabla_t \phi |^2 + \rho^4 | \nabla_\rho \phi |^2 + \rho^2 | \nasla \phi |^2 + \rho^{2p} | \phi |^2 ) & 0 \leq f \leq \frac{f_0}{2} \text{.} \end{cases}
\end{equation}

Recall the region $\Omega_{f_0, \rho_0}$ from \eqref{eq.Omega}, for $\rho_0 \ll f_0$, and let
\begin{equation}
\label{eq.Omega_ie} \Omega_i = \Omega_{f_0, \rho_0} \cap \left\{ f < \frac{f_0}{2} \right\} \text{,} \qquad \Omega_e = \Omega_{f_0, \rho_0} \cap \left\{ \frac{f_0}{2} < f < \frac{3 f_0}{4} \right\} \text{.}
\end{equation}
We now apply \eqref{eq.carleman} to $\bar{\phi} := \chi \phi$, with $\kappa$ given by \eqref{eq.ucp_vanish} and $p$ given by \eqref{eq.ucp_finite}.\footnote{In particular, $\kappa$ satisfies \eqref{eq.p_kappa}.}
Recalling \eqref{eq.F_bound}, then the left-hand side $L$ of \eqref{eq.carleman} can then be estimated
\begin{align}
\label{eq.carleman_L} L &\lesssim \int_{ \Omega_e } E^p_{\kappa, \lambda} f^{-p} ( \rho^2 | \nabla_\rho \phi |^2 + \rho^2 | \nabla_t \phi |^2 + \rho^{2+p} | \nasla \phi |^2 + |\phi|^2 ) \\
\notag &\qquad + \int_{ \Omega_i } E^p_{\kappa, \lambda} f^{-p} \rho^p ( \rho^4 | \nabla_t \phi |^2 + \rho^4 | \nabla_\rho \phi |^2 + \rho^2 | \nasla \phi |^2 + \rho^{2p} | \phi |^2 ) \\
\notag &\qquad + \lambda ( \lambda^2 + | \sigma | ) \int_{ \{ \rho = \rho_0 \} } [ | \nabla_t ( \rho^{ - \kappa } \bar{\phi} ) |^2 + | \nabla_\rho ( \rho^{ - \kappa } \bar{\phi} ) |^2 + | \rho^{- \kappa - 1} \bar{\phi} |^2 ] d \gm \\
\notag &:= L_1 + L_2 + L_3 \text{,}
\end{align}
where $E^p_{\kappa, \lambda}$ is as defined in \eqref{eq.carleman_exp}, while the right-hand side $R$ of \eqref{eq.carleman} satisfies
\begin{align}
\label{eq.carleman_R} R &\gtrsim_{g, p, K} \lambda^3 \int_{ \Omega_i } E^p_{\kappa, \lambda} \rho^{2p} | \phi |^2 + \lambda \int_{ \Omega_i } E^p_{\kappa, \lambda} ( \rho^4 | \nabla_t \phi |^2 + \rho^4 | \nabla_\rho \phi |^2 + \rho^2 | \nasla \psi |^2 ) \\
\notag &:= R_1 + R_2 \text{.}
\end{align}
In particular, note that all terms on the right-hand side of \eqref{eq.carleman} are non-negative.

Since $\rho \lesssim f$, we can absorb $L_2$ into $R_1 + R_2$ when $\lambda$ is large enough, so that
\begin{equation}
\label{eq.LR_1} L_1 + L_3 \gtrsim_{g, p, K} R_1 + R_2 \text{.}
\end{equation}
Next, by \eqref{eq.ucp_vanish}, along with the bounds
\[
| \partial_\rho \chi | + | \partial_t \chi | \lesssim_{ g, f_0 } \rho^{-1} \text{,}
\]
we see that $L_3 \rightarrow 0$ when $\rho_0 \searrow 0$.
We can also eliminate the $E^p_{\kappa, \lambda}$'s in \eqref{eq.LR_1}, since
\[
E^p_{\kappa, \lambda} \begin{cases} \leq \left( \frac{f_0}{2} \right)^{n - 2 - 2 \kappa} e^\frac{ - 2 \lambda \left( \frac{f_0}{2} \right)^p }{p} & \frac{f_0}{2} \leq f \leq \frac{3 f_0}{4} \text{,} \\ \geq \left( \frac{f_0}{2} \right)^{n - 2 - 2 \kappa} e^\frac{ - 2 \lambda \left( \frac{f_0}{2} \right)^p }{p} & f < \frac{f_0}{2} \text{.} \end{cases}
\]
As a result, we obtain that for large $\lambda$,
\begin{align}
\label{eq.LR_2} &\int_{ \left\{ \frac{f_0}{2} < f < \frac{3 f_0}{4} \right\} } \left( | \phi |^2 + \rho^2 | \nabla_\rho \phi|^2 + \rho^2 | \nabla_t \phi|^2+ \rho^{2+p} | \slashed{\nabla} \phi|^2  \right) \\
\notag &\quad \gtrsim \lambda \int_{ \left\{ f < \frac{ f_0 }{2} \right\} } \rho^{2p} | \phi |^2 \text{.}
\end{align}

Finally, one can see that \eqref{eq.ucp_vanish} and \eqref{eq.ucp_finite} imply that the left-hand side of \eqref{eq.LR_2} is finite.
As a result, taking $\lambda \nearrow \infty$ yields that $\phi \equiv 0$ on $\{ f < \frac{ f_0 }{2} \}$.

\section{The Static Borderline Case} \label{sec.border}

In this section, we consider a special class of ``borderline" static boundary metrics.
Let $( \mc{M}, g )$ be an admissible FG-aAdS segment, with
\begin{equation}
\label{eq.t_normalize_bs} \mc{S} := \Sph^{n - 1} \text{,} \qquad 0 = T_- < t < T_+ = \pi \text{.}
\end{equation}
Recall we are assuming static boundary data given by
\begin{equation}
\label{eq.met_bs} \gm := - dt^2 + \ga \text{,} \qquad \gs := - \frac{1}{2} ( dt^2 + \ga ) \text{,}
\end{equation}
i.e., the same boundary data as for AdS spacetime itself.\footnote{Recall $\ga$ denotes the round metric $\Sph^{n - 1}$. Note if $( \mc{M}, g )$ is an Einstein-vacuum spacetime, the assumption for $\gs$ in \eqref{eq.met_bs} is implied by the form for $\gm$. As such, this reprepresents the prototypical case in which the pseudoconvexity criterion of Definition \ref{def.pcp} barely fails.}
The main goal of this section is to prove results analogous to those in Section \ref{sec.proof} in the current setting.

In particular, \emph{any $g$ satisfying \eqref{eq.t_normalize_bs} and \eqref{eq.met_bs} fails the pseudoconvexity criterion of Definition \ref{def.pcp} for any $\xi \geq 0$}.
(When $\xi = 0$, the pseudoconvexity criterion holds for any $T_+ > \pi$ but barely fails for $T_+ = \pi$.)
Because of this, we must try to extract pseudoconvexity at one order higher than $\gs$.
Consequently, we assume, in addition to $( \mc{M}, g )$ being an FG-aAdS segment, that $g$ has the refined expansion
\begin{equation}
\label{eq.aads_bs} g = \rho^{-2} d \rho^2 + [ \rho^{-2} \gm_{a b} + \gs_{a b} + \rho \gb_{a b} + \mc{O} ( \rho^2 ) ] dx^a dx^b \text{,}
\end{equation}
i.e., we stipulate an extra (also smooth) third-order term in $\gb$ in the expansion.

\begin{definition} \label{def.aads_bs}
We refer to any FG-aAdS segment $( \mc{M}, g )$ that also satisfies \eqref{eq.t_normalize_bs}-\eqref{eq.aads_bs}, with $\gb$ smooth on $\mc{I}$, as a \emph{borderline FG-aAdS segment}.
\end{definition}

\begin{remark}
The error terms $\mc{O} ( \rho^2 )$ in \eqref{eq.aads_bs} can be replaced by slightly weaker decay, such as $\mc{O} ( \rho^2 \log \rho )$ or $\mc{O} ( \rho^{2 - \delta} )$ for $\delta < 1$.
However, to avoid further cluttering the existing presentation, we do not pursue this here.
\end{remark}

\begin{remark}
For metrics satisfying \eqref{ee}, there is actually no loss in working with the expansion \eqref{eq.aads_bs}.
If $g$ satisfies \eqref{ee} and $n$ is not equal to $2$ or $4$, then the FG-expansion of $g$ is precisely given by \eqref{eq.aads_bs}; see \cite{MR837196}.
On the other hand, if $n=2$ or $n=4$, then generally there is also an $\mc{O} ( \rho^2 \log \rho )$-term present in the expansion \eqref{eq.aads_bs}.
However, in that case, \eqref{ee} and \eqref{eq.met_bs} imply $\gb = 0$, and we would not be able to extract pseudoconvexity at this level (that is, Definition \ref{def.pcp_bs} fails to be satisfied). 
\end{remark}

\subsection{Asymptotic Expansions}

Because of the extra term $\gb$ in \eqref{eq.aads_bs}, we must recompute all the asymptotic expansions obtained in Section \ref{sec.aads}.
Since the computations here are similar to their previous counterparts, we only list the main results here and leave details to the reader.

\subsubsection{Metric Computations}

We begin with expansions for the metric and curvature.
These are analogues in the borderline case of Propositions \ref{thm.g} and \ref{thm.curv}.

\begin{proposition} \label{thm.g_bs}
Let $( \mc{M}, g )$ be a borderline FG-aAdS segment.
\begin{itemize}
\item The components of $g$ satisfy
\begin{align}
\label{eq.g_bs} g_{\rho \rho} = \rho^{-2} \text{,} &\qquad g_{\rho a} = 0 \text{,} \\
\notag g_{tt} = - \rho^{-2} - \frac{1}{2} + \rho \gb_{tt} + \mc{O} ( \rho^2 ) \text{,} &\qquad g_{t A} = \rho \gb_{t A} + \mc{O} ( \rho^2 ) \text{,} \\
\notag g_{AB} = \rho^{-2} \ga_{AB} - \frac{1}{2} \ga_{AB} &+ \rho \gb_{AB} + \mc{O} ( \rho^2 ) \text{.}
\end{align}

\item The dual of $g$ satisfies
\begin{align}
\label{eq.g_inv_bs} g^{\rho \rho} = \rho^2 \text{,} &\qquad g^{\rho a} = 0 \text{,} \\
\notag g^{tt} = - \rho^2 + \frac{1}{2} \rho^4 - \rho^5 \gb_{tt} + \mc{O} ( \rho^6 ) \text{,} &\qquad g^{t A} = \rho^5 \ga^{A B} \gb_{t B} + \mc{O} ( \rho^6 ) \text{,} \\
\notag g^{AB} = \rho^2 \ga^{AB} + \frac{1}{2} \rho^4 \ga^{AB} &- \rho^5 \ga^{AC} \ga^{BD} \gb_{CD} + \mc{O} ( \rho^6 ) \text{.}
\end{align}

\item The Christoffel symbols with respect to coordinate systems $\varphi \in \Xi$ satisfy
\begin{align}
\label{eq.Gamma_bs} \Gamma^\rho_{\rho \rho} = - \rho^{-1} \text{,} &\qquad \Gamma^\rho_{\rho a} = 0 \text{,} \\
\notag \Gamma^\rho_{tt} = - \rho^{-1} - \frac{1}{2} \rho^2 \gb_{tt} + \mc{O} ( \rho^3 ) \text{,} &\qquad \Gamma^\rho_{tA} = - \frac{1}{2} \rho^2 \gb_{tA} + \mc{O} ( \rho^3 ) \text{,} \\
\notag \Gamma^\rho_{AB} = \rho^{-1} \ga_{AB} - \frac{1}{2} \rho^2 \gb_{AB} + \mc{O} ( \rho^3 ) \text{,} &\qquad \Gamma^a_{\rho \rho} = 0 \text{,} \\
\notag \Gamma^t_{\rho t} = - \rho^{-1} + \frac{1}{2} \rho - \frac{3}{2} \rho^2 \gb_{tt} + \mc{O} ( \rho^3 ) \text{,} &\qquad \Gamma^t_{\rho A} = - \frac{3}{2} \rho^2 \gb_{tA} + \mc{O} ( \rho^3 ) \text{,} \\
\notag \Gamma^A_{\rho B} = - \rho^{-1} \delta^A_B - \frac{1}{2} \rho \delta^A_B &+ \frac{3}{2} \rho^2 \ga^{AC} \gb_{CB} + \mc{O} ( \rho^3 ) \text{,} \\
\notag \Gamma^A_{\rho t} = \frac{3}{2} \rho^2 \ga^{AB} \gb_{tB} &+ \mc{O} ( \rho^3 ) \text{.}
\end{align}
In addition, with $\mathring{\Gamma}$ denoting the Christoffel symbols for $\gm$, we have
\begin{equation}
\label{eq.Gamma_t_bs} \Gamma^a_{b c} = \mathring{\Gamma}^a_{b c} + \mc{O} ( \rho^3 ) \text{,} \qquad \Gamma^t_{a b} = \mc{O} ( \rho^3 ) \text{,} \qquad \Gamma^a_{t b} = \mc{O} ( \rho^3 ) \text{.}
\end{equation}

\item With respect coordinate systems $\varphi \in \Xi$, we have for any $\phi \in \Gamma \ul{T}^0_l \mc{M}$ that
\begin{equation}
\label{eq.curv_bs} | \mc{R}_{\rho a} \phi | \lesssim_{g, l} \rho^2 | \phi | \text{,} \qquad | \mc{R}_{t A} \phi | \lesssim_{g, l} \rho^3 | \phi | \text{,} \qquad | \mc{R}_{A B} \phi | \lesssim_{g, l} | \phi | \text{.}
\end{equation}
\end{itemize}
\end{proposition}

\subsubsection{The Modified Foliation}

Next, we define the analogues of \eqref{eq.f} and \eqref{eq.S}:

\begin{definition} \label{def.f_bs}
Define the quantities $\fb$, $\fbdb$ and $\Sb$ by
\begin{equation}
\label{eq.f_bs} \fb := \fbdb^{-1} \rho \left( 1 - \frac{1}{4} \rho^2 \right) \text{,} \qquad \fbdb := \sin t \text{,} \qquad \Sb := \fb^{n - 3} \grad \fb \text{.}
\end{equation}
\end{definition}

The goal once again is to show that the level sets of $\fb$ are pseudoconvex.
However, this pseudoconvexity will now be discerned from $\gb$ rather than $\gs$.

\begin{remark}
Notice that $\fb$ contains an extra term $- \frac{1}{4} \fbdb^{-1} \rho^3$ which has no analogue in \eqref{eq.f}.
This is required in order to cancel certain terms so that the leading-order terms in the pseudoconvexity computations are those associated with $\gb$.
In particular, these leading-order terms contain one extra power of $\rho$ compared to their analogues in Section \ref{sec.aads}, thus extra care must be taken in order to see them.
\end{remark}

\begin{remark}
In contrast to \eqref{eq.psi}, the function $\fbdb$ here is smooth.
Thus, like in \cite{hol_shao:uc_ads}, here we can avoid difficulties arising from the non-smoothness of $\fbd$ and $f$.
\end{remark}

\begin{remark}
Note that $\rho \lesssim \fb$, but only when $\rho \ll 1$.
\end{remark}

\begin{proposition} \label{thm.f_bs}
Let $( \mc{M}, g )$ be an FG-aAdS segment satisfying \eqref{eq.t_normalize_bs}-\eqref{eq.aads_bs}.
Then, whenever $\rho \ll_g 1$, the following asymptotic expansions hold:
\begin{itemize}
\item The gradient of $\fb$ satisfies
\begin{align}
\label{eq.f_grad_bs} \grad \fb &= \fb \left[ \rho - \frac{1}{2} \rho^3 + \mc{O} ( \rho^5 ) \right] \partial_\rho + \fbdb' \fb^2 \left[ \rho - \frac{1}{4} \rho^3 + \gb_{tt} \rho^4 + \mc{O} ( \rho^5 ) \right] \partial_t \\
\notag &\qquad - \fbdb' f^2 \left[ \ga^{AB} \gm_{tB} \rho^4 + \mc{O} ( \rho^5 ) \right] \cdot \partial_A \text{,} \\
\notag \nabla^\alpha \fb \nabla_\alpha \fb &= \fb^2 [ 1 - ( \fbdb' )^2 \fb^2 ] ( 1 - \rho^2 ) - ( \fbdb' )^2 \fb^4 \rho^3 \gb_{tt} + \mc{O}_0 ( \fb^4 \rho^4 ) \\
\notag &= \fb^2 + \mc{O}_0 ( \fb^4 ) \text{.}
\end{align}

\item The second derivatives of $\fb$ satisfy:
\begin{align}
\label{eq.f_hess_bs} \nabla_{\rho \rho} \fb &= \fb \rho^{-2} [ 1 - 2 \rho^2 + \mc{O} ( \rho^4 ) ] \text{,} \\
\notag \nabla_{\rho t} \fb &= - \fbdb' \fb^2 \rho^{-2} \left[ 2 - \frac{1}{2} \rho^2 + \frac{3}{2} \rho^3 \gb_{tt} + \mc{O} ( \rho^4 ) \right] \text{,} \\
\notag \nabla_{t t} \fb &= \fb \rho^{-2} \left\{ 1 + 2 ( \fbdb' )^2 \fb^2 + \left[ \frac{1}{2} + ( \fbdb' )^2 \fb^2 \right] \rho^2 + \frac{1}{2} \gb_{tt} \rho^3 + \mc{O}_0 ( \fb \rho^3 ) \right\} \text{,} \\
\notag \nabla_{A B} \fb &= - \fb \rho^{-2} \left[ \ga_{AB} - \frac{1}{2} \rho^2 \ga_{AB} - \frac{1}{2} \rho^3 \gb_{AB} + \mc{O}_0 ( \rho^4 ) \right] \text{,} \\
\notag \nabla_{\rho A} \fb &= - \fbdb' \fb^2 \rho^{-2} \left[ \frac{3}{2} \rho^3 \gb_{tA} + \mc{O} ( \rho^4 ) \right] \text{,} \\
\notag \nabla_{t A} \fb &= \fb \rho^{-2} \left[ \frac{1}{2} \rho^3 \gb_{tA} + \mc{O}_0 ( \fb \rho^4 ) \right] \text{.}
\end{align}

\item Furthermore, $\fb$ satisfies
\begin{equation}
\label{eq.f_box_bs} \Box \fb = - (n - 1) \fb + \mc{O}_0 ( \fb^3 ) \text{.}
\end{equation}
\end{itemize}
\end{proposition}

\subsubsection{Adapted Frames}

We now define the corresponding frames adapted to level sets of $\fb$, and we subsequently list their key properties:

\begin{definition} \label{def.frame_bs}
We define $( \Nb, \Vb, \Eb_1, \dots, \Eb_{n-1} )$ as follows:
\begin{itemize}
\item Let $( \Eb_1, \dots, \Eb_n )$ denote local orthonormal frames on level sets of $(\rho, t)$:
\begin{equation}
\label{eq.frame_E_bs} \Eb_X := \rho \Eb_X^A \partial_A \text{,} \qquad \Eb_X^A = \mc{O} ( 1 ) \text{.}
\end{equation}

\item Let $\Nb$ be the inward-pointing unit normal to level sets of $\fb$:
\begin{equation}
\label{eq.frame_N_bs} \Nb := | \nabla^\alpha \fb \nabla_\alpha \fb |^{ - \frac{1}{2} } \grad \fb \text{.}
\end{equation}

\item Let $\Vb$ denote the remaining (future, timelike) component:
\begin{align}
\label{eq.frame_V_bs} \Vb &:= | g ( \check{V}, \check{V} ) |^{- \frac{1}{2}} \check{V} \text{,} \\
\notag \check{V} &:= \left( 1 - \frac{3}{4} \rho^2 \right) \partial_t + \fbdb' \fb \partial_\rho - \sum_{ X = 1 }^{n - 1} g \left( \left( 1 - \frac{3}{4} \rho^2 \right) \partial_t + \fbdb' \fb \partial_\rho, \Eb_X \right) \Eb_X \text{.}
\end{align}
\end{itemize}
\end{definition}

\begin{proposition} \label{thm.frame_bs}
Let $( \mc{M}, g )$ be a borderline FG-aAdS segment.
Then, the frames $( \Nb, \Vb, \Eb_X )$ in \eqref{eq.frame_E_bs}-\eqref{eq.frame_V_bs} are orthonormal.
Furthermore, whenever $f, \rho \ll_g 1$:
\begin{itemize}
\item $\Nb$ and $\Vb$ have asymptotic expansions
\begin{align}
\label{eq.frame_NV_bs} \Nb &= [ 1 - ( \fbdb' )^2 \fb^2 ]^{ - \frac{1}{2} } [ \rho + \mc{O}_0 ( \fb \rho^4 ) ] \partial_\rho \\
\notag &\qquad + [ 1 - ( \fbdb' )^2 \fb^2 ]^{ - \frac{1}{2} } \fbdb' \fb \left[ \rho + \frac{1}{4} \rho^3 + \gb_{tt} \rho^4 + \mc{O}_0 ( \fb \rho^4 ) \right] \partial_t \\
\notag &\qquad + [ 1 - ( \fbdb' )^2 \fb^2 ]^{ - \frac{1}{2} } \fbdb' \fb [ \ga^{AB} \gb_{tB} \rho^4 + \mc{O}_0 ( \rho^5 ) ] \partial_A \text{,} \\
\notag \Vb &= [ 1 - ( \fbdb' )^2 \fb^2 - \rho^2 - \gb_{tt} \rho^3 + \mc{O} ( \rho^4 ) ]^{ - \frac{1}{2} } \\
\notag &\qquad \cdot \left[ \left( \rho - \frac{3}{4} \rho^3 \right) \partial_t + \fbdb' \fb \rho \partial_\rho - \rho^4 E_X^A E_X^B \gb_{t A} \partial_B + \mc{O} ( \rho^5 ) \cdot \partial_B \right] \text{.}
\end{align}

\item The following inversion formulas hold:
\begin{align}
\label{eq.frame_inv_bs} [ 1 - ( \fbdb' )^2 \fb^2 ]^\frac{1}{2} \rho \partial_\rho &= [ 1 + \mc{O}_0 ( \rho^2 ) ] \Nb - [ \fbdb' \fb + \mc{O}_0 ( \rho^2 ) ] \Vb + \sum_{ X = 1 }^{ n - 1 } \mc{O}_0 ( \rho^2 ) \cdot \Eb_X \text{,} \\
\notag [ 1 - ( \fbdb' )^2 \fb^2 ]^\frac{1}{2} \rho \partial_t &= [ 1 + \mc{O}_0 ( \rho^2 ) ] \Vb - [ \fbdb' \fb + \mc{O}_0 ( \rho^2 ) ] \Nb + \sum_{ X = 1 }^{ n - 1 } \mc{O}_0 ( \rho^2 ) \cdot \Eb_X \text{.}
\end{align}

\item For any $\phi \in \Gamma \ul{T}^0_l \mc{M}$, we have that
\begin{equation}
\label{eq.curv_frame_bs} | \mc{R}_{\Nb \Vb} \phi | \lesssim_{g, l} \rho^4 | \phi | \text{,} \qquad | \mc{R}_{\Nb \Eb_X} \phi | \lesssim_{g, l} \rho^4 | \phi | \text{.}
\end{equation}
\end{itemize}
\end{proposition}

\subsection{The Pseudoconvexity Criterion}

We now define the analogue of Definition \ref{def.pcp}, the pseudoconvexity property, in our current borderline case.
Note that Definition \ref{def.pcp} itself now barely fails to apply, hence the name \emph{borderline}.\footnote{In particular, at best, one can only find $\zeta$ such that $K = 0$ in \eqref{eq.pcp_positive}.}
Our refined pseudoconvexity condition below looks for positivity at one order higher---indeed, we will show pseudoconvexity can be extracted precisely from the positivity of the new term $\gb$ in the metric expansion \eqref{eq.aads_bs}.

\begin{definition} \label{def.pcp_bs}
We say that the \emph{borderline pseudoconvexity property} holds on $\mc{I}$ iff there exists $K > 0$ and $\zeta \in C^\infty ( \mc{M} )$ such that:\footnote{Since $\gm$ is static, we no longer require the parameter $\xi$ in Definition \ref{def.pcp}.}
\begin{enumerate}
\item $\zeta = \mc{O} ( 1 )$.

\item For any vector field $X := X^t \partial_t + X^A \partial_A$ on $\mc{I}$, the tensor field
\begin{equation}
\label{eq.Q_inf_bs} \hat{\mathbf{Q}}_\zeta := - \frac{3}{2} \gb - \zeta \gm \text{,}
\end{equation}
satisfies the positivity property
\begin{equation}
\label{eq.pcp_positive_bs} \hat{\mathbf{Q}}_\zeta ( X, X ) \geq K [ (X^t)^2 + \gm_{AB} X^A X^B ] \text{.}
\end{equation}
\end{enumerate}
\end{definition}

As before, we show that the borderline pseudoconvexity property implies that the level sets of $\fb$ are pseudoconvex, although the pseudoconvexity degenerates at one order higher than in Theorem \ref{thm.pseudoconvex} as one goes to $\mc{I}$.

\begin{definition} \label{def.dft_bs}
Given $\zeta \in C^\infty ( \mc{M} )$, we let
\begin{equation}
\label{eq.w_bs} \wb_\zeta := \fb + \fb \rho^3 \zeta \text{,} \qquad \dftb_{\xi, \zeta} := - ( \nabla \Sb + \fb^{n-3} \wb_\zeta \cdot g ) \text{.}
\end{equation}
\end{definition}

Here, $\wb_\zeta$ and $\dftb_{\xi, \zeta}$ serve the same purpose as $w_{\xi, \zeta}$ and $\pi_{\xi, \zeta}$ in Definition \ref{def.dft}.
Analogous to Section \ref{sec.proof_pseudo}, the positivity of $\dftb_{\xi, \zeta}$ in directions tangent to the level sets of $\fb$ implies the pseudoconvexity of these level sets.
Again, the choice of $\wb_\zeta$ corresponds to the conformal invariance of pseudoconvexity and acts as an additional degree of freedom for establishing positivity.
This positivity, and its underlying pseudoconvexity, is manifested in the following theorem and its proof.

\begin{theorem} \label{thm.pseudoconvex_bs}
Let $( \mc{M}, g )$ be a borderline FG-aAdS segment, and suppose the borderline pseudoconvexity property holds at $\mc{I}$, with $K$ and $\zeta$ being the parameters from Definition \ref{def.pcp_bs}.
Then, for any $1$-form $\theta$ on $\mc{M}$, we have when $\rho, f \ll_g 1$ that
\begin{align}
\label{eq.pseudoconvex_bs} \dftb_{\xi, \zeta}^{\alpha \beta} \theta_\alpha \theta_\beta &\geq [ K \fb^{n-2} \rho^3 + \mc{O}_0 ( \fb^{n - 1} \rho^3 ) ] \left( | \theta ( \Vb ) |^2 + \sum_{ X = 1 }^{ n-1 } | \theta ( \Eb_X ) |^2 \right) \\
\notag &\qquad - [ (n - 1) \fb^{n-2} + \mc{O}_0 ( \fb^n ) ] | \theta ( \Nb ) |^2 \text{.}
\end{align}
\end{theorem}

\begin{proof}
The proof proceeds similarly to that of Theorem \ref{thm.pseudoconvex}.
The main step is to use \eqref{eq.f_hess_bs}, \eqref{eq.frame_E_bs}, and \eqref{eq.frame_NV_bs} in order to expand $\nabla^2 \fb$ with respect to the orthonormal frames $( \Nb, \Vb, \Eb_1, \dots, \Eb_{n - 1} )$.
From this, we obtain
\begin{align}
\label{eq.Q_bs} \dftb_\zeta (\Vb, \Vb) &= - \fb^{n - 2} \rho^3 \left( \frac{3}{2} \gb_{tt} + \zeta \gm_{tt} \right) + \mc{O}_0 ( \fb^{n - 1} \rho^3 ) \text{,} \\
\notag \dftb_\zeta (\Vb, \Eb_X) &= - \fb^{n - 2} \rho^3 E^A_X \left( \frac{3}{2} \gb_{t A} + \zeta \gm_{t A} \right) + \mc{O}_0 ( \fb^{n - 1} \rho^3 ) \text{,} \\
\notag \dftb_\zeta (\Eb_X, \Eb_Y) &= - \fb^{n - 2} \rho^3 E_X^A E_Y^B \left( \frac{3}{2} \gb_{A B} + \zeta \gm_{A B} \right) + \mc{O}_0 ( \fb^{n - 2} \rho^4 ) \text{,} \\
\notag \dftb_\zeta (\Nb, \Nb) &= - (n - 1) \fb^{n - 2} + \mc{O}_0 ( \fb^n ) \text{,} \\
\notag \dftb_\zeta (\Nb, \Vb) &= \mc{O}_0 ( \fb^{n - 1} \rho^2 ) \text{,} \\
\notag \dftb_\zeta (\Nb, \Eb_X) &= \mc{O}_0 ( \fb^{n - 1} \rho^3 ) \text{.}
\end{align}
An important technical point is the following: while the leading-order terms in the first three equations in \eqref{eq.Q_bs} (which imply pseudoconvexity) are more degenerate than in the non-static case, the cross-terms $\dftb_\zeta (\Nb, \Vb)$ and $\dftb_\zeta (\Nb, \Eb_X)$ (which we wish to be error terms) also improved by a power of $\rho$.

Defining $\check{\theta}$ as in \eqref{eq.theta_var}, we compute, using \eqref{eq.Q_bs}, that
\begin{align}
\label{eq.pseudoconvex_0_bs} \dftb_\zeta^{\alpha \beta} \theta_\alpha \theta_\beta &= \fb^{n - 2} \rho^3 \cdot \hat{\mathbf{Q}}_\zeta^{a b} \check{\theta}_a \check{\theta}_b + \sum_{ \mf{A}, \mf{B} \in \{ \Vb, \Eb_1, \dots, \Eb_{n-1} \} } \mc{O}_0 ( \fb^{n - 1} \rho^3 ) \cdot \theta ( \mf{A} ) \theta ( \mf{B} ) \\
\notag &\qquad - [ (n - 1) \fb^{n - 2} + \mc{O}_0 ( \fb^n ) ] \cdot | \theta ( N ) |^2 \\
\notag &\qquad + \sum_{ \mf{A} \in \{ \Vb, \Eb_1, \dots, \Eb_{n-1} \} } \mc{O}_0 ( \fb^{n - 1} \rho^2 ) \cdot \theta ( N ) \theta ( \mf{A} ) \text{.}
\end{align}
The term containing $\hat{\mathbf{Q}}_\zeta$ can be handled in the same manner as its analogue in the proof of Theorem \ref{thm.pseudoconvex}.
Since we can bound for any $\mf{A} \in \{ \Vb, \Eb_1, \dots, \Eb_{n-1} \}$,
\[
\fb^{n - 1} \rho^2 \cdot \theta ( N ) \theta ( \mf{A} ) \lesssim \fb^{n - 1} \rho \cdot [ \theta ( \mf{A} ) ]^2 + \fb^{n - 1} \rho^3 \cdot [ \theta (N) ]^2 \text{,}
\]
then \eqref{eq.pseudoconvex_bs} follows immediately from \eqref{eq.pcp_positive_bs} and \eqref{eq.pseudoconvex_0_bs}.
\end{proof}

\subsubsection{Examples}

As mentioned in Section \ref{sec.intro}, when $n = 3$, Schwarzschild-AdS spacetimes satisfy \eqref{eq.met_bs} and hence can be considered as borderline FG-aAdS segments.
Furthermore, letting $M \in \R$ denote the mass of the spacetime, we have that
\begin{equation}
\label{eq.aads_ads_bs} \gb = \frac{2}{3} M ( 2 dt^2 + \ga ) \text{.}
\end{equation}

In particular, the borderline pseudoconvexity condition of Definition \ref{def.pcp_bs} fails to hold in the postive-mass case ($M > 0$), as well as for pure AdS spacetime ($M = 0$).
However, the borderline pseudoconvexity condition is satisfied for Schwarzschild-AdS spacetimes with \emph{negative} mass.
As mentioned in the introduction, this should be compared with the asymptotically flat case \cite{alex_schl_shao:uc_inf}, where \emph{positive} mass leads to a foliation of pseudoconvex hypersurfaces near spacelike infinity.

\subsection{The Carleman Estimate}

We now state and prove our corresponding Carleman estimate in the borderline setting:

\begin{theorem} \label{thm.carleman_bs}
Consider an $(n+1)$-dimensional borderline FG-aAdS segment
\[
( \mc{M}, g ) \text{,} \qquad \mc{M} = ( 0, \rho_\ast ) \times ( 0, T \pi ) \times \mc{S} \text{,}
\]
and assume the borderline pseudoconvexity property holds on $\mc{I}$, with parameters $K$, $\zeta$.
Fix also $l \geq 0$, constants $p, \kappa \in \R$ satisfying \eqref{eq.p_kappa}, and $0 < \rho_0 \ll f_0 \ll_{g, l, p, K} 1$.

Then, there exist constants $C, \mc{C} > 0$, depending on $g$, $p$, and $K$, such that for any $\sigma \in \R$ and $\lambda \in [1 + \kappa, \infty)$, and for any $\phi \in \Gamma \ul{T}^0_l \mc{M}$ such that
\begin{itemize}
\item $\phi$ has compact support on every level set of $( \rho, t )$, and

\item both $\phi$ and $\nabla \phi$ vanish on $\{ f = f_0 \}$,
\end{itemize}
the following inequality holds:
\begin{align}
\label{eq.carleman_bs} &\int_{ \Omega_{ f_0, \rho_0 } } \fb^{n - 2 - 2 \kappa} e^\frac{ - 2 \lambda \fb^p }{p} \fb^{-p} | ( \Box + \sigma ) \phi |^2 \\
\notag &\qquad + \mc{C} \lambda ( \lambda^2 + | \sigma | ) \int_{ \{ \rho = \rho_0 \} } [ | \nabla_t ( \rho^{ - \kappa } \phi ) |^2 + | \nabla_\rho ( \rho^{ - \kappa } \phi ) |^2 + | \rho^{- \kappa - 1} \phi |^2 ] d \gm \\
\notag &\quad \geq C \lambda \int_{ \Omega_{ f_0, \rho_0 } } \fb^{n - 2 - 2 \kappa} e^\frac{ - 2 \lambda \fb^p }{p} ( \rho^5 | \nabla_t \phi |^2 + \rho^5 | \nabla_\rho \phi |^2 + \rho^3 | \nasla \phi |^2 ) \\
\notag &\quad \qquad + \lambda [ \kappa^2 - ( n - 2 ) \kappa + \sigma - (n - 1) ] \int_{ \Omega_{ f_0, \rho_0 } } \fb^{n - 2 - 2 \kappa} e^\frac{ - 2 \lambda \fb^p }{p} | \phi |^2 \\
\notag &\quad \qquad + C \lambda^3 \int_{ \Omega_{ f_0, \rho_0 } } \fb^{n - 2 - 2 \kappa} e^\frac{ - 2 \lambda \fb^p }{p} \fb^{2p} | \phi |^2 \text{,}
\end{align}
where $\Omega_{ f_0, \rho_0 } := \{ \fb < f_0 \text{, } \rho > \rho_0 \}$.
\end{theorem}

\begin{proof}
We adopt analogues of the notations established in Section \ref{sec.proof_carleman}:
\begin{itemize}
\item We now conjugate using $\fb$ instead of $f$:
\begin{equation}
\label{eq.F_bs} F := \kappa \cdot \log \fb + \lambda p^{-1} \fb^p \text{,} \qquad \psi := e^{-F} \phi \text{.}
\end{equation}

\item Moreover, we define the borderline analogues of \eqref{eq.Sw}:
\begin{equation}
\label{eq.Sw_bs} \Sb_\zeta \psi := \nabla_{ \Sb } \psi + \hb_\zeta \psi \text{,} \qquad \hb_\zeta := \fb^{n - 3} \wb_\zeta + \frac{1}{2} \nabla^\alpha \Sb_\alpha \text{.} 
\end{equation}

\item Finally, we define $\mc{L}$ and $\mc{N}$ as in \eqref{eq.L} and \eqref{eq.rho_normal}, respectively.
\end{itemize}
The key step is the following analogue of Lemma \ref{thm.psi_est}:
\begin{align}
\label{eq.psi_est_bs} \lambda^{-1} \fb^{n - 2 - p} | \mc{L} \psi |^2 &\geq C \lambda \fb^{n - 2 + p} | \nabla_{ \Nb } \psi |^2 + C \fb^{n - 2} \rho^3 ( | \nabla_{ \Vb } \psi |^2 + | \nasla \psi |^2 ) \\
\notag &\qquad + [ \kappa^2 - ( n - 2 ) \kappa + \sigma - (n - 1) ] \fb^{n - 2} | \psi |^2 \\
\notag &\qquad + C ( \lambda \fb^{n - 2 + p} + \lambda^2 \fb^{n - 2 + 2p} ) | \psi |^2 + \nabla^\beta P_\beta \text{,} \\
\notag P ( \mc{N} ) &\leq \mc{C} \fb^{n - 2} \rho^2 ( | \nabla_t \psi |^2 + | \nabla_\rho \psi |^2 ) + \mc{C} ( \lambda^2 + | \sigma | ) \fb^{n - 2} | \psi |^2 \text{.}
\end{align}
Note that the estimate now holds on all of $\mc{M} \cap \Omega_{ f_0, \rho_0 }$, since $\fb$ is now everywhere smooth.
As before, $C > 0$ depends on $g$, $p$, $K$, while $\mc{C} > 0$ depends on $g$ and $p$.

The proof is now almost identical to that of Lemma \ref{thm.psi_est}.
Besides replacing $f$ by $\fb$, the only other differences are in the powers of $\rho$ in some asymptotic expansions:
\begin{itemize}
\item Here, the pseudoconvexity of the level sets of $\fb$ is more degenerate than that of $f$.
By Theorem \ref{thm.pseudoconvex}, we deduce that the analogue of \eqref{eq.set_pi} is
\begin{align}
\label{eq.set_pi_bs} \dftb_\zeta^{\alpha \beta} \nabla_\alpha \psi^I \nabla_\beta \psi_I &\geq [ K \fb^{n - 2} \rho^3 + \mc{O}_0 ( \fb^{n - 1} \rho^3 ) ] \cdot ( | \nabla_{ \Vb } \psi |^2 + | \nasla \psi |^2 ) \\
\notag &\qquad - [ (n - 1) \fb^{n - 2} + \mc{O}_0 ( \fb^n ) ] \cdot | \nabla_{ \Nb } \psi |^2 \text{.}
\end{align}

\item A similar difference lies in the curvature expansion (in the tensorial case).
More specifically, by \eqref{eq.curv_frame_bs}, the analogue of \eqref{eq.set_R} is
\begin{align}
\label{eq.set_R_bs} - \Sb^\alpha \nabla^\beta \psi_I \mc{R}_{\alpha \beta} \psi^I &= \mc{O}_0 ( \fb^{n - 2} ) \left[ \mc{R}_{ \Nb \Vb } \psi^I \nabla_{ \Vb } \psi_I - \sum_{ X = 1 }^{n - 1} \mc{R}_{ \Nb \Eb_X } \psi^I \nabla_{ \Eb_X } \psi_I \right] \\
\notag &\geq \mc{O}_0 ( \fb^{n - 2} \rho^4 ) \cdot ( | \nabla_{ \Vb } \psi | + | \nasla \psi | ) | \psi | \\
\notag &\geq \mc{O}_0 ( \fb^{n - 2} \rho^4 ) \cdot ( | \nabla_{ \Vb } \psi |^2 + | \nasla \psi |^2 ) + \mc{O}_0 ( \fb^n ) \cdot | \psi |^2 \text{.}
\end{align}
\end{itemize}
The remaining steps, being essentially the same as before, are left to the reader.

From \eqref{eq.psi_est_bs}, we express $\psi$ in terms of $\phi$, and we expand the frame elements $\Nb$, $\Vb$, $\Eb_X$ in terms of coordinate derivatives.
The derivation is identical to the proof of Lemma \ref{thm.phi_est}; the only real difference is that due to the extra power of $\rho$ in \eqref{eq.set_pi_bs} and \eqref{eq.set_R_bs}, we inherit an extra power of $\rho$ in the coordinate derivatives:
\begin{align}
\label{eq.phi_est_bs} \lambda^{-1} \fb^{-p} E^p_{\kappa, \lambda} | ( \Box + \sigma ) \phi |^2 &\geq C \hat{E}^p_{\kappa, \lambda} ( \rho^5 | \nabla_t \phi |^2 + \rho^5 | \nabla_\rho \phi |^2 + \rho^3 | \nasla \phi |^2 ) \\
\notag &\qquad + \hat{E}^p_{\kappa, \lambda} [ \kappa ^2 - ( n - 2 ) \kappa + \sigma - (n - 1) ] | \phi |^2 \\
\notag &\qquad + C \lambda^2 \hat{E}^p_{\kappa, \lambda} \fb^{2p} | \phi |^2 + \nabla^\beta P_\beta \text{,} \\
\notag \rho^{-n} \cdot P ( \mc{N} ) &\leq \mc{C} [ | \nabla_t ( \rho^{-\kappa} \phi ) |^2 + | \nabla_\rho ( \rho^{-\kappa} \phi ) |^2 ] \\
\notag &\qquad + \mc{C} ( \lambda^2 + | \sigma | ) | \rho^{-\kappa - 1} \phi |^2 \text{.}
\end{align}
Here, $C$ and $\mc{C}$ satisfy the same assumptions as above, and
\begin{equation}
\label{eq.carleman_exp_bs} \hat{E}^p_{\kappa, \lambda} := \fb^{n - 2 - 2 \kappa} e^\frac{ - 2 \lambda \fb^p }{p} \text{.}
\end{equation}

The final step is to integrate \eqref{eq.phi_est_bs} and apply the divergence theorem, which results in \eqref{eq.carleman_bs} and completes the proof.
This process is similar to that in Section \ref{sec.proof_carleman_integral} but is simpler since all quantities here are smooth on $\mc{M} \cap \Omega_{ f_0, \rho_0 }$.
Thus, here we do not have to deal with any discontinuities on $\{ t = \frac{ \pi }{2} \}$.
\end{proof}

\subsection{Unique Continuation} \label{sec.proof_borderline}

We conclude this section with the analogue of Theorem \ref{thm.uc_ads} in the static borderline case.

\begin{theorem} \label{thm.uc_ads_bs}
Consider an $(n+1)$-dimensional borderline FG-aAdS segment
\[
( \mc{M}, g ) \text{,} \qquad \mc{M} = ( 0, \rho_\ast ) \times ( 0, \pi ) \times \mc{S} \text{,}
\]
and consider on $( \mc{M}, g )$ the wave equation \eqref{eq.wave}, for some $l \geq 0$ and $\sigma \in \R$.  
Furthermore, assume that the following properties hold:
\begin{enumerate}
\item The borderline pseudoconvexity property (see Definition \ref{def.pcp_bs}) holds on $\mc{I}$.

\item There exist $0 < p < 1$ and $C > 0$ such that for any $\phi \in \Gamma \ul{T}^0_l \mc{M}$,
\begin{equation}
\label{eq.uc_ads_rhs_bs} | \mc{G} ( \phi, \nabla \phi) |^2 \leq C \rho^p [ \rho^5 | \nabla_\rho \phi |^2 + \rho^5 | \nabla_t \phi |^2 + \rho^3 | \nasla \phi |^2 + \rho^{2p} | \phi |^2 ] \text{.}
\end{equation}
\end{enumerate}
Then, the local unique continuation property holds on $( \mc{M}, g )$ for \eqref{eq.wave}.
\end{theorem}

Theorem \ref{thm.uc_ads_bs} is proved in a completely analogous manner as Theorem \ref{thm.uc_ads}, except we now use the borderline Carleman estimate, Theorem \ref{thm.carleman_bs}, in the place of Theorem \ref{thm.carleman}.
The details are left to the reader.

\begin{remark}
By inspecting the proof of Theorem \ref{thm.uc_ads_bs}, one can observe that in the borderline case, the finiteness condition \eqref{eq.ucp_vanish} can be replaced by
\[
\int_{ \mc{M} } \rho^{3+p} | \nasla \phi |^2 < \infty \text{,} \qquad 0 < p < 1 \text{.}
\]
\end{remark}

\appendix

\section{Reduction to the Fefferman-Graham Gauge} \label{sec.fg}

In this appendix, we prove the following theorem, which demonstrates that any admissible aAdS segment can be rewritten as an admissible FG-aAdS segment via an appropriate change of coordinates.

\begin{theorem} \label{thm.fefferman_graham}
Let $( \ul{\mc{M}}, \ul{g} )$ be an admissible aAdS segment, with
\[
\ul{\mc{M}} := ( 0, \ul{\rho}_\ast ) \times \ul{\mc{I}} := ( 0, \ul{\rho}_\ast ) \times ( \ul{T}_-, \ul{T}_+ ) \times \mc{S} \text{,}
\]
and with induced AdS-type boundary $( \ul{\mc{I}}, \gm )$.
Then, the following statements hold:
\begin{enumerate}
\item Given any $\ul{T}_- < T_- < T_+ < \ul{T}_+$, there exists $\rho_\ast > 0$ and an isometry $\Phi$ between an open subset $\ul{\mc{D}}$ of $\ul{\mc{M}}$ and an admissible FG-aAdS segment
\[
( \mc{M}, g ) \text{,} \qquad \mc{M} := ( 0, \rho_\ast ) \times \mc{I} := ( 0, \rho_\ast ) \times ( T_-, T_+ ) \times \mc{S} \text{,}
\]
with induced AdS-type boundary $( \mc{I}, \gm |_{ \mc{I} } )$.

\item Letting $\urho$ and $\rho$ denote the projections to the first components of $\ul{\mc{M}}$ and $\mc{M}$, respectively, then the following comparison holds:\footnote{Much more precise comparisons hold; these can be found within the proof.}
\begin{equation}
\label{eq.fg_rho_est} \rho \simeq \urho \circ \Phi \text{.}
\end{equation}
In particular, $\Phi^{-1}$ maps the conformal boundary $\mc{I}$ of $\mc{M}$ to $\ul{\mc{I}}$.
\end{enumerate}
\end{theorem}

\begin{remark}
In practice, the aAdS segments on which our unique continuation theorems apply are subsets of larger spacetimes.
Thus, the (arbitrarily small) shortening of the time interval in Theorem \ref{thm.fefferman_graham} would not result in any loss of generality.
\end{remark}

The remainder of this appendix is dedicated to the proof of Theorem \ref{thm.fefferman_graham}.
Let $\urho$ and $\ut$ be the usual projections on $\ul{\mc{M}}$, and let $\ul{h} := \urho^2 \ul{g}$.
Recall from \eqref{eq.aads_gen} that
\begin{align}
\label{eq.aads_h} \ul{h} &= ( 1 + \urho^2 \uggs_{ \rho \rho } ) d \urho^2 + ( - d \ut^2 + \uggm_{A B} d \ux^A d \ux^B + \urho^2 \uggs_{a b} d \ux^a d \ux^b ) \\
\notag &\qquad + \uO ( \urho^3 ) \cdot d \ux^\alpha d \ux^\beta \text{.}
\end{align}
A few notational clarifications are in order here:
\begin{itemize}
\item To avoid clutter, we do not underline symbols in superscript and subscript indices.
For $( \ul{\mc{M}}, \ul{g} )$-related quantities (e.g., $\ul{h}$, $\uggm$, $\uggs$), indices are understood to be with respect to $( \urho, \ut, \ux^A )$-coordinates.

\item We define the class $\uO ( \zeta )$ to be as in Definition \ref{def.O_scr}, except now to be with respect to the underlined $( \urho, \ut, \ux^A )$-coordinate systems.
\end{itemize}

Let $\ul{\Lambda}^\mu_{\alpha \beta}$ denote the Christoffel symbols for $\ul{h}$, in the $( \urho, \ut, \ux^A )$-coordinates.
From direct computations using \eqref{eq.aads_h}, we see that
\begin{align}
\label{eq.aads_h_deriv} \ul{\Lambda}^\rho_{ \rho \rho } = \uggs_{ \rho \rho } \urho + \uO ( \urho^2 ) \text{,} &\qquad \ul{\Lambda}^\rho_{ \rho a } = \uO ( \urho^2 ) \text{,} \\
\notag \ul{\Lambda}^\rho_{ a b } = - \ggs_{ab} \ul{\rho} + \uO ( \urho^2 ) \text{,} &\qquad \ul{\Lambda}^a_{ \rho \rho } = \uO ( \urho^2 ) \text{,} \\
\notag \ul{\Lambda}^a_{ \rho b } = \uggm^{a d} \uggs_{b d} \urho + \uO ( \urho^2 ) \text{,} &\qquad \ul{\Lambda}^a_{ b c } = \uO ( 1 ) \text{.}
\end{align}

\subsection{Geodesic Coordinates}

First, note that we can formally extend $\ul{h}$ to $\ul{\rho} \leq 0$ by dropping error terms and defining
\begin{equation}
\label{eq.aads_h_ext} \ul{h} |_{ \urho \leq 0 } := ( 1 + \urho^2 \uggs_{ \rho \rho } ) d \urho^2 + ( - d \ut^2 + \uggm_{A B} d \ux^A d \ux^B + \rho^2 \uggs_{a b} d \ux^a d \ux^b ) \text{.}
\end{equation}
Observe that this extended $\ul{h}$ is $C^2$ in $\urho$ and smooth in the remaining $\ux^a$-coordinates.
Thus, it makes sense to speak of derivatives of quantities ``at $\ul{\mc{I}} = \{ \urho = 0 \}$".

Let $\gamma$ be the family of $\ul{h}$-geodesics beginning on $\urho = 0$ and satisfying the initial conditions $\gamma' |_{ \urho = 0 } = \partial_{ \urho }$.
Let $\uaff$ denote the affine parameters (with respect to $\ul{h}$) of these $\gamma$, with $\uaff = 0$ on $\urho = 0$.
Given coordinates $( \ux^a ) = ( \ut, \ux^A )$ on $\urho = 0$, we define coordinates $( x^a ) := ( t, x^A )$ on the spacetime by transporting the $\ux^a$'s along $\gamma$.

Consider now the coordinates $( \uaff, x^a )$, and note that
\[
\ul{D}_{ \partial_{ \uaff } } \partial_{ \uaff } = 0 \text{,} \qquad \ul{h} ( \partial_{ \uaff }, \partial_{ \uaff } ) = \ul{h} ( \partial_{ \uaff }, \partial_{ \uaff } ) |_{ \urho = 0 } = 1 \text{,}
\]
where $\ul{D}$ is the Levi-Civita connection for $\ul{h}$.
Furthermore, we have
\[
\partial_{ \uaff } [ \ul{h} ( \partial_{ \uaff }, \partial_{ x^a } ) ] = \ul{h} ( \partial_{ \uaff } , D_{ \partial_{ \uaff } } \partial_{ x^a } ) = \ul{h} ( \partial_{ \uaff } , D_{ \partial_{ x^a } } \partial_{ \uaff } ) = 0 \text{.}
\]
Thus, we can write $\ul{h}$ as
\begin{equation}
\label{eq.fg_met_tilde} \ul{h} = d \uaff^2 + \ul{h}_{ a b } d x^a d x^b \text{.}
\end{equation}

In addition, we now restrict ourselves to $t \in ( T_-, T_+ )$, so that the geodesic $\gamma$ emanating from each $P \in \{ \urho = 0 \}$ with $t (P) \in (T_-, T_+)$ exists for some uniform interval $\uaff \in ( - \uaff_0, \uaff_0 )$ while remaining in a compact subset of the extended $\ul{\mc{M}}$.
In particular, in this region, all quantities under consideration will be bounded.

The remaining goal of this subsection is to compare vector fields in the $( \uaff, x^a )$-coordinates with those in the $( \urho, \ux^a )$-coordinates.
For this, we define the coefficients
\begin{equation}
\label{eq.fg_XA} \partial_{ \uaff } := X^\rho \partial_{ \urho } + X^a \partial_{ \ux^a } \text{,} \qquad \partial_{ x^a } = A_a^\rho \partial_{ \urho } + A_a^b \partial_{ \ux^b } \text{.}
\end{equation}
The goal, then, is to control the $X^\alpha$'s and $A_a^\alpha$'s.

\subsubsection{Bounds for the $X^\alpha$'s}

First, we note that
\[
0 = \ul{D}_{ \partial_{ \uaff } } \partial_{ \uaff } = \partial_{ \uaff } ( X^\mu ) \cdot \partial_\mu + X^\alpha X^\beta \ul{\Lambda}^\mu_{ \alpha \beta } \cdot \partial_\mu \text{,}
\]
which expands to a system of differential equations:
\begin{align}
\label{eq.fg_ev} \partial_{ \uaff } X^\rho &= - ( X^\rho )^2 [ \uggs_{\rho \rho} \ul{\rho} + \uO ( \urho^2 ) ] + X^a X^b [ \uggs_{ab} \urho + \uO ( \urho^2 ) ] + 2 X^\rho X^a \cdot \uO ( \urho^2 ) \text{,} \\
\notag \partial_{ \uaff } X^c &= ( X^\rho )^2 \cdot \uO ( \urho^2 ) - 2 X^\rho X^a \cdot [ \uggm^{cd} \uggs_{ad} \urho + \uO ( \urho^2 ) ] - X^a X^b \cdot \ul{\Lambda}^c_{ a b } \text{.}
\end{align}
Moreover, note that by \eqref{eq.aads_h_ext}, the equations \eqref{eq.fg_ev} extend to $\urho \leq 0$, with the $\uO ( \urho^2 )$-terms vanishing on $\urho \leq 0$.
The equations also imply that the $X^\alpha$'s are twice continuously differentiable in $\urho$ in this extended region.

We can now determine the asymptotics of the $X^\alpha$'s at $\uaff = 0$.
By definition,
\begin{equation}
\label{eq.fg_init_0} X^\rho |_{ \uaff = 0 } = 1 \text{,} \qquad X^a |_{ \uaff = 0 } = 0 \text{.}
\end{equation}
Furthermore, the evolution equations \eqref{eq.fg_ev} imply that
\begin{equation}
\label{eq.fg_init_1} \partial_{ \uaff } ( X^\alpha ) |_{ \uaff = 0 } = 0 \text{.}
\end{equation}
For second derivatives, we differentiate \eqref{eq.fg_ev}.
Noting $\partial_{ \uaff } \urho = X^\rho$, we have
\begin{align}
\label{eq.fg_init_2} \partial_{ \uaff }^2 ( X^\rho ) |_{ \uaff = 0 } &= - ( X^\rho )^2 \cdot \uggs_{\rho \rho} X^\rho |_{ \uaff = 0 } + X^a X^b \cdot \uggs_{ab} X^\rho |_{ \uaff = 0 } = - \uggs_{\rho \rho} \text{,} \\
\notag \partial_{ \uaff }^2 ( X^a ) |_{ \uaff = 0 } &= - 2 X^\rho X^a \cdot \uggm^{cd} \uggs_{ad} X^\rho |_{ \uaff = 0 } - \partial_{ \uaff } ( X^a X^b \cdot \ul{\Lambda}^c_{ a b } ) |_{ \uaff = 0 } = 0 \text{.}
\end{align}

Since the $\ux^a$- and $x^a$-coordinates coincide at $\uaff = 0$, we can also use \eqref{eq.fg_ev} to compute some higher derivatives at $\uaff = 0$.
Indeed, given any integer $l > 0$ and arbitrary indices $a_1, \dots, a_l$, we have that
\begin{align}
\label{eq.fg_init_high} \partial_{ x^{a_1} } \dots \partial_{ x^{a_l} } X^\alpha |_{ \uaff = 0 } &= 0 \text{,} \\
\notag \partial_{ \uaff } \partial_{ x^{a_1} } \dots \partial_{ x^{a_l} } X^\alpha |_{ \uaff = 0 } &= 0 \text{,} \\
\notag \partial_{ \uaff }^2 \partial_{ x^{a_1} } \dots \partial_{ x^{a_l} } X^\rho |_{ \uaff = 0 } &= - \partial_{ \ux^{a_1} } \dots \partial_{ \ux^{a_l} } \uggs_{\rho \rho} \text{,} \\
\notag \partial_{ \uaff }^2 \partial_{ x^{a_1} } \dots \partial_{ x^{a_l} } X^a |_{ \uaff = 0 } &= 0 \text{.}
\end{align}

Since we restrict ourselves to a relatively compact region in $\ul{\mc{M}}$, then applying Taylor's theorem along with \eqref{eq.fg_init_0}-\eqref{eq.fg_init_high} yields
\begin{equation}
\label{eq.fg_taylor_est} \left| \partial_{ x^{a_1} } \dots \partial_{ x^{a_l} } \left( X^\rho - 1 + \frac{1}{2} \uggs_{\rho \rho} \uaff^2 \right) \right| \lesssim_l \uaff^3 \text{,} \qquad | \partial_{ x^{a_1} } \dots \partial_{ x^{a_l} } X^a | \lesssim_l \uaff^3
\end{equation}
for any $l > 0$ and indices $a_1, \dots, a_l$.
Since $\partial_{ \uaff } \urho = X^\rho$, then integrating the first inequality in \eqref{eq.fg_taylor_est} with respect to $\uaff$ results in the estimate
\begin{equation}
\label{eq.fg_taylor_rho} \left| \partial_{ x^{a_1} } \dots \partial_{ x^{a_l} } \left( \ul{\rho} - \uaff + \frac{1}{6} \ul{\bar{\mf{g}}}_{\rho \rho} \tilde{\rho}^3 \right) \right| \lesssim_l \uaff^4 \text{.}
\end{equation}

Note also that \eqref{eq.fg_init_1}-\eqref{eq.fg_init_high} imply
\begin{align}
\label{eq.fg_taylor_est_deriv} \left| \partial_{ x^{a_1} } \dots \partial_{ x^{a_l} } \partial_{ \uaff } \left( X^\rho - 1 + \frac{1}{2} \uggs_{\rho \rho} \uaff^2 \right) \right| \lesssim_l \uaff^2 \text{,} &\qquad | \partial_{ x^{a_1} } \dots \partial_{ x^{a_l} } \partial_{ \uaff } X^c | \lesssim_l \uaff^2 \text{,} \\
\notag \left| \partial_{ x^{a_1} } \dots \partial_{ x^{a_l} } \partial_{ \uaff }^2 \left( X^\rho - 1 + \frac{1}{2} \uggs_{\rho \rho} \uaff^2 \right) \right| \lesssim_l \uaff \text{,} &\qquad | \partial_{ x^{a_1} } \dots \partial_{ x^{a_l} } \partial_{ \uaff }^2 X^c | \lesssim_l \uaff \text{.}
\end{align}

Next, by induction, we can take successive $\uaff$-derivatives of \eqref{eq.fg_ev} and control the left-hand side by the right-hand side (which is lower-order), using bounds already obtained in the previous iteration.
(In particular, throughout these differentiations, we recall that $\partial_{ \uaff } ( \urho - \uaff ) = X^\rho - 1$.)
From this process, we obtain
\begin{align}
\label{eq.fg_taylor_est_gen} \left| \partial_{ x^{a_1} } \dots \partial_{ x^{a_l} } \partial_{ \uaff }^k \left( X^\rho - 1 + \frac{1}{2} \uggs_{\rho \rho} \uaff^2 \right) \right| &\lesssim_l \uaff^{3 - k} \text{,} \\
\notag | \partial_{ x^{a_1} } \dots \partial_{ x^{a_l} } \partial_{ \uaff }^k X^c | &\lesssim_l \uaff^{3 - k} \text{,}
\end{align}
for each nonnegative integer $k$.
From \eqref{eq.fg_taylor_est_gen}, we obtain the asymptotic bounds
\begin{align}
\label{eq.fg_X_asymp} X^\rho - 1 + \frac{1}{2} \uggs_{\rho \rho} \uaff^2 &= \mc{O}_\sigma ( \uaff^3 ) \text{,} \\
\notag X^a &= \mc{O}_\sigma ( \uaff^3 ) \text{,} \\
\notag \urho - \uaff + \frac{1}{6} \uggs_{\rho \rho} \uaff^3 &= \mc{O}_\sigma ( \uaff^4 ) \text{,}
\end{align}
where we define $\mc{O}_\sigma ( \zeta )$ as in Definition \ref{def.O_scr}, but with respect to $( \uaff, x^a )$-coordinates.

\subsubsection{Bounds for the $A_a^\alpha$'s}

Observe that the $A_a^\alpha$'s satisfy
\begin{equation}
\label{eq.fg_A_obs} A_a^\rho = \partial_{ x^a } \ul{\rho} \text{,} \qquad A_a^b = \partial_{ x^a } \ul{x}^b \text{,} \qquad \partial_{ \uaff } A_a^b = \partial_{ x^a } X^b \text{.}
\end{equation}
Thus, from \eqref{eq.fg_X_asymp} and \eqref{eq.fg_A_obs}, we conclude that
\begin{equation}
\label{eq.fg_A_asymp} A_a^\rho + \frac{1}{6} \partial_{ x^a } \uggs_{ \rho \rho } \uaff^3 = \mc{O}_\sigma ( \uaff^4 ) \text{,} \qquad A_a^b = \delta_a^b + \mc{O}_\sigma ( \uaff^4 ) \text{.}
\end{equation}

\subsection{The Metric Expansion}

From \eqref{eq.fg_X_asymp} and \eqref{eq.fg_A_asymp}, we conclude that\footnote{From \eqref{eq.fg_X_asymp} and \eqref{eq.fg_A_asymp}, we infer that the classes $\uO ( \zeta )$ and $\mc{O}_\sigma ( \zeta )$ coincide.}
\begin{align}
\label{eq.fg_h_ab} \ul{h} ( \partial_{ x^a }, \partial_{ x^b } ) &= A_a^c A_b^d \cdot \ul{h} ( \partial_{ \ux^c }, \partial_{ \ux^d } ) + A_a^c A_b^\rho \cdot \ul{h} ( \partial_{ \ux^c }, \partial_{ \urho } ) \\
\notag &\qquad + A_a^\rho A_b^d \cdot \ul{h} ( \partial_{ \urho }, \partial_{ \ux^d } ) + A_a^\rho A_b^\rho \cdot \ul{h} ( \partial_{ \urho }, \partial_{ \urho } ) \\
\notag &= \ul{h} ( \partial_{ \ux^a }, \partial_{ \ux^b } ) + \mc{O}_\sigma ( \uaff^4 ) \text{.}
\end{align}
It then follows from \eqref{eq.aads_h}, \eqref{eq.fg_met_tilde}, and \eqref{eq.fg_h_ab} that
\begin{equation}
\label{eq.fg_g_pre} g = \urho^{-2} d \uaff^2 + \ul{\rho}^{-2} [ \uggm_{ab} + \uggs_{ab} \urho^2 + \mc{O}_\sigma ( \uaff^3 ) ] d x^a d x^b \text{.}
\end{equation}

\subsubsection{The Radial Normalization}

We make one final change of variables $\uaff \mapsto \rho$, satisfying that $\rho = 0$ and $\uaff = 0$ coincide at $\ul{\mc{I}}$, and that
\begin{equation}
\label{eq.fg_rho_def} \frac{ d \uaff }{ \urho } = \frac{ d \rho }{ \rho } \text{.}
\end{equation}
Rearranging the above and recalling the last part of \eqref{eq.fg_X_asymp} yields
\begin{equation}
\label{eq.fg_rho_1} \frac{ d ( \log \rho ) }{ d \uaff } = \frac{1}{ \urho } = \frac{1}{ \uaff } + \frac{1}{6} \uggs_{ \rho \rho } \uaff + \mc{O}_\sigma ( \uaff^2 ) \text{,}
\end{equation}
and integrating \eqref{eq.fg_rho_1} results in the relation
\begin{equation}
\label{eq.fg_rhos_pre} \rho = \uaff \cdot e^{ \frac{1}{12} \uggs_{ \rho \rho } \uaff^2 + \mc{O}_\sigma ( \uaff^3 ) } = \uaff + \frac{1}{12} \uggs_{ \rho \rho } \uaff^3 + \mc{O}_\sigma ( \uaff^4 ) \text{.}
\end{equation}

Inverting the relation in \eqref{eq.fg_rhos_pre} results in the expansion
\begin{equation}
\label{eq.fg_rhos} \uaff = \rho - \frac{1}{12} \uggs_{ \rho \rho } \rho^3 + \mc{O} ( \rho^4 ) \text{,}
\end{equation}
while \eqref{eq.fg_X_asymp} and \eqref{eq.fg_rhos} imply
\begin{equation}
\label{eq.fg_rhoss} \urho = \rho - \frac{1}{4} \uggs_{ \rho \rho } \rho^3 + \mc{O} ( \rho^4 ) \text{.}
\end{equation}
Applying \eqref{eq.fg_rho_def} and \eqref{eq.fg_rhoss} to \eqref{eq.fg_g_pre} yields the FG-aAdS expansion
\begin{equation}
\label{eq.fg_final} g = \rho^{-2} d \rho^2 + \rho^{-2} \left[ \uggm_{ab} + \left( \uggs_{ab} + \frac{1}{2} \uggs_{ \rho \rho } \uggm_{ab} \right) \rho^2 + \mc{O} ( \rho^3 ) \right] d x^a d x^b \text{.}
\end{equation}

Finally, to complete the proof, we set the isometry $\Phi$ to be the map represented by the compositions of the coordinate transformations
\[
( \urho, \ux^a ) \mapsto ( \uaff, x^a ) \mapsto ( \rho, x^a ) \text{.}
\]

\section{Einstein-Vacuum Spacetimes} \label{sec.vacuum}

Let $( \mc{M}, g )$ denote an admissible FG-aAdS segment.
In this appendix, we briefly elaborate on the case in which $( \mc{M}, g )$ also satisfies the Einstein-vacuum equations\footnote{We choose this particular normalization of the cosmological constant in equation \eqref{eq.einstein_vacuum} so that the AdS (and Kerr-AdS) metric has the expansion \eqref{eq.aads_ads} at infinity.}
\begin{equation}
\label{eq.einstein_vacuum} \operatorname{Ric} [g] = - \frac{ n (n - 1) }{2} g \text{.}
\end{equation}
Recall (see \cite{grah_witt:conf_ads}, for instance) the following:

\begin{proposition} \label{thm.einstein_vacuum}
Suppose $n \geq 3$, and suppose also that $( \mc{M}, g )$ satisfies \eqref{eq.einstein_vacuum}.
Then, $-\gs$ is precisely the Schouten tensor associated with $( \mc{I}, \gm )$,
\begin{equation}
\label{eq.schouten} \isch_{a b} := \frac{1}{n - 2} \left[ R\mathring{i}c_{a b} - \frac{1}{ 2 (n - 1) } \mathring{R} \cdot \gm_{a b} \right] \text{,}
\end{equation}
where $R\mathring{i}c$ and $\mathring{R}$ denote the Ricci and scalar curvatures on $( \mc{I}, \gm )$.
\end{proposition}

\subsection{The Static Case}

We now further specialize to the case of static boundaries.
More specifically, we assume our conformal boundary has the form
\begin{equation}
\label{eq.aads_static} \mc{I} := ( 0, \pi T ) \times \mc{S} \text{,} \qquad \gm := - dt^2 + \gamma \text{,}
\end{equation}
where $( \mc{S}, \gamma )$ is an $(n - 1)$-dimensional Riemannian manifold, with $n \geq 3$.
Below, we show that in this setting, the pseudoconvexity property of Definition \ref{def.pcp} can be directly connected to positivity of the Ricci curvature $\mc{R}ic$ of $( \mc{S}, \gamma )$:

\begin{proposition} \label{thm.pseudoconvex_vacuum}
Suppose $n \geq 3$, and suppose $( \mc{M}, g )$ satisfies \eqref{eq.einstein_vacuum} and \eqref{eq.aads_static}.
If $\mc{R}ic \geq (n - 2) C$ uniformly on $\mc{S}$ for some $C > 0$, then:
\begin{enumerate}
\item For any $c \in ( 0, C )$, the following is uniformly positive definite:
\begin{equation}
\label{eq.pseudoconvex_vacuum} - \gs - c^2 dt^2 + \left[ \frac{1}{ 2 (n - 1) (n - 2) } \mc{R} - \frac{1}{2} ( C^2 + c^2 ) \right] \cdot \gm \text{,}
\end{equation}
where $\mc{R}$ denotes the scalar curvature of $\gamma$.

\item The pseudoconvexity property holds whenever $T > C^{-1}$, i.e., whenever $( \mc{M}, g )$ has a time span of greater than $C^{-1} \pi$.
\end{enumerate}

On the other hand, if $\mc{R}ic \leq 0$ at any point of $\mc{S}$, then the pseudoconvexity property, as expressed in Definition \ref{def.pcp}, fails to hold.\footnote{Here, we stress that the failure of the pseudoconvexity property only implies that the foliation \emph{defined by the particular $f$ in \eqref{eq.f}} needs not be pseudoconvex. In particular, this needs not imply that the local unique continuation from $\mc{I}$ must fail. See also the remark following the proof.}
\end{proposition}

\begin{proof}
Observe that direct computations yield
\begin{align}
\label{eq.schouten_static} \isch_{t A} &\equiv 0 \text{,} \qquad \isch_{t t} = \frac{1}{ 2 (n - 1) (n - 2) } \mc{R} \text{,} \\
\notag \isch_{A B} &= \frac{1}{n - 2} \mc{R}ic_{A B} - \frac{1}{ 2 (n - 1) (n - 2) } \mc{R} \cdot \gamma_{A B} \text{,}
\end{align}
from which we obtain
\begin{equation}
\label{eq.schouten_main} \mathring{P} + \frac{1}{ 2 (n - 1) (n - 2) } \mc{R} \cdot \gm = 0 dt^2 + \frac{1}{n - 2} \mc{R}ic_{A B} d x^A d x^B \text{,}
\end{equation}
from which \eqref{eq.pseudoconvex_vacuum} follows.
In particular, \eqref{eq.pseudoconvex_vacuum} implies that Definition \ref{def.pcp} is indeed satisfied whenever $\xi = 0$ and $T := c^{-1} > C^{-1}$.

Finally, if $\mc{R}ic \leq 0$ at some point $Q \in \mc{S}$, then we can see that at $Q$,
\[
- \gs - \zeta \gm = \frac{1}{n - 2} \mc{R}ic_{A B} d x^A d x^B - \left[ \frac{1}{ 2 (n - 1) (n - 2) } \mc{R} + \zeta \right] \cdot \gm 
\]
cannot be made positive-definite for any $\zeta$, hence the pseudoconvexity property is violated.
This completes the proof of the proposition.
\end{proof}

\begin{remark}
Like for Definition \ref{def.pcp}, the positivity of $\mc{R}ic$ is a gauge-dependent property which is not necessarily preserved by a conformal transformation of $\gm$; see the remark below Theorem \ref{theo:mti}.
\end{remark}

In particular, in the classical gravity setting $n = 3$, we have
\begin{equation}
\label{eq.curvature_3} \mc{R}ic_{A B} = \mc{K} \cdot \gamma_{A B} \text{,}
\end{equation}
where $\mc{K}$ is the Gauss curvature of $( \mc{S}, \gamma )$.
Proposition \ref{thm.pseudoconvex_vacuum} and \eqref{eq.curvature_3} imply:

\begin{corollary} \label{thm.pseudoconvex_vacuum_class}
Suppose $n = 3$ and $( \mc{M}, g )$ satisfies \eqref{eq.einstein_vacuum} and \eqref{eq.aads_static}.
\begin{enumerate}
\item If $\mc{K} \geq C > 0$ uniformly on $\mc{S}$, then the pseudoconvexity property holds whenever $T > C^{-1}$, that is, when $( \mc{M}, g )$ has time span greater than $C^{-1} \pi$.

\item If $\mc{K} \leq 0$ at any point of $\mc{S}$, then the pseudoconvexity property, as expressed in Definition \ref{def.pcp}, fails to hold.
\end{enumerate}
\end{corollary}

Note for AdS (as well as Kerr-AdS) spacetime, we have $\mc{K} \equiv 1$, hence the pseudoconvex condition holds for time intervals of length greater than $\pi$.

\section{A Modified Carleman Estimate} \label{sec.extra}

In this appendix, we establish, under additional assumptions, a technical refinement of the Carleman estimate of Theorem \ref{thm.carleman}.
This will be applied in upcoming works toward proving symmetry extension results; see \cite{hol_shao:symm_ext}.

More specifically, for ``non-borderline" values of $\kappa$, we can establish Carleman estimates applicable to wave equations with first order terms that:
\begin{itemize}
\item Are also small with respect to an $L^\infty$-norm, but

\item Decay slightly less than was allowed in Theorem \ref{thm.carleman}.
\end{itemize}
The precise statement of the refined estimate is given below:

\begin{theorem} \label{thm.carleman_extra}
Assume the hypotheses of Theorem \ref{thm.carleman}, and suppose also that
\begin{equation}
\label{eq.kappa_extra} \kappa > \kappa_0 \text{,} \qquad \kappa_0 := \begin{cases} \frac{ n - 2 }{2} + \sqrt{ \frac{ n^2 }{4} - \sigma } & \kappa \leq \frac{ n^2 - 1 }{4} \text{,} \\ \frac{ n - 1 }{2} & \kappa > \frac{ n^2 - 1 }{4} \text{.} \end{cases}
\end{equation}
In addition, let $\mc{X}$ be a vector field on $\mc{M}$ whose components satisfy
\begin{equation}
\label{eq.X_extra} | \mc{X}^\rho | \leq \varepsilon \rho^2 \text{,} \qquad | \mc{X}^t | \leq \varepsilon \rho^2 \text{,} \qquad | g ( \mc{X}, E_X ) | \leq \varepsilon \rho \text{,}
\end{equation}
where $E_X$ denotes any of the frame elements from Definition \ref{def.frame}, and where $\varepsilon > 0$ is sufficiently small with respect to $K$, $\kappa - \kappa_0$.

Then, there exist constants $C, \mc{C} > 0$, depending on $g$, $p$, $K$, and $\kappa - \kappa_0$, such that for any $\sigma \in \R$ and $\lambda \in [1 + \kappa, \infty)$, and for any $\phi \in \Gamma \ul{T}^0_l \mc{M}$ satisfying the conditions listed in Theorem \ref{thm.carleman}, the following inequality holds:
\begin{align}
\label{eq.carleman_extra} &\int_{ \Omega_{ f_0, \rho_0 } } f^{n - 2 - 2 \kappa} e^\frac{ - 2 \lambda f^p }{p} f^{-p} | ( \Box + \sigma + \nabla_{ \mc{X} } ) \phi |^2 \\
\notag &\qquad + \mc{C} \lambda ( \lambda^2 + | \sigma | ) \int_{ \{ \rho = \rho_0 \} } [ | \nabla_t ( \rho^{ - \kappa } \phi ) |^2 + | \nabla_\rho ( \rho^{ - \kappa } \phi ) |^2 + | \rho^{- \kappa - 1} \phi |^2 ] d \gm \\
\notag &\quad \geq C \lambda \int_{ \Omega_{ f_0, \rho_0 } } f^{n - 2 - 2 \kappa} e^\frac{ - 2 \lambda f^p }{p} ( \rho^4 | \nabla_t \phi |^2 + \rho^4 | \nabla_\rho \phi |^2 + \rho^2 | \nasla \phi |^2 ) \\
\notag &\quad \qquad + C \lambda \int_{ \Omega_{ f_0, \rho_0 } } f^{n - 2 - 2 \kappa} e^\frac{ - 2 \lambda f^p }{p} ( 1 + \lambda^2 f^{2p} ) | \phi |^2 \text{.}
\end{align}
\end{theorem}

\begin{remark}
With regards to unique continuation results, using the refined estimate \eqref{eq.carleman_extra}, one can remove the extra $\rho^p$ decay that was stipulated in \eqref{eq.uc_ads_rhs}.
However, this works only if the first-order coefficients in \eqref{eq.wave} are sufficiently small and if one chooses a non-optimal order of vanishing $\kappa$.
\end{remark}

\begin{remark}
Essentially, first-order terms in \eqref{eq.wave} with small coefficients that decay slightly less can be treated as a part of the operator on the left-hand side of \eqref{eq.carleman_extra}, rather than as an error term to be absorbed (as in the proof of Theorem \ref{thm.uc_ads}).
\end{remark}

As the proof of Theorem \ref{thm.carleman_extra} is largely similar to that of Theorem \ref{thm.carleman}, below we will only emphasize points where the two derivations differ.

\subsubsection{Proof of Theorem \ref{thm.carleman_extra}}

For convenience, we adopt the same notations as in the proof of Theorem \ref{thm.carleman}, in particular, \eqref{eq.F}-\eqref{eq.L}.
Define also the operator
\begin{equation}
\label{eq.L_extra} \tilde{\mc{L}} := e^{-F} ( \Box + \sigma + \nabla_{ \mc{X} } ) e^F = \mc{L} + \nabla_{ \mc{X} } + F' \cdot \mc{X} f \text{.}
\end{equation}
Computing as in the proof of Theorem \ref{thm.carleman}, we have
\begin{align}
\label{eq.carleman_extra_1} \tilde{\mc{L}} \psi^I S_{ \xi, \zeta } \psi_I &= \mc{L} \psi^I S_{ \xi, \zeta } \psi_I + \nabla_{ \mc{X} } \psi^I \nabla_S \psi_I + F' \mc{X} f \cdot \psi^I \nabla_S \psi_I \\
\notag &\qquad + h_{ \xi, \zeta } \psi^I \nabla_{ \mc{X} } \psi_I + h_{ \xi, \zeta } F' \mc{X} f \cdot | \psi |^2 \\
\notag &= \mc{L} \psi^I S_{ \xi, \zeta } \psi_I + I_1 + I_2 + I_3 + I_4 \text{.}
\end{align}
The first term on the right-hand side of \eqref{eq.carleman_extra_1} can be handled precisely as before, yielding \eqref{eq.conj}.
Combining all this results in the inequality
\begin{align}
\label{eq.conj_extra} \tilde{\mc{L}} \psi^I S_{\xi, \zeta} \psi_I &\geq [ ( 2 \kappa - n + 1 ) f^{n - 2} + 2 \lambda f^{n - 2 + p} + \lambda \cdot \mc{O}_0 (f^n) ] | \nabla_N \psi |^2 \\
\notag &\qquad + [ K f^{n - 2} \rho^2 + \mc{O}_0 ( f^{n - 1} \rho^2 ) ] ( | \nabla_V \psi |^2 + | \nasla \psi |^2 ) \\
\notag &\qquad + \left[ ( \kappa^2 - n \kappa + \sigma ) f^{n - 2} + \frac{2 - p}{2} \lambda ( 2 \kappa - n + p ) f^{n - 2 + p} \right] | \psi |^2 \\
\notag &\qquad + [ (1 - p) \lambda^2 f^{n - 2 + 2p} + \lambda^2 \cdot \mc{O}_0 ( f^n ) ] | \psi |^2 + \nabla^\beta ( P^Q_\beta + P^S_\beta ) \\
\notag &\qquad + I_1 + I_2 + I_3 + I_4 \text{.}
\end{align}

Now, \eqref{eq.kappa_extra} implies that there is some $0 \leq B < 1$, depending on $\kappa - \kappa_0$, with
\begin{equation}
\label{eq.B_extra} ( 1 - B ) ( 2 \kappa - n + 1 ) > 0 \text{,} \qquad ( \kappa^2 - n \kappa + \sigma ) + B ( 2 \kappa - n + 1 ) > 0 \text{.}
\end{equation}
We again apply \eqref{eq.hardy_ptwise} to the terms $f^{n - 2} | \nabla_N \psi |^2$ and $f^{n - 2 + p} | \nabla_N \psi |^2$ in the right-hand side of \eqref{eq.conj_extra}, except we also multiply the former term by $B$.
Defining
\begin{equation}
\label{eq.P_natural_extra} \tilde{P}^H_\beta := B ( 2 \kappa - n + 1 ) f^{n - 3} \nabla_\beta f \cdot | \psi |^2 + \frac{2 - p}{2} \lambda f^{n - 3 + p} \nabla_\beta f \cdot | \psi |^2 \text{,}
\end{equation}
then these computations yield the inequality
\begin{align}
\label{eq.hardy_extra} \tilde{\mc{L}} \psi^I S_{\xi, \zeta} \psi_I &\geq \nabla^\beta ( P^Q_\beta + P^S_\beta + \tilde{P}^H_\beta ) + ( 1 - B ) ( 2 \kappa - n + 1 ) f^{n - 2} | \nabla_N \psi |^2 \\
\notag &\qquad + [ \lambda f^{n - 2 + p} + \lambda \cdot \mc{O}_0 (f^n) ] | \nabla_N \psi |^2 \\
\notag &\qquad + [ K f^{n - 2} \rho^2 + \mc{O}_0 ( f^{n - 1} \rho^2 ) ] ( | \nabla_V \psi |^2 + | \nasla \psi |^2 ) \\
\notag &\qquad + [ ( \kappa^2 - n \kappa + \sigma ) + B ( 2 \kappa - n + 1 ) ] f^{n - 2} | \psi |^2 \\
\notag &\qquad + \frac{2 - p}{2} \, \lambda \left( 2 \kappa - n + 1 + \frac{p}{2} \right) f^{n - 2 + p} | \psi |^2 \\
\notag &\qquad + [ (1 - p) \lambda^2 f^{n - 2 + 2p} + \lambda^2 \cdot \mc{O}_0 ( f^n ) ] | \psi |^2 \\
\notag &\qquad + I_1 + I_2 + I_3 + I_4 \text{.}
\end{align}

Note that by \eqref{eq.X_extra}, along with computations in the proof of Lemma \ref{thm.phi_est}, we have
\begin{align}
\label{eq.X_extra_1} | \nabla_{ \mc{X} } \psi |^2 &\leq \varepsilon^2 ( \rho^4 | \nabla_\rho \psi |^2 + \rho^4 | \nabla_t \psi |^2 + \rho^2 | \nasla \psi |^2 ) \\
\notag &\lesssim \varepsilon^2 \rho^2 ( | \nabla_N \psi |^2 + | \nabla_V \psi |^2 + | \nasla \psi |^2 ) \text{.}
\end{align}
Thus, recalling \eqref{eq.S}, \eqref{eq.f_grad}, and \eqref{eq.X_extra_1}, we have
\begin{align}
\label{eq.I1_extra} | I_1 | &\leq [ f^{n - 2} + \mc{O}_0 ( f^n ) ] | \nabla_{ \mc{X} } \psi | | \nabla_N \psi | \\
\notag &\leq \varepsilon^2 D_1 f^{n - 2} \rho^2 ( | \nabla_N \psi |^2 + | \nabla_V \psi |^2 + | \nasla \psi |^2 ) \\
\notag &\qquad + \frac{ ( 1 - B ) ( 2 \kappa - n + 1 ) }{2} [ f^{ n - 2 } + \mc{O}_0 ( f^n ) ] | \nabla_N \psi |^2 \\
\notag &\leq \varepsilon^2 D_1 f^{n - 2} \rho^2 ( | \nabla_V \psi |^2 + | \nasla \psi |^2 ) \\
\notag &\qquad + \left[ \frac{ ( 1 - B ) ( 2 \kappa - n + 1 ) }{2} f^{ n - 2 } + \mc{O}_0 ( f^n ) \right] | \nabla_N \psi |^2 \text{.}
\end{align}
Similarly, recalling also \eqref{eq.h_est}, we obtain
\begin{align}
\label{eq.I3_extra} | I_3 | &\leq \mc{O}_0 ( f^n ) \cdot | \psi | | \nabla_{ \mc{X} } \psi | \\
\notag &\leq \varepsilon^2 D_1 f^{n - 2} \rho^2 ( | \nabla_N \psi | + | \nabla_V \psi |^2 + | \nasla \psi |^2 ) + \mc{O}_0 ( f^n ) \cdot | \psi |^2 \\
\notag &\leq \varepsilon^2 D_1 f^{n - 2} \rho^2 ( | \nabla_V \psi |^2 + | \nasla \psi |^2 ) + \mc{O}_0 ( f^n ) \cdot ( | \nabla_N \psi |^2 + | \psi |^2 ) \text{.}
\end{align}
In particular, as long as $\varepsilon$ is sufficiently small with respect to $K$ and $\kappa - \kappa_0$, then $I_1$ and $I_3$ can be absorbed into positive terms on the right-hand side of \eqref{eq.hardy_extra}.

For the remaining terms, we also note from \eqref{eq.f_deriv}, \eqref{eq.F_deriv}, and \eqref{eq.X_extra} that
\begin{align}
\label{eq.X_extra_2} | F' \mc{X} f | &= ( \kappa f^{-1} + \lambda f^{-1 + p} ) \cdot \varepsilon \rho^2 \cdot \mc{O}_0 ( f \rho^{-1} ) \\
\notag &= \mc{O}_0 ( \rho ) + \lambda \cdot \mc{O}_0 ( \rho f^p ) \text{.}
\end{align}
Using \eqref{eq.X_extra_2} and similar bounds as before, we then control
\begin{align}
\label{eq.I24_extra} | I_2 | &\leq \lambda \cdot \mc{O}_0 ( \rho ) \cdot | \psi | | \nabla_S \psi | \\
\notag &\lesssim \lambda \cdot \mc{O}_0 ( f^{ n - 1 } ) \cdot ( | \nabla_N \psi |^2 + | \psi |^2 ) \text{,} \\
\notag | I_4 | &\leq \lambda \cdot \mc{O}_0 ( f^n \rho ) \cdot | \psi |^2 \text{.}
\end{align}
Combining \eqref{eq.hardy_extra} with \eqref{eq.I1_extra}-\eqref{eq.I24_extra} and letting $\lambda$ be sufficiently large yields
\begin{align}
\label{eq.hardy_extra_1} \tilde{\mc{L}} \psi^I S_{\xi, \zeta} \psi_I &\geq \nabla^\beta ( P^Q_\beta + P^S_\beta + \tilde{P}^H_\beta ) + \left[ \frac{1}{2} \lambda f^{n - 2 + p} + \lambda \cdot \mc{O}_0 (f^n) \right] | \nabla_N \psi |^2 \\
\notag &\qquad + \left[ \frac{1}{2} K f^{n - 2} \rho^2 + \mc{O}_0 ( f^{n - 1} \rho^2 ) \right] ( | \nabla_V \psi |^2 + | \nasla \psi |^2 ) \\
\notag &\qquad + [ ( \kappa^2 - n \kappa + \sigma ) + B ( 2 \kappa - n + 1 ) ] f^{n - 2} | \psi |^2 \\
\notag &\qquad + \frac{2 - p}{4} \, \lambda \left( 2 \kappa - n + 1 + \frac{p}{2} \right) f^{n - 2 + p} | \psi |^2 \\
\notag &\qquad + [ (1 - p) \lambda^2 f^{n - 2 + 2p} + \lambda^2 \cdot \mc{O}_0 ( f^n ) ] | \psi |^2 \text{.}
\end{align}

Estimating as in the proof of Lemma \ref{thm.psi_est} results in the following variant of \eqref{eq.psi_est}:
\begin{align}
\label{eq.psi_est_extra} \lambda^{-1} f^{n - 2 - p} | \tilde{\mc{L}} \psi |^2 &\geq C \lambda f^{n - 2 + p} | \nabla_N \psi |^2 + C f^{n - 2} \rho^2 ( | \nabla_V \psi |^2 + | \nasla \psi |^2 ) \\
\notag &\qquad + C ( f^{n - 2} + \lambda f^{n - 2 + p} + \lambda^2 f^{n - 2 + 2p} ) | \psi |^2 \\
\notag &\qquad + \nabla^\beta ( P^Q + P^S + \tilde{P}^H )_\beta \text{,}
\end{align}
From this point, the proof of Theorem \ref{thm.carleman_extra} proceeds entirely analogously to that of Theorem \ref{thm.carleman}, hence we omit the details here.

\raggedright
\bibliographystyle{amsplain}
\bibliography{ads}

\end{document}